\documentclass[11pt,letterpaper]{article}
\usepackage[margin=0.9in]{geometry}
\usepackage[T1]{fontenc}
\usepackage[utf8]{inputenc}
\usepackage{mathptmx}
\usepackage{courier}
\usepackage{microtype}
\usepackage{graphicx,multirow,mathrsfs,xcolor,textcomp,booktabs,xspace,anyfontsize}
\usepackage{mathtools,amssymb,amsthm}
\usepackage[font=footnotesize]{subcaption}
\usepackage[ruled,noend,algo2e]{algorithm2e}
\usepackage{natbib}
\bibpunct[, ]{(}{)}{,}{a}{}{,}

\newtheorem{theorem}{Theorem}
\newtheorem{lemma}{Lemma}
\newtheorem{proposition}{Proposition}

\theoremstyle{definition}
\newtheorem{definition}{Definition}

\usepackage[hidelinks]{hyperref}
\usepackage{xurl}
\hypersetup{
  pdftitle={Load Balancing with Evolving Workloads in Large-Scale Language Model Serving},
  pdfauthor={Zixi Chen; Tianci Bu; Chendong Song; Xin Lu; Yinyu Ye; Zijie Zhou},
  pdfkeywords={load balancing; large language model serving; online assignment}
}
\makeatletter
\renewcommand{\proof}[1]{%
  \par\pushQED{\qed}\normalfont\topsep6pt\relax
  \trivlist\item[\hskip\labelsep\itshape #1]\ignorespaces}
\renewcommand{\endproof}{\popQED\endtrivlist\@endpefalse}
\makeatother
\newcommand{\Halmos}{\qedhere}

\newcommand{\IO}{\ensuremath{\mathrm{IO}}}
\newcommand{\JSQ}{\ensuremath{\mathrm{JSQ}}}

\newcommand{\E}{\mathbb E}
\newcommand{\Pp}{\mathbb P}
\newcommand{\1}{\mathbf 1}
\newcommand{\Geo}{\operatorname{Geo}}
\newcommand{\Var}{\operatorname{Var}}
\newcommand{\supp}{\operatorname{supp}}
\newcommand{\pos}[1]{\bigl(#1\bigr)_+}

\allowdisplaybreaks[2]
\title{\bfseries Load Balancing with Evolving Workloads\\
  in Large-Scale Language Model Serving}
\author{\parbox{\dimexpr\textwidth-2\tabcolsep\relax}{\centering
  \normalsize
  Zixi Chen\textsuperscript{3}\qquad
  Tianci Bu\textsuperscript{2,\(\dagger\)}\qquad
  Chendong Song\textsuperscript{1,\(\dagger\)}\\[3pt]
  Xin Lu\textsuperscript{2,*}\qquad
  Yinyu Ye\textsuperscript{1,4,*}\qquad
  Zijie Zhou\textsuperscript{1,*}\\[9pt]
  \small
  \textsuperscript{1}Department of Industrial Engineering and Decision Analytics, HKUST,\\
  Clear Water Bay, Hongkong, China.\\[3pt]
  \textsuperscript{2}Department of System Engineering, National University of Defense Technology,\\
  109 Deya Road, Changsha, 410073, Hunan, China.\\[3pt]
  \textsuperscript{3}School of Mathematical Sciences, Peking University,\\
  Yiheyuan Road, 100871, Beijing, China.\\[3pt]
  \textsuperscript{4}Department of Management Science and Engineering, Stanford University,\\
  450 Jane Stanford Way, Stanford, 94305, California, United States.\\[7pt]
  \footnotesize
  \textsuperscript{*}Corresponding authors: Xin Lu
  (\href{mailto:xin.lu.lab@outlook.com}{xin.lu.lab@outlook.com}),\\
  Yinyu Ye (\href{mailto:yyye@ust.hk}{yyye@ust.hk}),
  Zijie Zhou (\href{mailto:jerryzhou@ust.hk}{jerryzhou@ust.hk}).\\[3pt]
  Contributing authors: \href{mailto:chenzixi22@stu.pku.edu.cn}{chenzixi22@stu.pku.edu.cn};
  \href{mailto:btc010001@gmail.com}{btc010001@gmail.com};
  \href{mailto:songcd@ust.hk}{songcd@ust.hk}.\\[3pt]
  \textsuperscript{\(\dagger\)}Tianci Bu and Chendong Song contributed equally to this work.
}}

\date{}
\begin{document}
\maketitle
\begin{abstract}
Serving large language models requires assigning requests whose workloads grow during processing and whose completion times are unknown. With a synchronization barrier, lightly loaded workers wait for the slowest worker at every decode iteration. We model this loss through aggregate workload imbalance and formulate a joint admission and assignment problem with sticky requests and capacity limits. Short lookahead incorporates predicted completions of ongoing requests. For exact optimization without lookahead, we prove a comparison under independent geometric decode lengths, bounded independent nondegenerate prompt sizes, and batched releases satisfying a cumulative supply condition. The long-run expected imbalance improvement over first come, first served admission with request-count join-shortest-queue grows at least as the square root of the product of worker capacity and the logarithm of worker count, in the stated scaling regime. The comparison extends to capacity-aware round robin, random routing, and power of two choices. BalanceRoute approximations connect the formulation to practical routing. Experiments using up to 144 NPUs with proprietary and public traces show lower workload spread and higher throughput than the evaluated baselines. These empirical results complement the guarantee for exact optimization.

\end{abstract}
\begin{center}\small
\end{center}
\section{Introduction}\label{sec:intro}

Load balancing is central to the efficient operation of a large language model (LLM) serving platform. The platform receives users' prompts and generates responses using computing devices shared by many requests. Service begins by processing the prompt's tokens, which are units of text, and storing information needed to generate the response. This initial stage is called \emph{prefill}. The subsequent \emph{decode} stage generates the response one token at a time, with its total length unknown in advance. Systems process requests in batches to use the devices efficiently \citep{yu2022orca,kwon2023efficient}. They can also dedicate separate pools of devices to prefill and decode, allowing the two stages to be managed separately \citep{zhong2024distserve,patel2023splitwise}.

After prefill, a router assigns requests to decode workers as capacity becomes available. A worker may be one computing device or a group of devices that processes a batch of active requests together. To generate the next token, it accesses stored information about each request's prompt and previously generated tokens. This information is called the key-value (KV) cache. The number of tokens represented in the cache grows as the response develops. We measure a worker's current workload by the total number of these tokens across its active requests, a proxy for the work needed in its next decode iteration. Thus, two workers serving the same number of requests can have quite different workloads.

In large-scale deployments, these workers also share parts of the model. DeepSeek-V3, for example, has 671 billion parameters
\citep{deepseek2024v3}. Storing its main-model parameters at 16-bit
precision requires about 1.34~TB, whereas an NVIDIA H200 GPU has only 141~GB of memory
\citep{nvidia2023h200}. At this precision, a deployment that keeps the
parameters in device memory must spread them across several devices.
Each device holds parts of one model and exchange the computed results with other devices
to generate its response; an isolated part cannot perform the complete
model computation. In the configuration studied here, this cooperation
spans workers serving different request batches: each worker performs
local computation for its requests and also participates in shared
model computation with other workers. \textbf{Their results must be brought
together at repeated synchronization points}, where workers that finish
their local work sooner wait for the others \citep{vllm2026dp}. This
arrangement appears in DeepSeek-V3's reported deployment and in
NVIDIA's distributed inference software
\citep{deepseek2024v3,alvarez2025wideep}. An uneven assignment of requests
can consequently leave capacity idle across the synchronization group.

Our collaboration with Huawei revealed the cost of this imbalance in its production serving platform. Under the default routing policy, the workers carry markedly different token workloads, as shown in Figure~\ref{fig:production-imbalance}. The execution-time breakdown records the resulting wait for synchronization separately from the time spent communicating between devices. In the reported production measurement, \textbf{barrier waiting accounts for approximately 40\% of execution time}. This observation motivates improving the use of installed capacity through the assignment of requests.

\begin{figure}[t]
    \centering
    \begin{subfigure}[t]{0.55\linewidth}
        \centering
        \includegraphics[width=\linewidth]{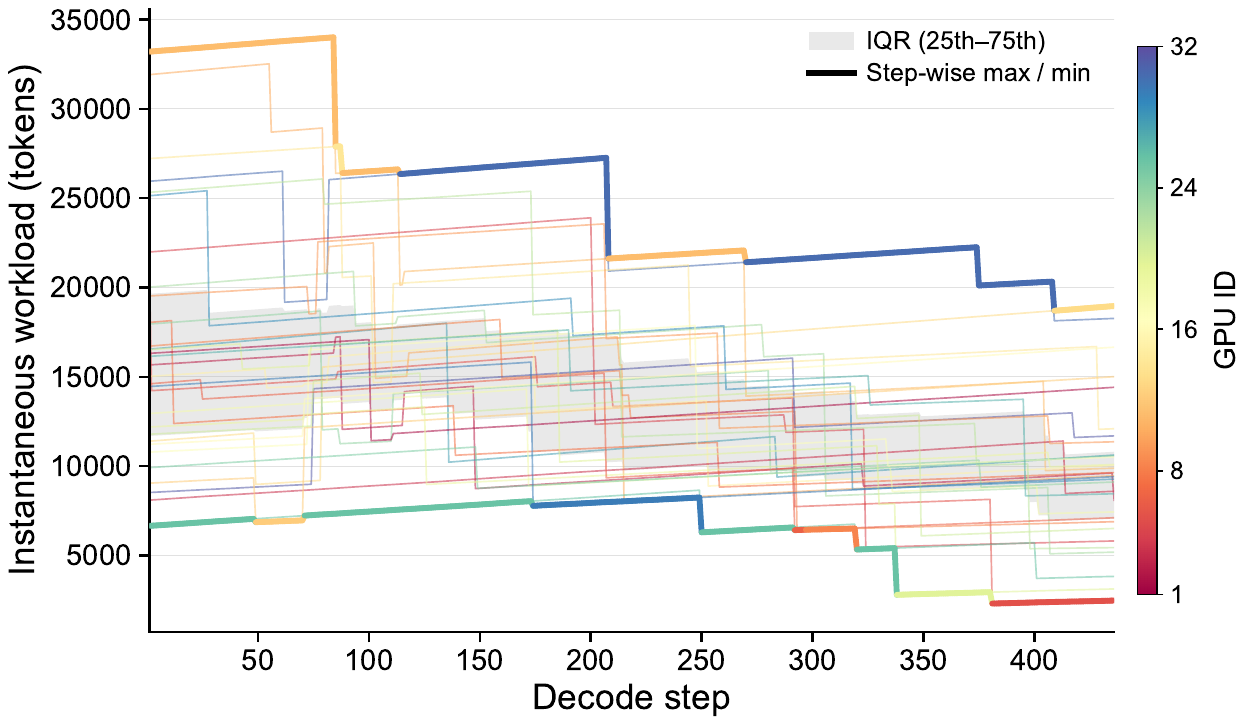}
        \caption{Token workloads across 32 devices}
        \label{fig:production-workloads}
    \end{subfigure}\hfill
    \begin{subfigure}[t]{0.42\linewidth}
        \centering
        \includegraphics[width=\linewidth]{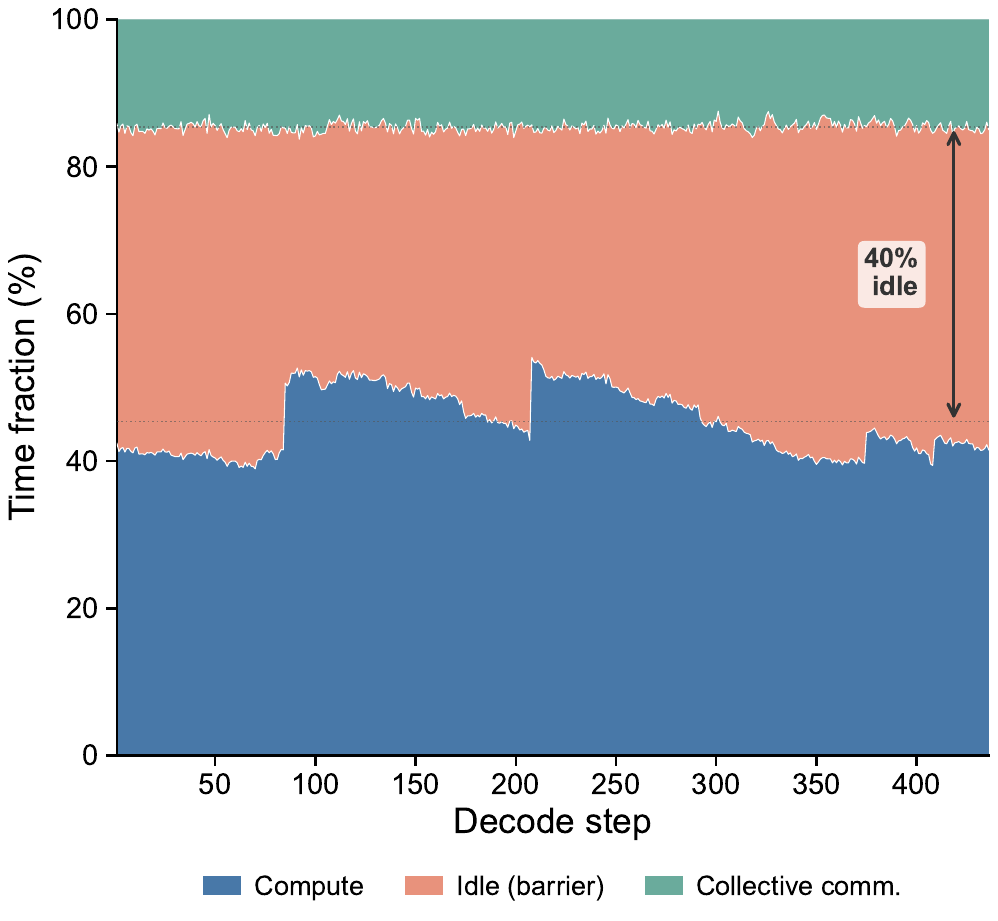}
        \caption{Execution-time breakdown}
        \label{fig:production-time}
    \end{subfigure}
    \caption{Production measurements from Huawei under their default admission and routing policy. The left panel shows instantaneous token workloads, with the stepwise maximum and minimum emphasized and the interquartile range shaded. The right panel separates computation, barrier waiting, and collective communication. The approximately 40\% waiting share refers to this production measurement.}
    \label{fig:production-imbalance}
\end{figure}

The difficulty is that workloads change after assignment. Each ongoing request adds a token in a decode iteration. When a request finishes, its accumulated cache is removed and a position becomes available on its worker. A worker with a large workload may therefore become lightly loaded after a single completion within one decode step. The completion time is unknown when the request enters service. Moving an ongoing request to another worker also requires transferring its cache, which has huge communication cost in real, so we study sticky assignments that remain fixed on one worker until completion. Balancing only the current number of requests overlooks both their different token workloads and these changes during service.

There is also an admission choice about which requests to assign. Each prefill worker in the prefill cluster releases decode-ready requests in batches, and they can wait in a common pool before entering decode cluster. When positions become available, the router must choose which waiting requests enter and which workers receive them. We refer to \textbf{the first decision as admission and the second as routing}. Unselected requests continue to wait. Every policy in our model fills as many positions as possible; the choice concerns how the waiting requests fit the workloads already in service. This combines a familiar capacity-allocation decision with an unusual service process: a request's workload grows during its stay and disappears at an uncertain completion time.

We develop an operational model  that connects workload differences to total waiting at the shared barrier. The model
leads to an imbalance metric, an assignment score, and a short lookahead
for prediction. These provide a common basis for optimization, analysis,
and implementation. Our contributions are as follows.

\begin{itemize}
    \item \textbf{A proposed imbalance metric for load balancing.}
    We propose aggregate imbalance as an optimization objective for
    synchronized LLM serving. The model makes this metric proportional
    to total barrier waiting and explains why joint request selection
    and routing require more than minimizing the largest workload.
    Production measurements and large-scale experiments support its
    operational value.
    \item \textbf{Short lookahead and a two-stage decision rule.}
An assignment that balances current workloads can become unbalanced
when ongoing requests finish. A common approach is to predict how many tokens a request will
generate before generation begins
\citep{zheng2023response,shahout2024don}. However, these predictions are
uncertain: the same prompt can produce responses of very different
lengths \citep{zheng2026uncertainty}. Moreover, even accurate length
forecasts do not determine future workloads fully, because future arrivals
remain unknown. Our model suggests focusing on a more accessible
prediction task: which ongoing requests will finish within a short
window, and approximately when. These completions change both worker
workloads and the capacity available for new requests in a short window. We therefore
propose a short lookahead mechanism and use its predictions in an
optimization-based two-stage decision rule. Both stages use the proposed imbalance
objective to balance workloads over the short prediction window.
As new requests arrive and workloads change, we repeat the decision
using updated information. Newly arrived requests provide further
opportunities to restore balance. The intuition is that we maintain balance
through these successive short-window decisions, which finally controls the long-time load balance.
    \item \textbf{An analytical comparison.}
    Under hidden geometric decode lengths and the stated arrival and
    scaling conditions, we establish a long-run expected-imbalance
    advantage for the exact current-workload policy over several
    practical baselines. The guaranteed relative improvement becomes
    larger as each worker can process larger batches of requests.
    The main difficulty is that the policy chooses requests using
    current workloads, so the requests in service cannot be treated
    as an independent sample from the arrival distribution.
    Our proof compares feasible assignments while tracking the two
    effects of each completion together: it removes workload and frees
    capacity for another request. We also obtain bounds on the expected
    ages of the selected requests, accounting for the selection decisions
    that brought them into service.
    \item \textbf{Implementation and system evidence.} The model led to BalanceRoute, which has been implemented in Huawei's real production system. Experiments with up to 144 NPUs evaluate fast approximations of the policy, including the contribution of short lookahead prediction.
\end{itemize}

The system implementation and empirical evaluation of BalanceRoute
are reported in our conference version \citep{BalanceRoute2026}. The present paper
develops the underlying operations methodology: how synchronization
and request dynamics lead to an optimization criterion, how that
criterion guides assignment and prediction, and what can be established
analytically about the resulting decisions. The model also provides
a basis for studying other admission and routing policies in synchronized
serving systems. Its theoretical guarantee concerns the exact
optimization policy; the practical algorithms use restricted searches,
with their performance supported by the reported experiments.

\subsection{Related Works}

Parallel-server routing is a well-established topic in operations research. Work on balanced routing, JSQ, and sampled choices studies delay, queue lengths, and the information needed for effective routing \citep{chen2012balanced,gupta2019slowdown,mukherjee2018universality,mendelson2025sparse}. The appropriate routing rule depends on the service model and the information available, as illustrated by \citet{whitt1986deciding}. Our setting retains a count-based capacity constraint, but its objective is the repeated waiting caused by differences in evolving workloads across synchronized workers. We also choose from a common waiting pool. The latter feature connects the problem to online allocation of reusable resources, where an assignment occupies capacity for an uncertain duration \citep{goyal2025reusable}. Here the assigned request's workload continues to grow during that occupancy, and its load contributes to a common barrier.

Recent work has brought LLM service decisions into mathematical scheduling and queueing models. \citet{jaillet2025online,wang2025llm,chen2025adaptively,kong2026geometry,feng2026competitive,feng2026general} study online scheduling with cache constraints and output-length predictions within a single computational worker. \citet{ao2025optimizing,ao2026service} develop fluid-guided scheduling under memory constraints, including policies for unknown output lengths. \citet{nie2026queueing} analyze stability with growing cache memory, and \citet{chen2026onlinelp} develop online optimization for routing with multiple service objectives. These studies establish the importance of accounting for memory and uncertain request lengths. Our focus is the synchronization loss across workers and the joint selection and assignment of requests to control that loss.

Systems research provides the batching, cache management, and separation of service stages on which our setting builds \citep{yu2022orca,kwon2023efficient,zhong2024distserve,patel2023splitwise}. Response-length prediction and learned ranking have also been used to guide sequence scheduling \citep{zheng2023response,fu2024efficient}. We use predictions to project the workloads that remain after near-term completions. The analysis of zero lookahead isolates the value of observable workloads, while the experiments examine what completion prediction adds in the serving system.

Section~\ref{sec:model} describes the service process, model, and comparison policies. Section~\ref{sec:load} develops the optimization method, and Section~\ref{sec:theory} gives the analytical comparison and its proof mechanism. Sections~\ref{sec:implementation} and~\ref{sec:experiments} present the practical algorithms and empirical results.

\section{Service Process and Assignment Model}\label{sec:model}

\subsection{Decode workloads and synchronization}

An LLM request begins with a prompt and generates a response one
token at a time. Processing a request involves two stages:
\emph{prefill} and \emph{decode}. During prefill, the system computes
the key and value (KV) information for each prompt token and stores
all KVs in a cache. During each decode step, it uses the cached KV
information from all earlier tokens to generate the next response token
\citep{kwon2023efficient}. The existing KV information is reused,
and the cache is extended by one token's KV information at each step.
A request with a longer history therefore requires more cached data
to be read. Its cache grows as generation continues and is released
when the request finishes.

These cache reads motivate our definition of decode workload.
Decoding is usually
\emph{bandwidth bound}: its speed is limited mainly by the rate to read KV data from the cache. For a fixed model and cache
format, each token contributes the same amount of KV information.
This motivates approximating local processing time by an affine
function of the number of cached tokens. We therefore define the
\emph{workload} of an active request at the beginning of a decode step as its
\textbf{total number of tokens retained in its cache}, including
its prompt tokens and the response tokens generated so far.
Counting requests alone misses this variation: two workers can
serve the same number of requests while reading very different
amounts of data. Figure~\ref{fig:dp} illustrates this difference
and the service structure.

The two stages have different resource needs. Prefill processes
the known prompt tokens together and requires substantial
computation, whereas decoding repeatedly reads and extends the
existing caches. This difference motivates a modern LLM serving structure, \emph{prefill-decode
disaggregation}, which runs the two stages on separate pools of
devices \citep{patel2023splitwise,zhong2024distserve}. Requests
become ready for decode when prefill has prepared their prompt
caches. A router then assigns them to decode workers. Here, a
\emph{worker} is a destination for a request batch; it may consist
of one device or several cooperating devices. Each worker maintains
the caches of its own active requests.

Different workers need not always advance together. To deploy some small-size LLMs, if each worker can hold an independent copy of the model, it can process its batch independently. However, in industrial large-scale deployment studied in this paper, a request assigned to one worker also uses model components stored on other workers since a single worker may not be able to store the whole LLM if the model size is very large. The workers exchange intermediate data to use these components and combine the results before proceeding. For example, a mixture-of-experts model distributes specialized components, called experts, across devices. Then, these exchanges impose common synchronization points \citep{vllm2026dp} since communication across workers requires all workers finishing their local processing. We study the assignment problem within one such synchronization group. DeepSeek-V3 reports a decode configuration with 80 request-processing groups of four GPUs, with expert computation distributed across all 320 GPUs \citep{deepseek2024v3}. NVIDIA also describes related distributed expert inference \citep{alvarez2025wideep}.

\begin{figure}[t]
    \centering
    \includegraphics[width=0.91\linewidth]{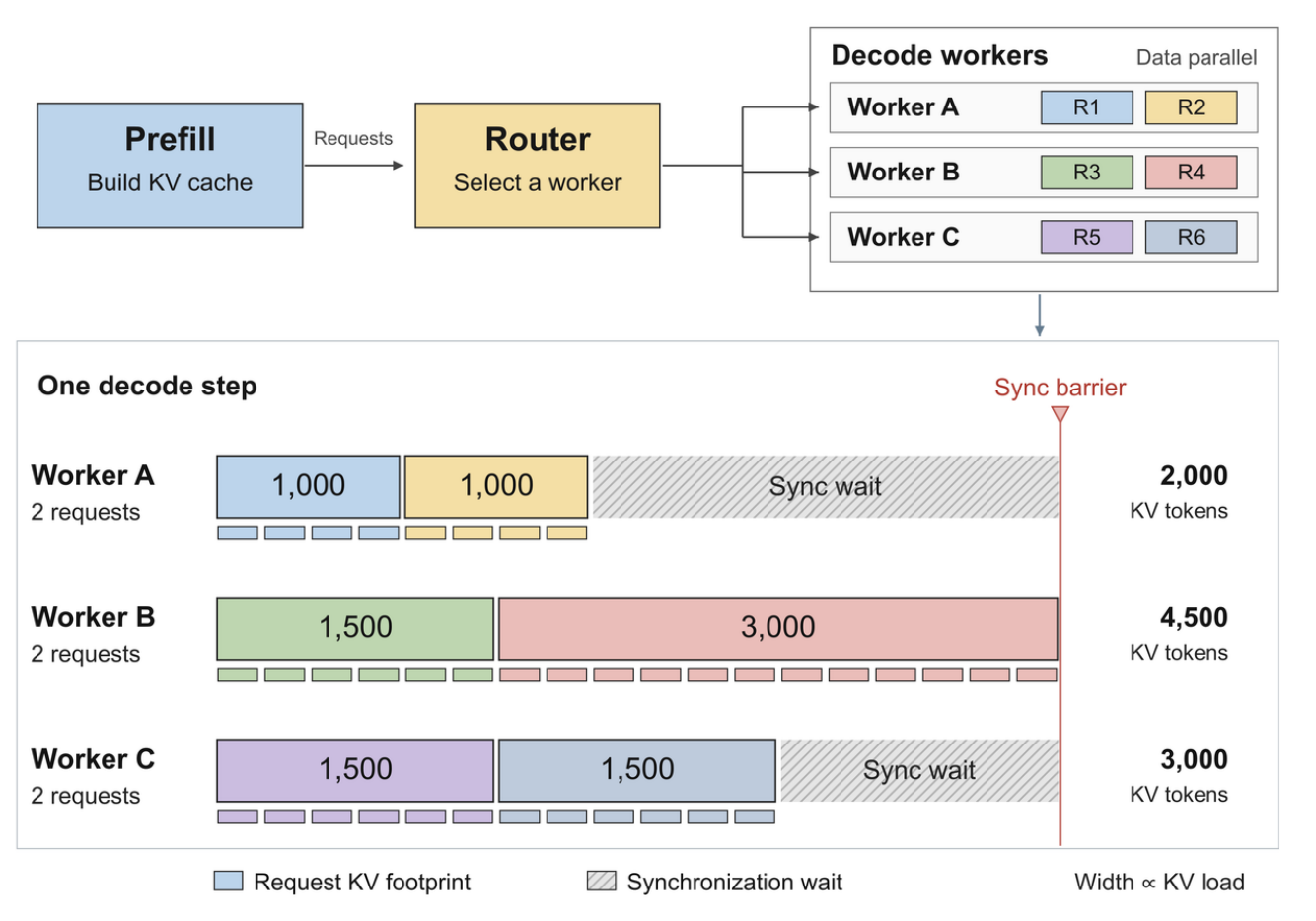}
    \caption{Request assignment and synchronized computation. Prefill prepares the prompt caches, and the router assigns ready requests to decode workers. Each worker in this illustration holds two requests, but their aggregate token workloads are 2,000, 4,500, and 3,000. The workers with smaller workloads wait for the largest one. The router can choose both which waiting requests to admit and where to place them. The example is schematic.}
    \label{fig:dp}
\end{figure}

\subsection{State, timing, and feasible assignments}

We collect the model definitions below; Table~\ref{tab:model-notation} provides a reference for the notation.

We consider $G$ identical workers in one synchronization group. Each worker can process at most $B$ active requests at once; we call these $B$ positions \emph{slots}. An admitted request stays on its assigned worker until completion. We allow neither migration nor preemption since the communication cost of carrying KVs from one worker to another is high.

We use discrete time indexed by synchronized decode iterations, $k=1,2,\ldots$. Admission takes place at the beginning of each iteration, before the active requests generate their next tokens. The duration of an iteration in seconds can vary with its workload. Accordingly, $k$ counts iterations rather than elapsed time. The system is empty at $k=0$.

Let $A(k)$ be the number of requests that have completed prefill and become ready for decode by the beginning of iteration $k$, with $A(0)=0$. The number newly released at iteration $k$ is $A(k)-A(k-1)$. Each release batch is finite. For the analysis, release times and prefill batch sizes are fixed external inputs by an adversary since they do not follow a regular pattern in real: an iteration can receive a huge batch or no new requests. These adversarial released prefill batches defines the arrival process for the router. The router observes released prefill batches as they occur and does not know the future schedule. We index requests by $i=1,2,\ldots$ in release order, using a fixed order within a simultaneous batch. Requests that have been released but not admitted remain in one global waiting queue.

Request $i$ has prompt size $S_i$ and decode length $O_i$. Its prompt size is observed when it becomes ready for decode; its total decode length is hidden until completion. If it enters service at iteration $x_i$, it is active during iterations $x_i,\ldots,x_i+O_i-1$. At the start of an active iteration, after admission and before token generation, its workload is
\begin{equation}
 W_i(k)=S_i+k-x_i,\qquad x_i\le k\le x_i+O_i-1.
 \label{eq:decode-profile}
\end{equation}
The quantity $k-x_i$ counts the decode iterations already completed by the request. The request finishes at the end of iteration $x_i+O_i-1$ and is removed before the next admission decision. Thus, workload $W_i(k)$ counts the retained token history.

At the beginning of iteration $k$, the system removes the requests that fully completed in the previous iteration. Each survivor's cache now includes the additional token from that iteration. Newly released requests join the waiting queue with all its prefill KVs. Immediately before admission, let $R_k$ be the waiting requests, $a_g(k)$ the aggregate workload of the survivors on worker $g$, and $c_g(k)$ its free slots. The router observes these quantities, the released prompt sizes, and the history of admissions and completions.

Every policy considered here admits the maximum feasible number
\begin{equation}
 U_k=\min\left\{|R_k|,\sum_{g=1}^G c_g(k)\right\}.
 \label{eq:maximal-admission}
\end{equation}
This guarantees that there is no cheating policy which leaves position idle to avoid workload imbalance. A feasible decision selects exactly $U_k$ requests from $R_k$, assigns each selected request to one worker, and assigns no more than $c_g(k)$ requests to worker $g$. If all waiting requests fit, the router admits them all. Otherwise, it chooses which requests enter service.

For a policy $\pi$, let $\mathcal A_g(k;\pi)$ be the active requests on worker $g$ after the whole admission batch at iteration $k$. Worker $g$'s resulting total workload in step $k$ is
\begin{equation}
 L_g(k;\pi)=\sum_{i\in\mathcal A_g(k;\pi)}(S_i+k-x_i).
 \label{eq:worker-workload}
\end{equation}

\begin{figure}[t]
    \centering
    \includegraphics[width=\linewidth]{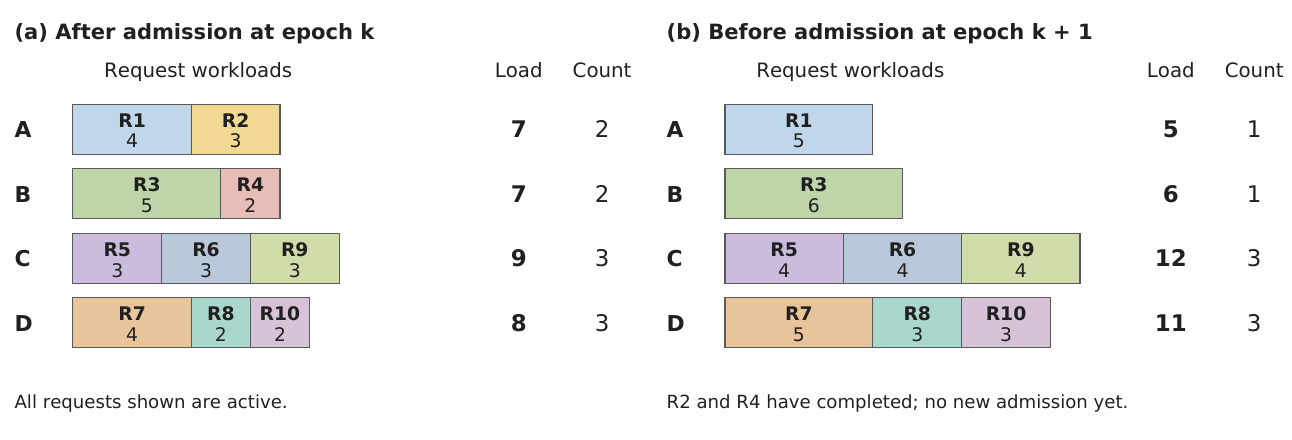}
    \caption{Workload changes over one decode iteration. The left state is after admission at iteration $k$; the right is before admission at iteration $k+1$. Requests R2 and R4 complete between the states, while each survivor's workload grows by one. The worker loads change from $(7,7,9,8)$ to $(5,6,12,11)$ before any further admission. Segment widths show individual request workloads. }
    \label{fig:decode-dynamics}
\end{figure}

Figure~\ref{fig:decode-dynamics} shows why balancing once is insufficient. A completion removes the request's accumulated workload and frees a slot for another request. Meanwhile, every survivor's workload grows. The next decision faces both different workloads and different available capacities. Therefore, even if at step $k$, $L_g(k;\pi) \approx L_{g'}(k;\pi)$ for workers $g$ and $g'$, $L_g(k+1;\pi)$ and $L_{g'}(k+1;\pi)$ may diverge significantly.

\subsection{The cost of workload imbalance}

Reading more cached history requires moving more data from memory. When this access is a major part of local processing, each additional unit of token workload adds approximately a common amount of processing time. We represent this dependence by the common affine approximation
\[
 t_{\mathrm{local}}(L)
   =t_{\mathrm{fixed}}+t_{\mathrm{token}}L,
 \qquad t_{\mathrm{fixed}}\ge0,\quad t_{\mathrm{token}}>0,
\]
where $t_{\mathrm{fixed}}$ is a fixed overhead and $t_{\mathrm{token}}$ is the time per unit of workload. 

Write $M(k;\pi)=\max_g L_g(k;\pi)$. The synchronization group waits for the worker with the largest local processing time. Under the affine approximation, worker $g$ waits $t_{\mathrm{token}}\bigl(M(k;\pi)-L_g(k;\pi)\bigr)$ for the heaviest worker. The common fixed overhead cancels from this difference. Summing across workers gives the aggregate imbalance
\begin{equation}
 I(k;\pi)=\sum_{g=1}^G\bigl(M(k;\pi)-L_g(k;\pi)\bigr)
         =G M(k;\pi)-\sum_{g=1}^G L_g(k;\pi).
 \label{eq:imbalancek}
\end{equation}
Total barrier waiting in this timing model is
$t_{\rm token} I(k;\pi)$. Aggregate imbalance therefore measures
the computing capacity lost while workers wait for one another.
For a given batch of active requests, reducing imbalance lowers
the largest worker workload and shortens the local processing
stage of a decode step. Faster synchronized decode steps shorten
the time between successive output tokens, reducing time per
output token (TPOT), and increase the throughput (number of output tokens
generated per second). 

For the analytical comparison, we use the expected average imbalance over the first $K$ decode iterations,
\begin{equation}
 \overline I_K^\pi=\frac1K\sum_{k=1}^K\mathbb E[I(k;\pi)].
 \label{eq:expected-average}
\end{equation}
The expectation is over the prompt sizes, decode lengths, and any random choices made by the policy. The long-run comparison takes $K$ to infinity. Table~\ref{tab:model-notation} collects the main notation.

\begin{table}[t]
    \centering
    \caption{Main notation for the assignment model}
    \label{tab:model-notation}
    \small
    \renewcommand{\arraystretch}{1.12}
    \begin{tabular}{p{0.27\linewidth}p{0.65\linewidth}}
        \toprule
        Notation & Meaning \\
        \midrule
        $G$, $B$ & Workers in the synchronization group; active-request capacity per worker. \\
        $k$ & Decode iteration, with admission before its computation. \\
        $S_i$, $O_i$ & Observed prompt size; hidden decode length of request $i$. \\
        $x_i$, $W_i(k)$ & Admission iteration; retained token workload of request $i$. \\
        $A(k)$ & Cumulative requests released to decode by iteration $k$. \\
        $R_k$ & Waiting requests available for admission at iteration $k$. \\
        $a_g(k)$, $c_g(k)$ & Surviving workload and free slots on worker $g$, before admission. \\
        $U_k$ & Maximum feasible number of admissions at iteration $k$. \\
        $\mathcal A_g(k;\pi)$, $L_g(k;\pi)$ & Active requests and their total workload after admission. \\
        $I(k;\pi)$, $\overline I_K^\pi$ & Aggregate imbalance; its expected average over $K$ iterations. \\
        \bottomrule
    \end{tabular}
\end{table}

\subsection{Baseline admission and routing rules}

In this subsection, we introduce several practical baselines used in real LLM serving. For admission, all baselines use first-come, first-served (FCFS): take the earliest waiting request and assign it to a worker with a free slot. Repeat until $U_k$ requests have entered service. The routing rules are listed in Table~\ref{tab:benchmark-destinations}. Join-shortest queue, round robin, random routing, and power of two choices are also supported by the vLLM router \citep{vllm-router}.

\begin{table}[t]
    \centering
    \caption{Destination rules used with global FCFS admission}
    \label{tab:benchmark-destinations}
    \begin{tabular}{p{0.35\linewidth}p{0.62\linewidth}}
        \toprule
        Routing policy \citep{vllm-router} & Choice among workers with a free slot \\
        \midrule
        Count-JSQ & Choose a worker with the fewest active requests. \\
        Round robin & Follow a fixed cyclic order, skip full workers, and advance the pointer after each assignment. \\
        Random & Choose uniformly among workers with a free slot. \\
        Power of two choices & Sample two workers uniformly and independently with replacement; choose the one with fewer active requests. \\
        \bottomrule
    \end{tabular}
\end{table}

The baseline used in main text combines FCFS admission with request-count join-shortest-queue routing. We call it \emph{FCFS+Count-JSQ} and denote it by $\JSQ$. Here, ``shortest'' means the fewest requests already in service. Active counts are updated after every assignment within a batch. In particular, in real practice, a request with a short prompt and one with a long prompt each contribute one to this count since the total number of tokens cannot be predicted accurately to differentiate these two prompts.

For all four baselines, worker, slot, and tie choices use counts, admission and completion history, or routing randomness; they do not inspect prompt sizes, token workloads, or unrevealed decode lengths. Fixed tie rules and any routing randomness are independent of the entire prompt-size and decode-length sequences. Each admitted request occupies a free slot. If only one worker has capacity, power of two choices samples that worker twice. These conventions make every baseline admit the maximum feasible number of requests.

The main theorem compares the exact workload-aware policy directly with $\JSQ$. Appendix~\ref{append:benchmark-policies} extends the comparison to the other FCFS routing rules defined above. These baselines have the same long-run expected average imbalance, although their initial assignments and transient workloads can differ. Our algorithm proposed in the next section changes both request selection and routing, and we measure the combined value.

\section{Short Lookahead Prediction for Admission and Routing Optimization}\label{sec:load}

In this section, based on the model dynamics introduced in Section~\ref{sec:model}, we motivate a novel short lookahead prediction mechanism and illustrate how it helps further admission and routing optimization.

\subsection{Short lookahead prediction mechanism}
\label{subsec:lookahead-motivation}

Recall that at iteration $k$, an active request has observed workload $W_i(k)$, which is defined in~\eqref{eq:decode-profile}. Let $r_i(k)$ be its remaining number of decode steps, including the current step. At iteration $k+h$, for $h=0,1,\ldots$, the contribution of this request is
\begin{equation}
 \bigl(W_i(k)+h\bigr)\mathbf 1\{r_i(k)>h\}.
 \label{eq:lookahead-request-profile}
\end{equation}
The first factor is known: each further step adds one token to the retained history. The indicator contains the uncertainty about completion. Once the request finishes, its whole workload disappears and its slot becomes available for another admission. 

To evaluate a current assignment, we first ask how the workloads
of ongoing requests will change over the next few iterations.
Equation~\eqref{eq:lookahead-request-profile} shows that, after
accounting for their known growth, the uncertain part is whether
each request remains in service in the near future. This suggests
comparing assignments over a window of $H\ge0$ future iterations,
using forecasts of the completions within it. For 
$h=0,\ldots,H$, all remaining lengths $r_i(k)>H$ give the same
workload path. We therefore forecast
which requests finish within the window and when. For requests that remain active beyond $H$, their precise
completion times do not affect their workloads within
the window.

Predicting the full remaining lengths of current requests
would still leave another part of the future workload
unknown: the contribution of requests admitted later.
For any policy $\pi$, the exact load at iteration $k+h$
can be written as
\begin{equation}
 L_g(k+h;\pi)
 =\underbrace{\sum_{i\in\mathcal A_g(k;\pi)}
   \bigl(W_i(k)+h\bigr)\mathbf 1\{r_i(k)>h\}}_{\text{requests active at }k}
   +N_g(k,h;\pi),
 \label{eq:lookahead-load-decomposition}
\end{equation}
where $N_g(k,h;\pi)$ is the workload of requests admitted after
iteration $k$ that remain active at $k+h$. Some may already be
waiting at $k$; others may not complete their prefills. Future
prefill batches can be irregular in both size and timing, and
the router cannot observe the releases counted by
$A(k+h)-A(k)$. These releases, subsequent admission decisions,
and the service progress of the later requests together
determine $N_g(k,h;\pi)$.

The two terms in Equation~\eqref{eq:lookahead-load-decomposition}
explain the tradeoff in choosing $H$. At $H=0$, admission and
routing use only current workloads. An assignment that balances
them now can leave a large gap when several requests on one
worker finish soon afterward. A positive window $H$ allows the
router to anticipate these departures. Increasing $H$ captures
more possible completions, but also extends the prediction into
a period affected by later admissions and greater forecast
uncertainty. We therefore choose a window that includes
departures capable of changing the current placement, while
leaving more distant admissions to decisions made with updated
information. For the deployment considered here, we suggest
$50\le H\le80$ decode iterations as a practical range. The realization of this short lookahead prediction mechanism will be introduced in Section \ref{sec:experiments}.

For $H\ge1$, we consider refreshing completion forecasts every
$H$ decode iterations, while updating the observed state and
making admission and routing decisions at every iteration.
Let $k$ denote a refresh iteration. The predictor forms workload
paths for the requests already in service through iteration
$k+H$. At an intermediate iteration $t=k+j$, where $0<j<H$,
the router reuses these forecasts and evaluates current
assignments over the remaining window $[t,k+H]$. A new refresh
at $k+H$ starts the next block.

For requests admitted between refreshes, since decode lengths are typically much longer than
the refresh period, we treat
these requests as remaining in service until the next refresh.
For a request admitted at iteration $x_i$ within the block,
we use
\[
 \widehat W_i(u)=S_i+u-x_i,
 \qquad x_i\le u\le k+H.
\]
This incorporates its known cache growth without immediately
calling the predictor for each new admission. An observed
completion removes the request immediately, even between
refreshes. Current workloads are always computed from the
observed active requests.

To connect these decisions with the long-run criterion,
partition time into successive blocks of $H$ iterations.
For $H\ge1$ and any positive integer $m$,
\begin{equation}
 \overline I_{mH}^{\pi}
 =\frac1m\sum_{j=0}^{m-1}
   \left[\frac1H\sum_{t=jH+1}^{(j+1)H}
         \mathbb E[I(t;\pi)]\right].
 \label{eq:lookahead-block-average}
\end{equation}
The bracketed term is the expected average imbalance in one
window. Sustained low window averages therefore give a low
long-run average. Within each window, the router
continues to adjust admissions and assignments as new
requests arrive and ongoing requests finish. The next
subsection formulates these decisions using the predicted
workloads.

\subsection{Optimizing predicted imbalance}
\label{subsec:lookahead-policy}

Fix an integer window length $H\ge0$. At iteration $k$, let $\widehat a_{g,h}(k)$ be the predicted workload at iteration $h$ from the requests already in service on worker $g$, before the current admissions. The observed value is $\widehat a_{g,0}(k)=a_g(k)$. Future values are formed from the request profile~\eqref{eq:lookahead-request-profile}, using completion forecasts.

For a waiting request $i$ admitted now, write $\widehat b_{i,h}(k)$ for its predicted workload, with $\widehat b_{i,0}(k)=S_i$ and $\widehat b_{i,h}(k)=S_i+h$, which uses the observed prompt size and predicts the request as remaining in service throughout the window. If it finishes earlier, its departure enters the next observed state. All prediction coefficients are nonnegative and are formed before the assignment using information available at iteration $k$.

Let $x_{ig}=1$ indicate that waiting request $i$ is admitted to worker $g$. The predicted total workload of $g$ at iteration $k+h$ under assignment $x$ is
\begin{equation}
 \widehat L_g(k+h;x)
 =\widehat a_{g,h}(k)
   +\sum_{i\in R_k}\widehat b_{i,h}(k)x_{ig}.
 \label{eq:lookahead-load}
\end{equation}
Denote by $\mathcal X_k$ the feasible assignments defined in Section~\ref{sec:model}: each request has at most one destination, worker $g$ receives at most $c_g(k)$ requests, and exactly $U_k$ requests are admitted. These constraints use the capacity available now.

We score an assignment by the aggregate imbalance along its predicted workload path:
\begin{equation}
 J_{k,H}(x)=\sum_{h=0}^H\gamma^h
 \left[G\max_g\widehat L_g(k+h;x)
       -\sum_{g=1}^G\widehat L_g(k+h;x)\right],
 \qquad \gamma\in(0,1].
 \label{eq:projected-objective}
\end{equation}
The choice $\gamma=1$ gives all iterations equal weight; smaller values place more weight on the closer iterations. For fixed prediction length $H$ and discount $\gamma$, the policy $\IO_H$ selects
\begin{equation}
 x^{\star}_{k,H}\in\operatorname*{arg\,min}_{x\in\mathcal X_k}J_{k,H}(x).
 \label{eq:lookahead-policy}
\end{equation}
Introducing one maximum-load variable $z_h$ for each iteration, with constraints $\widehat L_g(k+h;x)\le z_h$, gives an \emph{integer linear optimization}. Ties are resolved using information available to the policy.

The short lookahead prediction evaluates the current population in~\eqref{eq:lookahead-load-decomposition}, using forecasts for ongoing requests and the specified contributions of new admissions. It sets the later-admission term $N_g(k,h;\pi)$ to zero when comparing today's placements. The next iteration brings an updated state and a new admission decision. The policy solves this projected assignment problem exactly.

\begin{algorithm2e}[htbp]
 \caption{Short-lookahead assignment policy $\IO_H$}
 \label{alg:theory-balance-future}
 \KwIn{Surviving workloads $a_g(k)$, free slots $c_g(k)$, waiting prompts $\{S_i:i\in R_k\}$, window length $H$, discount $\gamma$}
 Compute $U_k$ from \eqref{eq:maximal-admission}\;
 Form the predicted workloads in \eqref{eq:lookahead-load} from the current observations and forecasts\;
 Solve \eqref{eq:lookahead-policy}, using a nonanticipative optimal tie rule\;
 Admit each selected request to its assigned worker and remove it from $R_k$\;
 At the next iteration, update the observed state and repeat\;
\end{algorithm2e}

\subsection{Algorithm without prediction: the special current-workload case}
\label{subsec:current-policy}

In this section, we describe a special case of $\IO_H$ when $H=0$.
At $H=0$, there is no future prediction and the predicted workloads are the observed surviving workloads plus the selected prompt sizes. The policy uses only this current information. We write $\IO=\IO_0$ for this member of the family. Its integer program is
\begin{equation}
\begin{aligned}
 \min_{x,z}\quad&
 Gz-\sum_{g=1}^G\left(a_g(k)+\sum_{i\in R_k}S_i x_{ig}\right)\\
 \text{subject to}\quad&
 a_g(k)+\sum_{i\in R_k}S_i x_{ig}\le z,
                                      &&g=1,\ldots,G,\\
 &\sum_{g=1}^Gx_{ig}\le1,             &&i\in R_k,\\
 &\sum_{i\in R_k}x_{ig}\le c_g(k),    &&g=1,\ldots,G,\\
 &\sum_{g=1}^G\sum_{i\in R_k}x_{ig}=U_k,\qquad
   x_{ig}\in\{0,1\}.
\end{aligned}
\label{eq:bf-ip}
\end{equation}
At an optimum, $z$ equals the largest resulting workload. The aggregate-load term accounts for the fact that selecting different requests changes the total workload admitted. A large prompt can fill a worker's gap or raise the common barrier, depending on its destination; the joint program weighs these effects together.

Section~\ref{sec:theory} analyzes the theoretical guarantee of $\IO_0$. This isolates the value of the imbalance objective and joint request selection and routing using observed information. In practice, the admission and routing decisions should be made with millisecond scale, and Section~\ref{sec:implementation} develops fast approximation algorithms of $\IO_H$ for both $H=0$ and $H>0$, using marginal scores and a two-stage search. 

\section{\texorpdfstring{Theory for $\IO_0$}{Theory for IO0}}\label{sec:theory}

We now analyze the theoretical guarantee of $\IO_0$, retaining the notation $\IO=\IO_0$. This case lets us study the value of the imbalance objective and joint selection and routing using observed workloads, without requiring a completion forecast. 

We assume the prefill size of each request follows some i.i.d. distribution, and the realization of the prefill size is observable. The decode length of each request is unknown and follows a geometric distribution\footnote{The rationale of decode along with several open-source dataset \citep{wang2023openchat,wang2024burstgpt,zheng2023lmsys,zhao2024wildchat} supports the fact that decode length has a geometric (discrete-exponential) pattern. See Figure \ref{fig:pd_dist}.}. Let $n=GB$ and $q=1-p$. Independently across requests,
\begin{equation}
 S_i\sim F,\qquad O_i\sim\operatorname{Geo}(p),\qquad
 \mathbb P(O_i=r)=p q^{r-1},\quad r\ge1,
 \label{eq:theory-marks}
\end{equation}
where $p\in(0,1)$, the entire prompt and lifetime sequences are independent, and $F$ is supported on a subset of $\{1,\ldots,s\}$ for an integer $s\ge2$. Write
\[
 \rho=1-\max_{t\in\{1,\ldots,s\}}\mathbb P(S_i=t)>0.
\]
Thus at least two prompt sizes occur with positive probability. 

We consider the overloaded arrival scenario, which is usually
a bottleneck in practice. An oblivious adversary chooses arrival time of each request, fixing the cumulative arrival
sequence $A(k)$ before the prompt sizes and decode lengths
are sampled. The router observes requests as they are arrived.
The model therefore combines adversarial arrival timing with
stochastic request workloads. This allows us to represent
irregular releases from prefill together with variation in
request sizes and response lengths.

Mathematically speaking, we make the following overloaded arrival assumption:
\begin{equation}
 \liminf_{k\to\infty}\frac{A(k)}{k}>\frac{np}{\rho}.
 \label{eq:theory-supply}
\end{equation}
This assumption ensures that, after accounting for all previous
admissions, the waiting pool eventually contains enough requests
outside every individual prompt class to support exchanges.
There is no lower bound on arrivals during a single iteration.
A schedule can include zero-arrival epochs, provided its
cumulative rate satisfies~\eqref{eq:theory-supply}.

Write
\[
 \sigma^2=\operatorname{Var}(S_i)+\frac{q}{p^2},
\]
which is the variance of the workload of one slot in the
stationary prompt-blind reference system. The next theorem quantifies the value of joint selection and
workload-aware routing by comparing $\IO$ with FCFS+Count-JSQ.
The bounds hold for every release schedule satisfying
\eqref{eq:theory-supply}.

\begin{theorem}
\label{thm:geometric}
Fix $B\ge1$ and $G\ge2$. Under~\eqref{eq:theory-marks},
for every arrival instance $A$ satisfying
\eqref{eq:theory-supply}, the following bounds hold when
the system starts empty. The exact policy $\IO$ defined
by~\eqref{eq:bf-ip} satisfies
\begin{equation}
 \limsup_{K\to\infty}\overline I_K^{\IO}
 \le \frac{(G-1)(s+\varepsilon_{B,G})}{p}.
 \label{eq:geometric-upper}
\end{equation}
If also
\begin{equation}
 B\ge B_0,\qquad
 \log G\le64\sigma^2\lambda_0^2 B,
 \label{eq:geometric-finite-condition}
\end{equation}
then FCFS+Count-JSQ satisfies
\begin{equation}
 \lim_{K\to\infty}\overline I_K^{\JSQ}
 \ge \frac{G\sigma\sqrt{B\log G}}{1024}.
 \label{eq:geometric-lower}
\end{equation}
Consequently,
\begin{equation}
 \liminf_{K\to\infty}
 \frac{\overline I_K^{\JSQ}}{\overline I_K^{\IO}}
 \ge
 \frac{p\sigma}{1024(s+\varepsilon_{B,G})}
 \frac{G}{G-1}\sqrt{B\log G}.
 \label{eq:geometric-ratio}
\end{equation}
A positive numerator divided by zero is interpreted as
$+\infty$.

More specifically, fix $p$, $F$, and $s$. Along any sequence
of $B\to\infty$ and $G\ge2$ satisfying $\log(GB)=o(B)$,
the finite conditions eventually hold and
$\varepsilon_{B,G}\to0$. For all sufficiently large
parameters,
\begin{equation}
 \inf_{A}
 \left\{
   \liminf_{K\to\infty}
   \frac{\overline I_K^{\JSQ}}{\overline I_K^{\IO}}
 \right\}
 \ge
 \frac{p\sigma}{2048s}
 \frac{G}{G-1}\sqrt{B\log G}.
 \label{eq:geometric-asymptotic}
\end{equation}
The infimum is over the fixed release schedules described
above. The time limit is taken for each fixed schedule and
fixed $B,G$ before the infimum over schedules and the scaling
of $B,G$.
\end{theorem}

The constants in the finite bounds are defined as follows.
Let $A_*\sim\operatorname{Geo}(p)-1$, with
$\mathbb P(A_*=a)=pq^a$ for $a\ge0$, and define
\begin{align}
 \mu_2
   &=\mathbb E[(s+1+A_*)^2]
     =\left(s+\frac1p\right)^2+\frac{q}{p^2},
     \notag\\
 \theta_0
   &=\min\left\{\frac1{s+1},-\frac12\log q\right\},
 &V
   &=pq\,\mathbb E\left[
       (s+1+A_*)^2e^{\theta_0(s+1+A_*)}
     \right],
     \notag\\
 \tau
   &=\min\{\theta_0,p/(2V)\},
 &\kappa
   &=p\tau/4,
     \notag\\
 \varepsilon_{B,G}
   &=\sqrt{
       8npq\mu_2\min\{1,4G e^{-\kappa B}\}
     }.
 \label{eq:geometric-noise-constants}
\end{align}
For the finite lower bound, define
\begin{equation}
 \lambda_0
   =\min\left\{
       \frac{\log2}{2s},
       \log\frac{1+q}{2q}
     \right\},
 \qquad
 R_{\mathrm{BE}}
   =\frac{4\sqrt2}{\sigma}\max\{s,96/p\},
 \qquad
 B_0
   =\left\lceil(16R_{\mathrm{BE}})^2\right\rceil.
 \label{eq:geometric-lower-constants}
\end{equation}
These constants are finite and positive.

The result separates active-count balance from workload balance. Once the system is full, $\JSQ$ has exactly $B$ active requests on every worker, yet differing prompt sizes and ages create workload fluctuations. Their maximum across workers produces the $\sqrt{B\log G}$ scale. In the stated scaling regime, the exact imbalance-aware policy keeps the long-run expected imbalance proportional to $G$, with a constant depending on the prompt range and completion hazard. The constants in the displayed lower bound are deliberately conservative; the result identifies a scale of improvement rather than predicting an experimental speedup. 

The complete proof appears in Appendix~\ref{append:geometric-proof}. Proposition~\ref{prop:benchmark-extension} extends the same lower bound to the stated capacity-aware global-FCFS routing policies. We next describe the proof in detail, because its main steps explain why joint selection, prompt diversity, and hidden completion noise can be handled without changing the policy.

\subsection{Proof mechanism and technical steps}\label{sec:proofsketch}

The upper bound uses prompt diversity to construct a
comparison assignment and obtain contraction of the
aggregate imbalance. Auxiliary slot clocks then control
completion fluctuations despite the optimizer's
selective admissions. For the lower bound, prompt-blind
admission and routing preserve an independent stationary
workload law across slots. The maximum across workers
then produces the additional $\sqrt{B\log G}$ factor. Next, we introduce the detailed proof sketch.

\subsubsection{Cumulative supply creates a usable waiting pool}

The first issue is that useful prompt diversity concerns the \emph{remaining} pool, not the arrival stream alone. The optimizer may preferentially admit particular prompt sizes, and a batched schedule need not provide new alternatives in every epoch.

Assign an independent Bernoulli completion coin of mean $p$ to every physical slot in every epoch, including empty slots. On occupied slots, these coins generate the geometric lifetimes. Let $W_{k-1}$ be the total number of successful ghost coins before epoch $k$ across all $n$ slots. Under any nonanticipative, maximally admitting policy, the cumulative number of previous admissions is at most $n+W_{k-1}$. The initial $n$ positions and subsequent completions account for every admission, while successful coins on empty slots only enlarge the bound.

For prompt class $t$, write $N_t(m)=|\{i\le m:S_i\ne t\}|$. Subtracting every previous admission, regardless of its prompt, gives the simultaneous bound
\begin{equation}
 |R_k\setminus\{i:S_i=t\}|
 \ge N_t(A(k))-n-W_{k-1}.
 \label{eq:sketch-supply}
\end{equation}
Since $\mathbb P(S_i\ne t)\ge\rho$, the rate condition provides a strictly positive linear gap between the right-hand arrival counts and the ghost completion count.

To make this quantitative, choose
\[
 \frac{np}{\rho}<\ell<\liminf_k A(k)/k,
 \qquad \Delta=\rho\ell-np>0.
\]
Eventually $A(k)\ge\ell k$. Hoeffding bounds show that, outside an event of probability at most
\[
 f_k\le d\exp\left\{-\frac{2\Delta^2 k}{9\ell}\right\}
          +\exp\left\{-\frac{2\Delta^2 k}{9n}\right\},
 \qquad d=|\operatorname{supp}(F)|,
\]
all $N_t(A(k))$ are at least $(np+2\Delta/3)k$ and $W_{k-1}$ is at most $(np+\Delta/3)k$. For sufficiently large $k$, \eqref{eq:sketch-supply} then leaves at least $n$ waiting requests outside every prompt class. Denote this event by $H_k$. It supplies enough total requests to fill all vacancies as well as enough alternatives for the exchanges below.

\subsubsection{A comparator proves drift for the original objective}

Fix an admission epoch at which the pool contains at least $C=\sum_g c_g$ requests outside every prompt class. Consider an auxiliary full assignment chosen to minimize lexicographically
\[
 \bigl(D,N_{\max},I\bigr),\qquad
 D=\max_g L_g-\min_g L_g,\qquad
 N_{\max}=|\{g:L_g=\max_h L_h\}|.
\]
The feasible pool is finite, so such an assignment exists. It is a comparator used to analyze \eqref{eq:bf-ip}; it is not the tie rule of the actual policy. The pool condition gives at least $C$ waiting requests, so maximal admission fills all $c_g$ vacancies on every worker. The optimizer and comparator therefore use the same feasible assignment set. Write $L^c$ for the comparator's load vector. Because $\IO$ minimizes $I$ over that set,
\begin{equation}
 I^{\IO}_{\mathrm{new}}\le I(L^c)
 :=G\max_g L_g^c-\sum_gL_g^c.
 \label{eq:sketch-comparator-transfer}
\end{equation}
It is possible that $D(L^{\IO})>D(L^c)$ even though \eqref{eq:sketch-comparator-transfer} holds. All ordering statements in the following exchange argument concern the comparator's selected prompts. They need not hold for the prompts selected by $\IO$. In that argument, write $L=L^c$ for brevity.

Let $r$ be a maximum-load worker of the comparator and let $g$ have deficit $L_r-L_g>s$. If both workers receive requests, exchanging a strictly larger selected prompt on $r$ with a smaller selected prompt on $g$ leaves both loads strictly inside their previous pair endpoints. This cannot increase $D$, and, if $D$ ties, reduces the number of maximum-load workers. Thus every prompt selected for $r$ is weakly smaller than every prompt selected for $g$.

Weak ordering alone is insufficient. If both workers select the same boundary prompt $t$, the pool assumption ensures an unselected request with prompt different from $t$: at least $C$ pool requests are outside class $t$, whereas the full assignment uses only $C$ requests and selects at least one $t$. A smaller unselected prompt can replace a $t$ on the maximum-load worker, improving the first or second lexicographic criterion. A larger unselected prompt can replace a $t$ on worker $g$ without reaching the old maximum; it cannot worsen $D$, and improves $I$ if the earlier criteria tie. Both contradict comparator optimality.

Because prompt sizes are integers, the resulting strict separation yields a threshold $\theta$ such that
\begin{equation}
 \sum_{i\in T_r}S_i-\sum_{i\in T_g}S_i
 \le \theta c_r-(\theta+1)c_g,
 \qquad 0\le\theta\le s,
 \label{eq:sketch-separation}
\end{equation}
where $T_r,T_g$ are the selected sets. Empty selected sets require no positive-capacity assumption: use $\theta=0$ if $c_r=0$, and $\theta=s$ if $c_r>0,c_g=0$. These boundary choices will be important because a worker may have no completions in an epoch.

\subsubsection{One noise process combines survival and vacant capacity}

Suppose the system was full after the previous admission. Write $w_{gj}$ for the previous workload in slot $j$ on worker $g$, and $\xi_{gj}$ for its current completion coin. The surviving pre-admission load and vacancy count satisfy
\[
 a_g=q(L_g^{\mathrm{old}}+B)+\zeta_g,\qquad
 \zeta_g=-\sum_{j=1}^B(\xi_{gj}-p)(w_{gj}+1),\qquad
 c_g=\sum_{j=1}^B\xi_{gj}.
\]
Treating survival and vacancy noise separately would discard their dependence through the same completion coins. Instead define the affine endpoint process
\begin{equation}
 E_g(t)=\zeta_g+t(c_g-Bp)
       =\sum_{j=1}^B(\xi_{gj}-p)(t-w_{gj}-1),
 \qquad 0\le t\le s+1.
 \label{eq:sketch-unified-noise}
\end{equation}
Combining \eqref{eq:sketch-separation} with the surviving loads gives, whenever the comparator deficit exceeds $s$,
\[
 L_r-L_g
 \le q(L_r^{\mathrm{old}}-L_g^{\mathrm{old}})
       -Bp+E_r(\theta)-E_g(\theta+1).
\]
Both arguments lie in $[0,s+1]$. Since $E_g$ is affine, its extreme values occur at the endpoints. Let $u^*$ be the range of the $2G$ numbers $\{E_g(0),E_g(s+1):g=1,\ldots,G\}$. Every large deficit is therefore at most
$q(M^{\mathrm{old}}-L_g^{\mathrm{old}})+u^*-Bp$; every remaining deficit is at most $s$.

Summing the comparator's deficits over $g\ne r$, and then using the objective optimality of $\IO$, yields the direct scalar inequality
\begin{equation}
 I^{\IO}_{\mathrm{new}}
 \le q I^{\IO}_{\mathrm{old}}
      +(G-1)\left[s+(u^*-Bp)_+\right].
 \label{eq:sketch-drift}
\end{equation}
The omitted worker's old deficit is nonnegative, so the old-gap sum is bounded by the complete old imbalance. This step transfers a comparator construction into a bound on the actual joint objective. 

\subsubsection{Ghost ages control adaptive coefficients}

The remaining obstacle is that $\IO$ uses observed prompts and ages, so the workloads multiplying the coins are adaptively selected. We need concentration without claiming that these workloads are independent. In this step, the conditioning history contains released prompts, observed decisions and completions, independent policy randomness, and past auxiliary coins. It excludes active requests' unrevealed total decode lengths and all future coins. Geometric memorylessness and prompt-lifetime independence then give independent fresh Bernoulli completion hazards. Appendix~\ref{append:geometric-proof} specifies this filtration before using the conditional moment bound.

For each slot, define a ghost age by
\[
 X(0)=0,\qquad X(k)=(1-\xi(k))(X(k-1)+1).
\]
An occupied request's actual age is at most its ghost age, including when the slot has been empty. At each fixed time the ghost ages are independent across slots, because they use independent slot coin histories. Each is stochastically bounded by $A_*$. For the global endpoint process, define $w=0$ on an empty slot. Then every coefficient in \eqref{eq:sketch-unified-noise} has absolute value at most $Z=s+1+X(k-1)$, whether the system is full or not.

Conditional on the past, the current completion coins are independent. Centering and the Taylor remainder give a one-slot moment generating function bound
\[
 \mathbb E\!\left[e^{\theta(\xi-p)a}\mid\text{past}\right]
 \le 1+\frac{pq\theta^2 Z^2e^{|\theta|Z}}2,
 \qquad |a|\le Z.
\]
We first multiply these conditional bounds, and only then take the unconditional expectation. The factors on the right depend solely on independent ghost ages. Using their stationary domination gives
\[
 \mathbb E[e^{\theta E_g(t)}]\le
 \exp(VB\theta^2/2),\qquad
 |\theta|\le\theta_0,\quad t\in\{0,s+1\}.
\]
This is an unconditional bound. Conditioning the ghost ages on saturation or pool diversity would invalidate the independence used in the factorization.

Chernoff's inequality with $\tau$ from \eqref{eq:geometric-noise-constants}, followed by a union bound over the endpoints, gives
\[
 \mathbb P(u^*>Bp)\le\min\{1,4G e^{-\kappa B}\}.
\]
Moreover, $u^{*2}$ is at most four times the sum of the squared endpoints. Conditional variances and ghost-age domination imply
\[
 \mathbb E[u^{*2}]\le8n pq\mu_2.
\]
Cauchy-Schwarz therefore bounds the expected excess by
$\mathbb E[(u^*-Bp)_+]\le\varepsilon_{B,G}$.

The same ghost construction gives a uniform second-moment bound on actual imbalance for every policy:
\[
 \sup_k\mathbb E[I(k)^2]
 \le [(G-1)B]^2\left(2s^2+\frac{4n}{p^2}\right).
\]
Thus failure of $H_{k-1}$ or $H_k$ contributes an expected error $b_k\to0$, by Cauchy-Schwarz and $f_k\to0$. Taking expectations in \eqref{eq:sketch-drift} gives
\begin{equation}
 e_k\le q e_{k-1}+(G-1)(s+\varepsilon_{B,G})+b_k,
 \qquad e_k=\mathbb E[I(k;\IO)].
 \label{eq:sketch-upper}
\end{equation}
Summation and division by $K$ yield \eqref{eq:geometric-upper}, since $1-q=p$ and the time average of $b_k$ vanishes. The resulting upper bound holds for every finite $B$; no separate calibration of completion counts is required.

\subsubsection{Prompt-blind routing has an independent slot law}

For $\JSQ$, global FCFS fixes the admitted stream order, while routing depends on active counts, slot eligibility, and independent tie randomness. Conditional on all completion coins and routing randomness, the map from admitted stream positions to physical slots is independent of prompt marks. The independent input stream consequently supplies independent $F$-distributed reset prompts to the slots. Auxiliary independent prompts can fill unused reset coordinates without affecting the real system.

Introduce a virtual system that keeps every slot occupied. Its workload follows
\[
 Y(k)=(1-\xi(k))(Y(k-1)+1)+\xi(k)S^{\mathrm{reset}}.
\]
It has the product stationary law
\[
 U=S+A_*,\qquad
 \operatorname{Var}(U)=\sigma^2.
\]
In particular, the stationary worker loads are independent sums of $B$ independent copies of $U$.

To connect this system with a real system that starts empty and may temporarily lack supply, use a recent deterministic time window. If every post-admission state in epochs $k-r,\ldots,k$ is full and each slot has at least one reset in epochs $k-r+1,\ldots,k$, the real and virtual systems agree at epoch $k$. The mismatch probability is at most
\[
 \sum_{j=k-r}^k f_j+nq^r.
\]
Taking $r=\lfloor k/2\rfloor$ makes this expression vanish. The uniform moment bounds then transfer expected imbalance to the product stationary law. This argument includes the final admission epoch and does not condition on a random last-shortage time.

Once saturated, count balancing alone offers no further discrimination: each worker has $B$ active requests. The workload variation in the product law remains. Appendix~\ref{append:benchmark-policies} shows that the same prompt-blind refresh argument covers the stated FCFS routing alternatives.

\subsubsection{Exponential tilting gives the maximum-load scale}

It remains to lower-bound the maximum of $G$ independent worker sums of $B$ copies of $U$. A normal approximation at a fixed tail level is insufficient when $G$ grows, because the relevant tail probability also shrinks. We instead tilt the law of $U$ exponentially.

Under tilt parameter $0\le\lambda\le\lambda_0$, the prompt distribution is reweighted by $e^{\lambda S}$, and the age becomes geometric with parameters $q_\lambda=qe^\lambda$ and $p_\lambda=1-q_\lambda$. The choice of $\lambda_0$ bounds the prompt reweighting and keeps $p_\lambda\ge p/2$. It follows that the tilted variance lies between $\sigma^2/2$ and $8\sigma^2$. A direct third-moment bound gives a standardized third absolute moment at most $R_{\mathrm{BE}}$. Berry-Esseen then has error at most $1/16$ for $B\ge B_0$.

Set $z=\sqrt{\log G}/16$ and $\lambda=2z/(\sigma\sqrt B)$. The finite condition \eqref{eq:geometric-finite-condition} ensures $\lambda\le\lambda_0$. A central interval under the tilted law has a fixed positive probability. Changing measure back to the original law, using the variance bounds to control its exponent, gives
\[
 \mathbb P\{L_g\ge B\mathbb E[U]+\sigma\sqrt B\,z\}
 >\frac15 e^{-32z^2-4\sqrt2 z}
 \ge\frac{e^{-1/4}}5 G^{-1/4}.
\]
The untilted Berry-Esseen bound also gives
$\mathbb P\{L_1\le B\mathbb E[U]\}\ge1/4$. Independence across workers implies that the probability that one of workers $2,\ldots,G$ exceeds the high threshold is bounded below by a positive constant uniformly over $G\ge2$.

A final step is needed to turn this tail event into an expected maximum. One cannot simply ignore negative values of a centered maximum. Let $M_{-1}=\max_{2\le g\le G}L_g$. Pointwise,
\[
 \max_gL_g=L_1+(M_{-1}-L_1)_+.
\]
Taking expectations and using $\mathbb E[L_1]=B\mathbb E[U]$ gives the exact identity
\[
 \mathbb E[\max_g L_g]-B\mathbb E[U]
   =\mathbb E[(M_{-1}-L_1)_+].
\]
The right side is nonnegative. On the independent high and low events above, this quantity is at least $\sigma\sqrt B\,z$. Keeping the explicit constants and multiplying by $G$ proves \eqref{eq:geometric-lower}, because
$\mathbb E[I]=G(\mathbb E[\max_g L_g]-B\mathbb E[U])$.

\subsubsection{Combining bounds and respecting the order of limits}

For fixed $B,G$ and a fixed admissible count sequence, the benchmark average converges to its stationary expectation, while \eqref{eq:geometric-upper} controls the optimizer's limsup. Their quotient gives \eqref{eq:geometric-ratio}. With $p,F,s$ fixed, the squared error satisfies
\[
 \varepsilon_{B,G}^2
 \le32nGpq\mu_2e^{-\kappa B}.
\]
The condition $\log(GB)=o(B)$ therefore makes $\varepsilon_{B,G}\to0$ and ensures the finite lower-bound conditions eventually hold. Replacing $s+\varepsilon_{B,G}$ by $2s$ then gives \eqref{eq:geometric-asymptotic}.

\section{A Two-Stage Approximation Method for Fast Routing}
\label{sec:implementation}

Section~\ref{sec:load} defines an assignment policy for each lookahead length $H$. Solving its joint integer optimization at every decode step can be expensive when many requests are waiting. BalanceRoute uses the structure of the projected imbalance objective to search for an assignment quickly. BR-0 is the approximation of $\IO_0$, which scores admissions using current workloads without any prediction; BR-H is the approximation of $\IO_H$, which uses short lookahead predicted workloads from ongoing requests. Both combine single-request admissions with a bounded search over request subsets. We first derive their scores and then describe these two allocation stages.

\subsection{BR-0: Routing from current loads}
\label{sec:br0}
\label{subsec:br0-marginal}

Fix a routing step and write $\ell_g$ for the load on worker $g$ before a candidate admission, including any requests already assigned during this step. Let $M=\max_g\ell_g$ and $m_g=M-\ell_g$. The margin $m_g$ is the additional load that worker $g$ can absorb before increasing the maximum. For a feasible subset $Q$ of waiting requests, its admission load is $b(Q)=\sum_{i\in Q}s_i$, where $s_i$ is the observed prefill size $S_i$.

\begin{proposition}[Marginal imbalance of an admission]
\label{prop:routing-marginal}
Admitting $Q$ to worker $g$ changes the current imbalance by
\begin{equation}
\label{eq:routing-marginal}
\Delta I_g(Q)=G\bigl(b(Q)-m_g\bigr)_+-b(Q).
\end{equation}
More generally, if disjoint subsets $Q_g$ are admitted simultaneously, with $b_g=b(Q_g)$, their total imbalance reduction is
\begin{equation}
\label{eq:routing-joint-gain}
I_{\mathrm{before}}-I_{\mathrm{after}}
=\sum_{g=1}^G b_g-G\max_{g\in[G]}(b_g-m_g)_+.
\end{equation}
\end{proposition}

\begin{proof}{Proof.}
For one admission, the total load increases by $b(Q)$ and the maximum becomes $M+(b(Q)-m_g)_+$. Substituting these quantities into $I=GM-\sum_g\ell_g$ gives~\eqref{eq:routing-marginal}. With simultaneous admissions the total load increases by $\sum_gb_g$ and the maximum becomes $M+\max_g(b_g-m_g)_+$, which gives~\eqref{eq:routing-joint-gain}.\Halmos
\end{proof}

BR-0 scores a candidate admission by its exact reduction in current imbalance:
\begin{equation}
\label{eq:br0-Fscore}
F_g^0(Q)=-\Delta I_g(Q)
=b(Q)-G\bigl(b(Q)-m_g\bigr)_+.
\end{equation}
Below the margin, each additional unit of load removes one unit of imbalance. Above the margin, each additional unit increases imbalance by $G-1$. This asymmetry explains why a routing rule should consider the maximum load together with available capacity. A worker with many free slots need not have a large load margin.

BR-0 builds an assignment by applying this score successively and updating the loads after each admission. The realized marginal gains then sum to the reduction in current imbalance over the whole routing step. The worker priorities and candidate windows described below determine which admissions are considered. Together they give a fast search heuristic for the $H=0$ program~\eqref{eq:bf-ip}.

\subsection{BR-H: Routing with a short projection}
\label{sec:brh}

BR-H scores candidate admissions across the short lookahead predicted workloads introduced in Section~\ref{subsec:lookahead-policy}. Its forecasts concern ongoing requests; the prediction interface and estimators appear in Appendix~\ref{subsec:brh-prediction}.

\subsubsection{Horizon F-score}
\label{subsec:brh-fscore}

Let $\widehat\ell_{g,h}$ be the projected load on worker $g$ at step $h\in\{0,\ldots,H\}$, including requests already assigned during the current routing step, with $\widehat\ell_{g,0}=\ell_g$. Define $M_h=\max_g\widehat\ell_{g,h}$ and $m_{g,h}=M_h-\widehat\ell_{g,h}$. For a candidate subset $Q$, let $b_h(Q)=\sum_{i\in Q}\widehat b_{i,h}(k)$ be its predicted workload contribution. Applying the marginal calculation at every step gives the reduction in the objective~\eqref{eq:projected-objective}:
\begin{equation}
\label{eq:routing-horizon-marginal}
\sum_{h=0}^{H}\gamma^h\left[b_h(Q)-G\bigl(b_h(Q)-m_{g,h}\bigr)_+\right],
\qquad \gamma\in(0,1].
\end{equation}
The expression rewards filling predicted gaps and charges for raising the maximum at each step. Its values depend on the workload projection supplied to the router.

Under the growing-cache prediction in Section~\ref{subsec:lookahead-policy}, $b_h(Q)=b(Q)+h|Q|$. The reported BR-H implementation uses the simpler approximation $b_h(Q)\approx b(Q)$ across the window and calibrates the overflow penalty through $\beta>0$. Its score is
\begin{equation}
\label{eq:brh-Fscore}
F_g^H(Q)
=\left(\sum_{h=0}^{H}\gamma^h\right)b(Q)
-\beta\sum_{h=0}^{H}\gamma^h\bigl(b(Q)-m_{g,h}\bigr)_+.
\end{equation}
At each step, the load reward has coefficient one, while $\beta$ controls the penalty for exceeding a projected margin. With $\beta=G$, this score is the exact marginal reduction of the objective when new-request contributions are held constant across the window. Other values calibrate the penalty for the practical search. The prediction of ongoing requests supplies the departure information, and the next routing step incorporates the newly observed state.

At $H=0$ and $\beta=G$,~\eqref{eq:brh-Fscore} reduces to~\eqref{eq:br0-Fscore}. For $H>0$, the discount $\gamma$ reduces the weight of later offsets, and $\beta$ controls the penalty for exceeding projected margins. A candidate is below every predicted margin when $b(Q)\leq\min_hm_{g,h}$. Otherwise, penalties become active at the steps whose margins are exceeded. As a function of $b(Q)$, the score is concave and piecewise linear, with a breakpoint at each distinct $m_{g,h}$.

\subsection{Two stages of admission}
\label{subsec:routing-two-stage}

Both routers use the same allocation structure, with their respective scores. Within a routing step, let $c_g$ be worker $g$'s current number of free slots, initially $c_g(k)$, and let $S_{\mathrm{tot}}=\sum_gc_g$. These counts are updated as requests are admitted. While $S_{\mathrm{tot}}>S_{\mathrm{greedy}}$, the router selects the worker with the most free slots, breaking ties by the smallest current load, and admits the single waiting request with the largest score for that worker. The score evaluates both the gain from filling a gap and the cost of any overflow.

When the remaining capacity is at most $S_{\mathrm{greedy}}$, the router considers workers in descending order of the pair $(c_g,q_g)$, where $q_g=m_g$ for BR-0 and $q_g=\min_hm_{g,h}$ for BR-H. For each selected worker, it examines up to $R_{\max}$ candidates from the head of the waiting pool and chooses a subset of at most $c_g$ requests that maximizes the score. The pool order determines this candidate window. If the best score is nonpositive, the router admits a single request with the largest score instead. This fallback keeps admissions moving when requests are waiting and capacity is available. Loads, margins, and worker priorities are updated after every admission.

The subset search is exact within its candidate window, while worker selection and the sequence of windows remain heuristic. With at most $R_{\max}$ candidates, enumeration requires at most $2^{R_{\max}}$ subset evaluations, each costing $O(H+1)$ with precomputed margins. This limits the combinatorial search without solving the global integer program. Full pseudocode appears in Appendix~\ref{append:alg} (Algorithm \ref{alg:br0} and \ref{alg:brh}); Section~\ref{sec:experiments} measures the resulting computational overhead and operating performance.

\section{Numerical evaluation in LLM serving}
\label{sec:experiments}

We evaluate BR-0 and BR-H using the production-scale experiments reported in the system study \citep{BalanceRoute2026}. The experiments measure the operating performance of the practical routers in Section~\ref{sec:implementation}. 

\subsection{Experimental setup}
\label{subsec:exp_setup}

\paragraph{Serving configuration.}
The serving stack is \texttt{vllm-ascend} with the V1 scheduler, continuous batching, PagedAttention, and \texttt{cudagraph\_mode=FULL\_DECODE\_ONLY} on the decode tier \citep{vllm-ascend,kwon2023efficient}. \texttt{MooncakeConnector} transfers KV caches between prefill and decode workers \citep{mooncake}. The experiments use up to $144$ NPUs. The default configuration runs DeepSeek-V3 $671$B with W8A8 quantization on Atlas $910$C: four prefill nodes with $16$ NPUs each and TP$=8$, and one decode node with $32$ NPUs, DP$=8$, and TP$=4$. This 4P1D configuration uses $96$ NPUs and gives $G=8$ decode workers. The per-worker batch limit is the engine parameter $B=\texttt{max\_num\_seqs}$.

\paragraph{Workloads.}
Both workloads are evaluated near cluster saturation. The proprietary trace, denoted 30002, contains $8{,}000$ requests with mean prompt length $3{,}197$ and mean output length $1{,}185$ tokens; its outputs have a long right tail. Azure-2024 is the conversation split of the public Azure LLM Inference Trace \citep{azure2024llmtrace,stojkovic2025dynamollm}. It is filtered to output lengths above $1{,}000$ tokens, giving $10{,}000$ requests with mean prompt length $4{,}652$ and mean output length $1{,}052$ tokens. The reported Azure results therefore concern this filtered workload rather than the complete public trace.

\paragraph{Baselines and router variants.}
The empirical baselines are the Random, Round-Robin, Power-of-Two-Choices, and Join-Shortest-Queue (JSQ) policies from the vLLM router; JSQ is the upstream \texttt{vllm-ascend} default \citep{vllm-router,mitzenmacher2001power}. These empirical comparators use the engine policies from the system evaluation. The size-agnostic analytical benchmark provides a separate theoretical comparison. Workload-aware endpoint routers address routing to independently scheduled instances, whereas the control point here is assignment from prefill to barrier-coupled decode workers with resident KV caches \citep{jain2024intelligent,jain2025performance}.

BR-0 uses no termination prediction. BR-H uses $H=80$ and is evaluated with either oracle or deployed predictions. The oracle replaces the estimated capped remaining length $\widehat c_i(k)$ by the trace value $\min\{r_i(k),H\}$. It is a reference for perfect prediction, not a deployable method. The deployed Survival estimator uses the marginal training-output CDF; ExactMatch uses a prompt-hash-specific CDF and falls back to Survival when the prompt has no matching history. Their specifications and offline accuracy are in Appendix~\ref{subsec:brh-prediction}.

\paragraph{Performance measures.}
We report the trace mean of the aggregate load imbalance defined in~\eqref{eq:imbalancek},
\[
I(k)=G\max_g L_g(k)-\sum_{g=1}^{G}L_g(k).
\]
For $K$ measured decode iterations, the reported mean imbalance is $K^{-1}\sum_{k=1}^{K}I(k)$, in tokens. Thus the experiments use the same imbalance metric as the model and theory. TPOT P95 is the $95$th percentile of time per output token, measured in milliseconds. Throughput is output tokens per second. The reported metrics concern the decode stage.

\subsection{Results on the two workloads}
\label{subsec:exp_main}

\begin{figure}[htbp]
\centering
\begin{subfigure}[t]{0.48\linewidth}
    \centering
    \includegraphics[width=\linewidth]{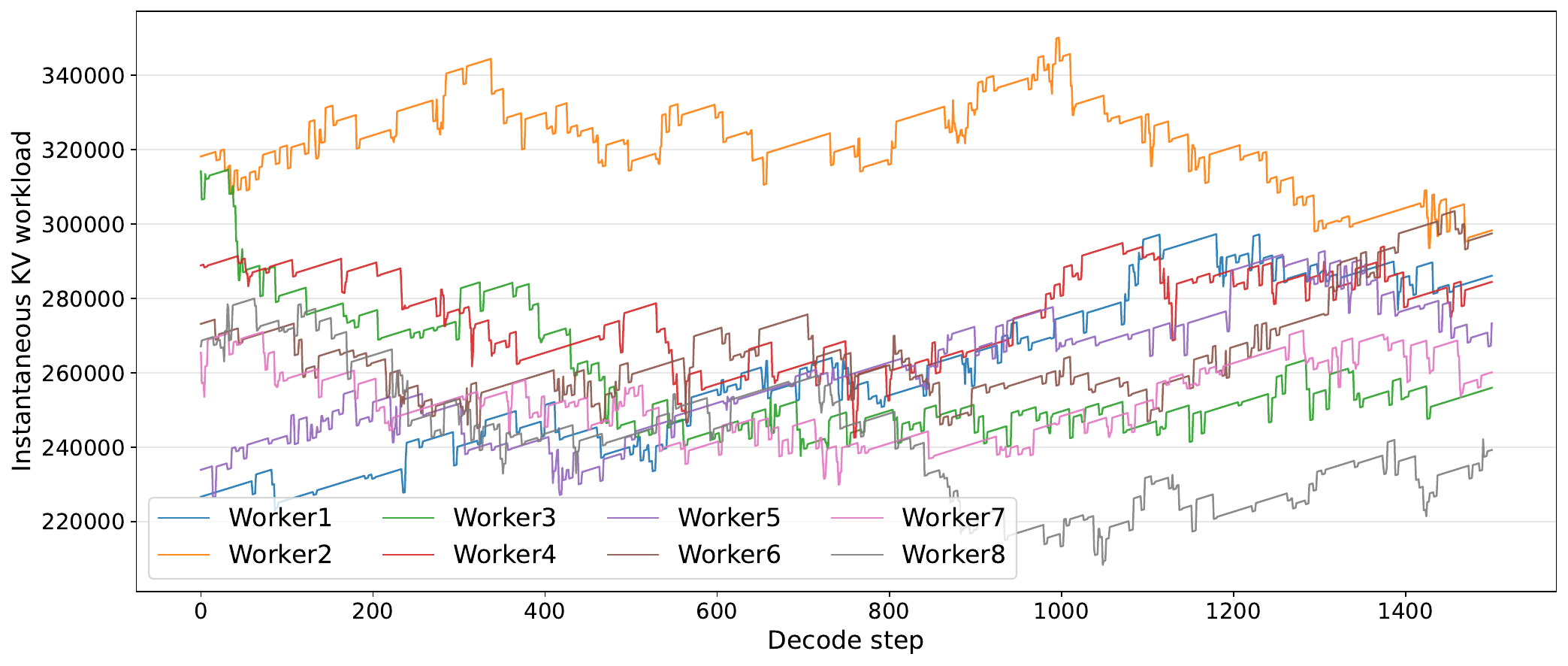}
    \caption*{\centering\footnotesize Random}
\end{subfigure}\hfill
\begin{subfigure}[t]{0.48\linewidth}
    \centering
    \includegraphics[width=\linewidth]{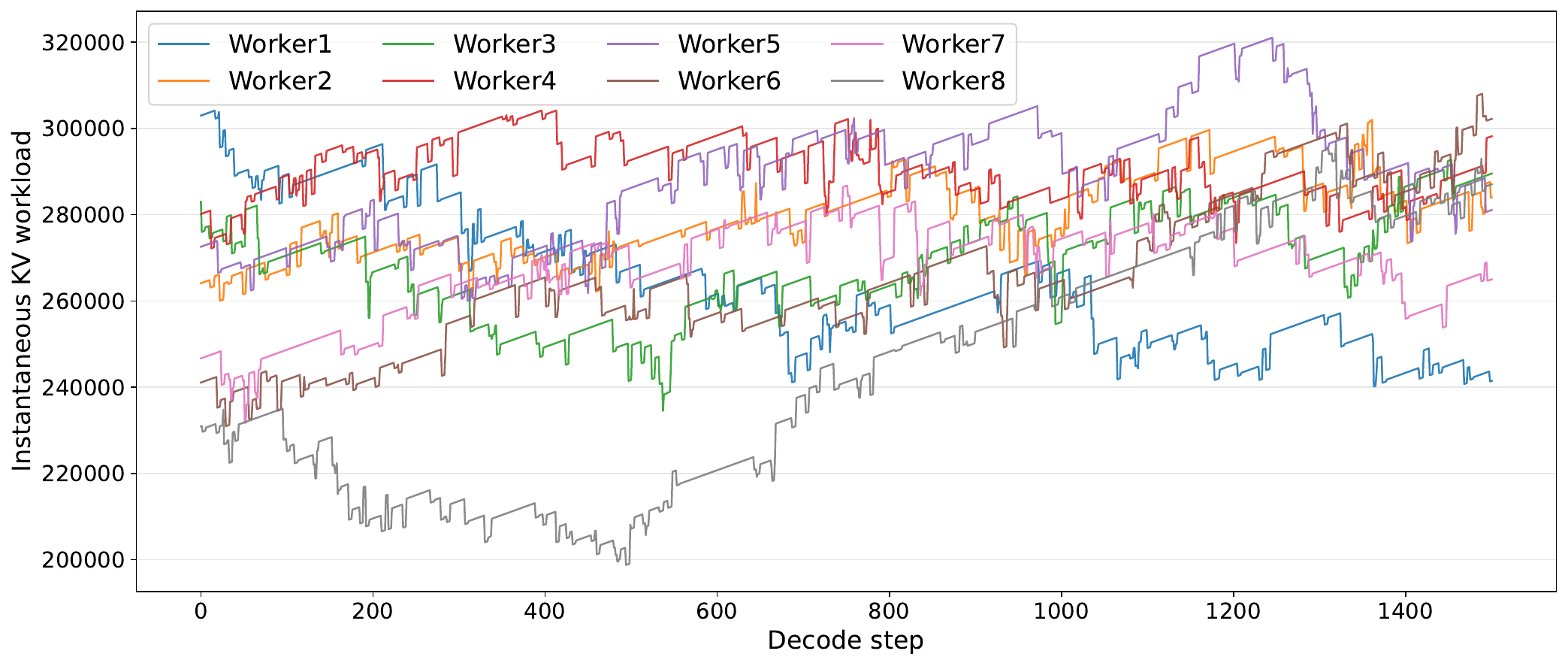}
    \caption*{\centering\footnotesize JSQ }
\end{subfigure}

\vspace{0.7em}
\begin{subfigure}[t]{0.48\linewidth}
    \centering
    \includegraphics[width=\linewidth]{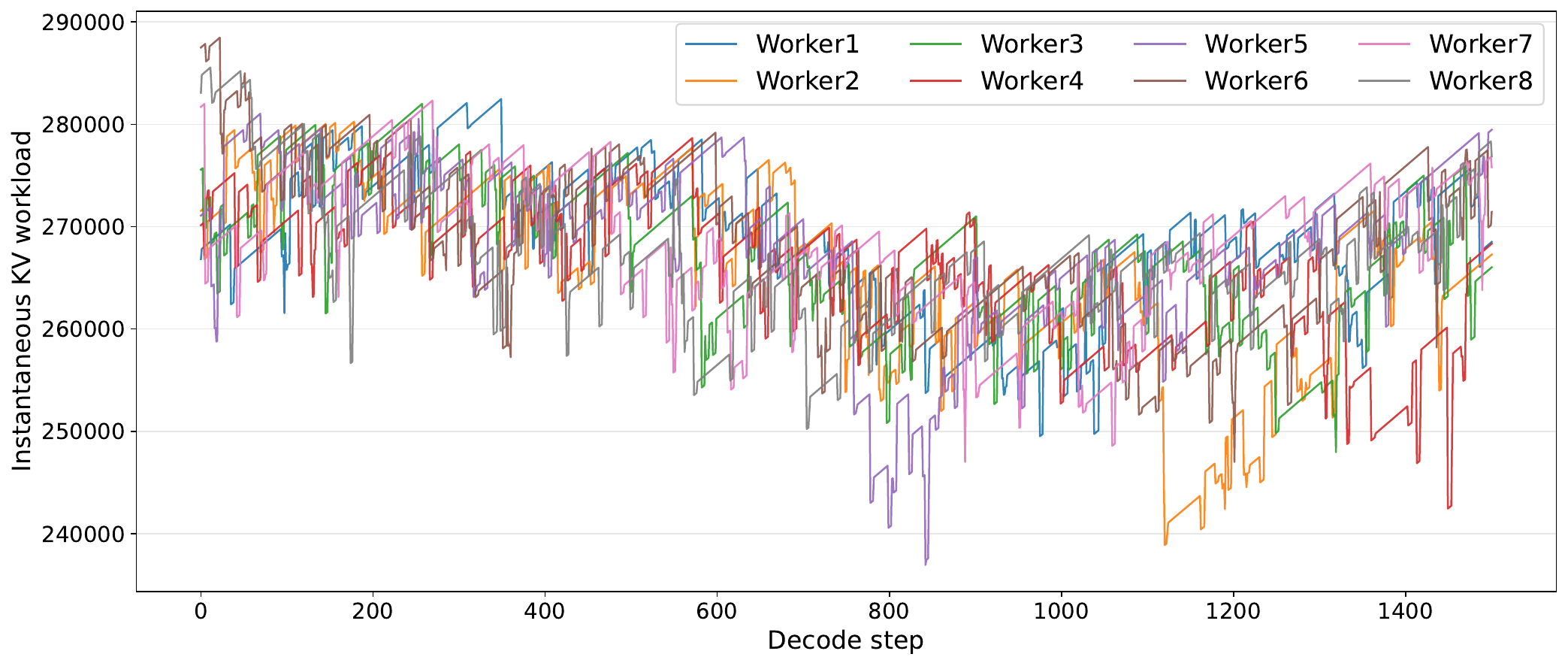}
    \caption*{\centering\footnotesize \textbf{BR-0}}
\end{subfigure}\hfill
\begin{subfigure}[t]{0.48\linewidth}
    \centering
    \includegraphics[width=\linewidth]{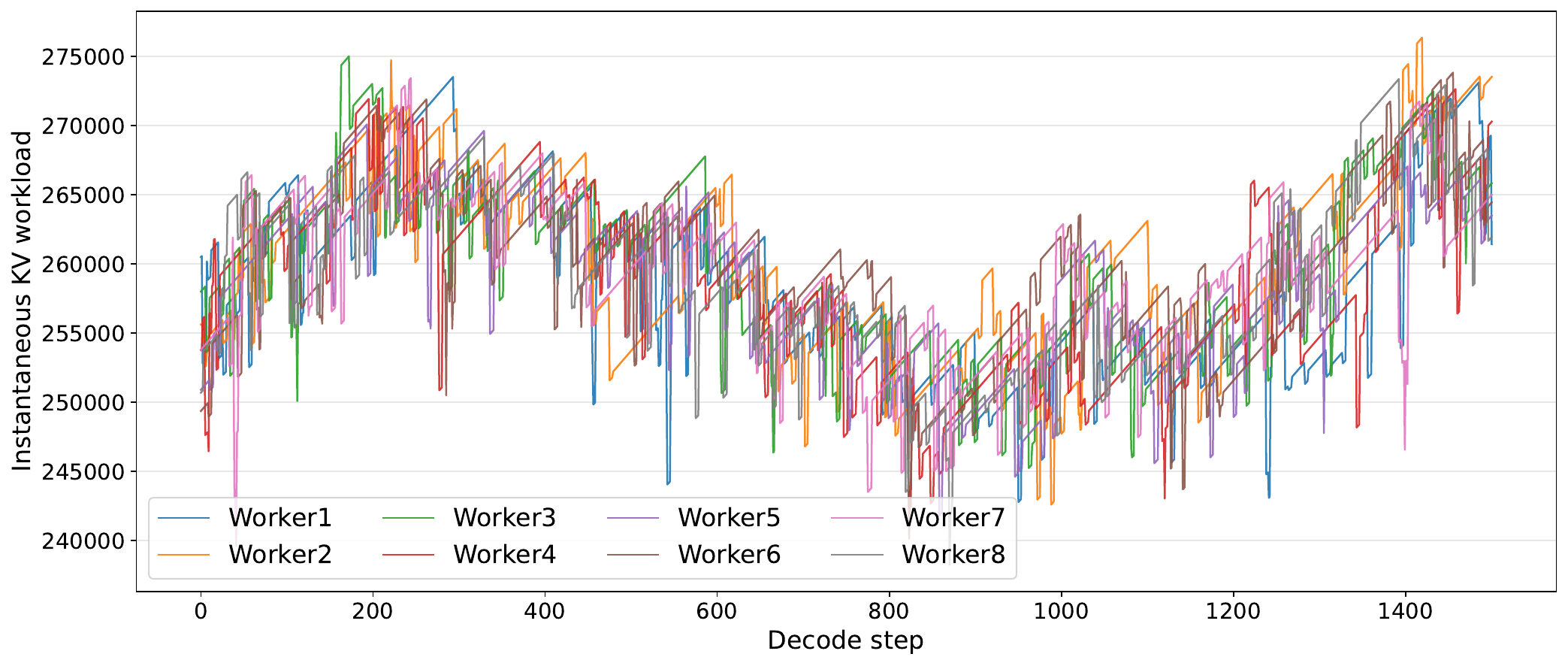}
    \caption*{\centering\footnotesize \textbf{BR-H} predicted, ExactMatch}
\end{subfigure}

\caption{\textbf{Per-worker KV-cache workload on 30002, $G{=}8$.} Each panel plots all eight workers over the same $1{,}500$-step decode segment. The remaining panels (RR, P2C, BR-H oracle) are in Appendix~\ref{app:full-kv-grids}.}
\label{fig:kv-workload-grid}
\end{figure}

\begin{table}[htbp]
\centering
\small
\setlength{\tabcolsep}{4pt}
\caption{\textbf{Main results on 30002 and Azure-2024.} $G{=}8$, heavy load (Azure-2024 filtered to output length $>1000$). Mean imbalance $I$ in tokens, TPOT P95 in ms, throughput (Tput) in tok/s; $\downarrow$ lower is better, $\uparrow$ higher is better. Primary deployment in \textbf{bold}.}
\label{tab:main_combined}
\begin{tabular}{lcccccc}
\toprule
& \multicolumn{3}{c}{\textbf{30002} (proprietary)} & \multicolumn{3}{c}{\textbf{Azure-2024}} \\
\cmidrule(lr){2-4} \cmidrule(lr){5-7}
Method & Mean $I$ $\downarrow$ & TPOT $\downarrow$ & Tput $\uparrow$ & Mean $I$ $\downarrow$ & TPOT $\downarrow$ & Tput $\uparrow$ \\
\midrule
Random                & 438{,}225 & 85.2 & 811     & 146{,}564 & 90.3 & 706 \\
Round-Robin           & 254{,}798 & 83.2 & 847     & 110{,}049 & 89.1 & 771 \\
Power-of-Two-Choices  & 262{,}940 & 84.0 & 836     & 109{,}033 & 89.3 & 794 \\
Join-Shortest-Queue   & 215{,}110 & 83.2 & 843     & 104{,}737 & 89.3 & 785 \\
\midrule
BR-0 \emph{prediction-free}                & 51{,}927          & 79.3          & 943             & 54{,}051          & 83.3          & 854 \\
BR-H \emph{oracle}, $(43, 0.86)$           & 23{,}975          & 78.0          & 1{,}022         & 37{,}831          & 78.1          & 951 \\
BR-H \emph{oracle}, $(14.67, 0.64)$        & 23{,}577          & 78.0          & 1{,}029         & 35{,}301          & 77.3          & 957 \\
BR-H \emph{deployed}, Survival             & 40{,}757          & 78.7          & 979             & 44{,}010          & 82.1          & 893 \\
\textbf{BR-H} \emph{deployed}, ExactMatch  & \textbf{36{,}277} & \textbf{78.8} & \textbf{977}    & \textbf{38{,}496} & \textbf{82.3} & \textbf{909} \\
\bottomrule
\end{tabular}
\end{table}

Every BR variant in Table~\ref{tab:main_combined} has a smaller mean imbalance, lower TPOT P95, and higher throughput than every baseline on both workloads. Relative to JSQ, BR-0 reduces mean imbalance by factors of approximately $4.1$ on 30002 and $1.94$ on Azure-2024. Its throughput is $943$ versus $843$ tok/s on 30002 and $854$ versus $785$ tok/s on Azure-2024. The best baseline throughput is attained by Round-Robin on 30002 and Power-of-Two-Choices on Azure-2024, so the strongest baseline depends on the performance measure.

Oracle BR-H further reduces the mean imbalance to about $24$ thousand tokens on 30002 and $35$ to $38$ thousand tokens on Azure-2024. At $(\beta,\gamma)=(14.67,0.64)$ its throughputs are $1{,}029$ and $957$ tok/s. With deployed ExactMatch predictions, the throughputs are $977$ and $909$ tok/s and the mean imbalances are $36{,}277$ and $38{,}496$ tokens. Survival gives $979$ and $893$ tok/s. Both deployed predictors improve TPOT P95 over all four baselines; the oracle rows have lower TPOT P95 than the deployed rows. Thus the data separate the gain from prediction-free routing, the additional gain available with perfect lookahead, and the gain realized by the two deployed estimators.

Figure~\ref{fig:kv-workload-grid} shows the eight worker trajectories over a common $1{,}500$-step segment of 30002. Appendix~\ref{app:full-kv-grids} contains the remaining 30002 panels and all reported Azure panels. Appendix~\ref{app:survival-30002} reports the supplementary Survival analysis.

\subsection{Scaling with the number of workers}
\label{subsec:exp_scaling}

The scaling experiments use $G=4$ in 2P1D with $48$ NPUs, $G=8$ in 4P1D with $96$ NPUs, and $G=16$ in 5P1D with $144$ NPUs. Request rate is scaled proportionally with $G$ to keep offered load per worker constant. BR-H uses oracle predictions at $H=80$ and $(\beta,\gamma)=(14.67,0.64)$ throughout this comparison. Deployed predictions are evaluated at $G=8$ only.

\begin{table}[htbp]
\centering
\caption{\textbf{Scaling summary on 30002.} $G{=}4$ (2P1D, 48 NPUs), $G{=}8$ (4P1D, 96 NPUs), $G{=}16$ (5P1D, 144 NPUs). BR-H runs with oracle prediction at $H{=}80$, $(\beta, \gamma){=}(14.67, 0.64)$. Each cell reports mean imbalance $I$ / TPOT P95 / throughput. Best in each column in \textbf{bold}.}
\label{tab:scaling_summary}
\begin{tabular}{lccc}
\toprule
Method  & $G{=}4$ (2P1D) & $G{=}8$ (4P1D) & $G{=}16$ (5P1D) \\
\midrule
Random       &  80{,}118 / 66.30 /  364.7 & 438{,}225 / 85.15 /  810.5 & 638{,}062 / 77.80 /  923.6 \\
Round-Robin  &  68{,}108 / 65.27 /  365.7 & 254{,}798 / 83.16 /  846.7 & 616{,}089 / 77.91 /  925.3 \\
P2C          &  64{,}012 / 65.40 /  366.4 & 262{,}940 / 84.01 /  836.3 & 688{,}283 / 77.46 /  927.3 \\
JSQ          &  75{,}326 / 66.45 /  365.4 & 215{,}110 / 83.20 /  842.8 & 675{,}607 / 77.40 /  926.6 \\
\midrule
BR-0         &   9{,}711 / 63.19 /  407.1 &  51{,}927 / 79.32 /  942.5 & 215{,}170 / 74.92 / 1{,}117.8 \\
BR-H \emph{oracle}
             & \textbf{6{,}576} / \textbf{62.52} / \textbf{415.3}
             & \textbf{23{,}577} / \textbf{78.04} / \textbf{1{,}028.9}
             & \textbf{117{,}067} / \textbf{72.04} / \textbf{1{,}247.4} \\
\bottomrule
\end{tabular}
\end{table}

\begin{figure}[htbp]
\centering
\includegraphics[width=0.8\linewidth]{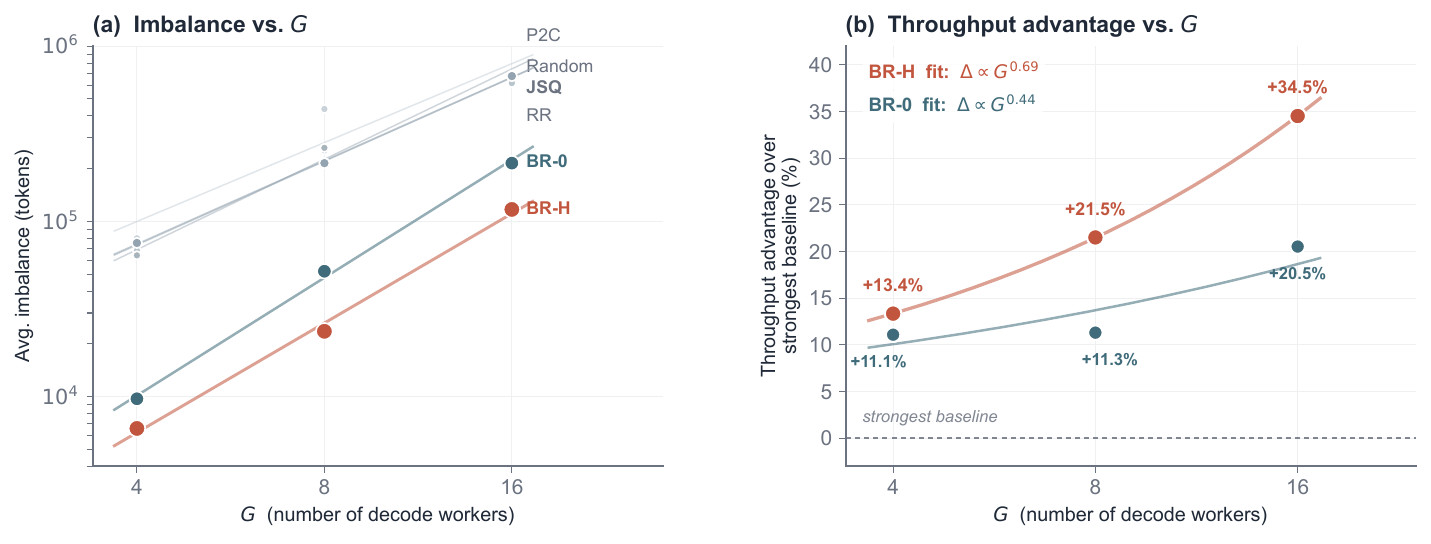}
\caption{\textbf{Scaling on 30002 across $G \in \{4, 8, 16\}$.} }
\label{fig:scaling}
\end{figure}

BR-0 and oracle BR-H outperform all four baselines at each tested configuration. Oracle BR-H has mean imbalances of approximately $6.6$, $24$, and $117$ thousand tokens at $G=4,8,16$, compared with baseline imbalances of approximately $64$ to $80$, $215$ to $438$, and $616$ to $688$ thousand tokens. Its throughput advantage over the best baseline at each configuration grows from approximately $13.3\%$ to $21.5\%$ to $34.5\%$. BR-0 has corresponding advantages of approximately $11\%$, $11\%$, and $20.5\%$. Oracle BR-H improves on BR-0 throughput by approximately $2\%$, $9\%$, and $12\%$.

These are comparisons within each serving configuration. Increasing $G$ also changes the prefill topology: 5P1D adds one prefill node relative to 4P1D while doubling decode workers, changing the prefill-to-decode ratio from $4{:}1$ to $5{:}2$. Absolute throughput is consequently not a controlled linear-scaling comparison. Appendix~\ref{app:scaling} contains the individual tables, both $G=16$ oracle operating points, and their worker trajectories. The experiments compare the mean imbalance $I$ and serving performance of the practical routers across these configurations.

\subsection{Parameter sensitivity and computation time}
\label{subsec:exp_robust_overhead}

The parameter study evaluates oracle BR-H with $H=80$ on 30002. It varies $\beta\in\{1,24,48,96\}$ at $\gamma=0.9$, and $\gamma\in\{0.5,0.7,0.9,1.0\}$ at $\beta=48$, giving seven distinct configurations. At $G=8$, the reported throughput is approximately $945$ to $1{,}042$ tok/s and TPOT P95 is approximately $77.3$ to $79.5$ ms. All seven configurations outperform the baselines on the reported measures. The exact values, the sensitivity map, and worker trajectories at $G=8$ and $G=16$ remain in Appendix~\ref{app:sensitivity-traces}. The sweep concerns these operating points and oracle predictions; it does not establish robustness for arbitrary parameters or prediction errors.

Dispatcher timing is measured on DeepSeek-V3 $671$B and Qwen3-30B-A3B, both with $G=8$ and $R_{\max}=4$. DeepSeek-V3 has median $1.20$ ms, mean $1.34$ ms, and $99$th percentile $2.77$ ms; Qwen3 has median $1.03$ ms and $99$th percentile $2.52$ ms. The medians differ by approximately $17\%$. Relative to an engine step of approximately $60$ ms, the dispatcher cost is approximately $50$ times smaller at the median and $22$ times smaller at the $99$th percentile. Stage~2 accounts for approximately $69\%$ of the total. Appendix~\ref{subsec:exp_overhead} reports the overhead figure, the stage-level timing statistics, and the measurement procedure.

Each reported performance cell comes from one trace replay. Repeated hardware trials are not available, so these tables do not provide replication-based confidence intervals. The two workloads and three cluster configurations establish the reported operating comparisons; the filtered Azure workload and oracle-only scaling of BR-H limit the conclusions that can be drawn from them.

\section{Conclusion}\label{sec:conclusion}

Load balancing in synchronous LLM decoding requires tracking work that changes while requests remain on their assigned workers. Equalizing active request counts leaves differences in resident KV workload, and the synchronization barrier turns those differences into repeated waiting. We model this loss through aggregate imbalance and use joint admission and assignment to fit new prompt workloads to the state left by ongoing requests. The short lookahead formulation provides a way to incorporate near-term completion information into this decision.

For the exact zero lookahead policy, we give a rigorous scale-dependent comparison with global FCFS admission and request-count routing. The cumulative supply condition permits irregular batching and empty release iterations while ensuring enough demand and prompt diversity in the long run. Under the stated growth regime, the imbalance improvement factor is at least of order $\sqrt{B\log G}$. The benchmark extension shows why changing among capacity-aware JSQ, round robin, random routing, and power of two choices cannot remove the workload fluctuations in this model when admission remains FCFS and destinations ignore prompt workloads. The result concerns their long-run imbalance in a regime of excess demand. It supplies no waiting-time stability or wall-clock throughput guarantee.

The BalanceRoute experiments provide complementary evidence about feasible routing decisions in a serving system \citep{BalanceRoute2026}. BR-0 improves balance using observable current workloads, and BR-H obtains further benefits from predictions of ongoing request completions. The reported proprietary and public trace results measure workload spread, throughput, and token latency for these approximations. They demonstrate practical gains beyond the controlled geometric model. Establishing approximation bounds for BR-0 or BR-H remains a separate mathematical question, as does proving a benefit from positive lookahead.

Several extensions matter for operation. Selecting requests from a common pool can improve aggregate balance while postponing particular requests, so waiting-time and tenant fairness constraints should enter the optimization explicitly. Separate worker queues reduce the available selection decisions and require a different comparison. Correlation between prompt size and output length, heterogeneous completion hazards, and release times that depend on request marks also change the stochastic structure used in our proof. These questions offer a path from the present saturated benchmark to service guarantees for a broader range of workloads. The central operational lesson is to distinguish capacity measured in request counts from work measured in evolving cache sizes, and to use the latter when synchronization makes the largest workload determine the group's progress.

\section{Data and code availability}\label{sec:data-availability}
The numerical study uses the experimental results reported by \citet{BalanceRoute2026}. The public Azure workload is available through the Microsoft Azure trace collection \citep{azure2024llmtrace}. The proprietary production trace is not included in this manuscript's source package. The package contains the tables and figure assets used here, together with source files for the analysis and a finite enumeration script supporting selected deterministic proof steps. The conference manuscript provides the implementation reference and its code-release link. 

\clearpage
\bibliographystyle{informs2014}
\bibliography{sn-bibliography}

@inproceedings{shahout2024don,
  title = {Don't Stop Me Now: Embedding Based Scheduling for {LLMs}},
  author = {Shahout, Rana and Malach, Eran and Liu, Chunwei and Jiang, Weifan and Yu, Minlan and Mitzenmacher, Michael},
  year = {2025},
  booktitle = {International Conference on Learning Representations (ICLR)},
  url = {https://openreview.net/forum?id=7JhGdZvW4T}
}

@inproceedings{zheng2023response,
  title = {Response Length Perception and Sequence Scheduling: An {LLM}-Empowered {LLM} Inference Pipeline},
  author = {Zheng, Zangwei and Ren, Xiaozhe and Xue, Fuzhao and Luo, Yang and Jiang, Xin and You, Yang},
  volume = {36},
  pages = {65517-65530},
  year = {2023},
  booktitle = {Advances in Neural Information Processing Systems},
  doi = {10.52202/075280-2859},
  publisher = {Curran Associates, Inc.},
  url = {https://proceedings.neurips.cc/paper_files/paper/2023/hash/ce7ff3405c782f761fac7f849b41ae9a-Abstract-Conference.html}
}

@inproceedings{patel2023splitwise,
  title = {{Splitwise}: Efficient Generative {LLM} Inference Using Phase Splitting},
  author = {Patel, Pratyush and Choukse, Esha and Zhang, Chaojie and Shah, Aashaka and Goiri, {\'I}{\~n}igo and Maleki, Saeed and Bianchini, Ricardo},
  pages = {118-132},
  year = {2024},
  booktitle = {2024 ACM/IEEE 51st Annual International Symposium on Computer Architecture (ISCA)},
  publisher = {IEEE},
  doi = {10.1109/ISCA59077.2024.00019}
}

@inproceedings{kwon2023efficient,
  title = {Efficient Memory Management for Large Language Model Serving with {PagedAttention}},
  author = {Kwon, Woosuk and Li, Zhuohan and Zhuang, Siyuan and Sheng, Ying and Zheng, Lianmin and Yu, Cody Hao and Gonzalez, Joseph E. and Zhang, Hao and Stoica, Ion},
  booktitle = {Proceedings of the 29th Symposium on Operating Systems Principles},
  pages = {611-626},
  year = {2023},
  publisher = {Association for Computing Machinery},
  doi = {10.1145/3600006.3613165}
}

@inproceedings{yu2022orca,
  title = {{Orca}: A Distributed Serving System for {Transformer}-Based Generative Models},
  author = {Yu, Gyeong-In and Jeong, Joo Seong and Kim, Geon-Woo and Kim, Soojeong and Chun, Byung-Gon},
  booktitle = {16th USENIX Symposium on Operating Systems Design and Implementation (OSDI 22)},
  pages = {521-538},
  year = {2022},
  publisher = {USENIX Association},
  url = {https://www.usenix.org/conference/osdi22/presentation/yu}
}

@inproceedings{zhong2024distserve,
  title = {{DistServe}: Disaggregating Prefill and Decoding for Goodput-optimized Large Language Model Serving},
  author = {Zhong, Yinmin and Liu, Shengyu and Chen, Junda and Hu, Jianbo and Zhu, Yibo and Liu, Xuanzhe and Jin, Xin and Zhang, Hao},
  booktitle = {18th USENIX Symposium on Operating Systems Design and Implementation (OSDI 24)},
  pages = {193-210},
  year = {2024},
  publisher = {USENIX Association},
  url = {https://www.usenix.org/conference/osdi24/presentation/zhong-yinmin}
}

@inproceedings{fu2024efficient,
  title = {Efficient {LLM} Scheduling by Learning to Rank},
  author = {Fu, Yichao and Zhu, Siqi and Su, Runlong and Qiao, Aurick and Stoica, Ion and Zhang, Hao},
  volume = {37},
  pages = {59006-59029},
  year = {2024},
  booktitle = {Advances in Neural Information Processing Systems},
  doi = {10.52202/079017-1882},
  publisher = {Curran Associates, Inc.},
  url = {https://proceedings.neurips.cc/paper_files/paper/2024/hash/6c8985579293e0209bdaa4f21bb1d237-Abstract-Conference.html}
}

@article{whitt1986deciding,
  title = {Deciding Which Queue to Join: Some Counterexamples},
  author = {Whitt, Ward},
  journal = {Operations Research},
  volume = {34},
  number = {1},
  pages = {55-62},
  year = {1986},
  publisher = {INFORMS},
  doi = {10.1287/opre.34.1.55}
}

@article{mukherjee2018universality,
  title = {Universality of Power-of-$d$ Load Balancing in Many-Server Systems},
  author = {Mukherjee, Debankur and Borst, Sem C. and van Leeuwaarden, Johan S. H. and Whiting, Philip A.},
  journal = {Stochastic Systems},
  volume = {8},
  number = {4},
  pages = {265-292},
  year = {2018},
  publisher = {INFORMS},
  doi = {10.1287/stsy.2018.0016}
}

@inproceedings{jain2025performance,
  title = {Performance Aware {LLM} Load Balancer for Mixed Workloads},
  author = {Jain, Kunal and Parayil, Anjaly and Mallick, Ankur and Choukse, Esha and Qin, Xiaoting and Zhang, Jue and Goiri, {\'I}{\~n}igo and Wang, Rujia and Bansal, Chetan and R{\"u}hle, Victor and Kulkarni, Anoop and Kofsky, Steve and Rajmohan, Saravan},
  booktitle = {Proceedings of the 5th Workshop on Machine Learning and Systems},
  pages = {19-30},
  year = {2025},
  publisher = {Association for Computing Machinery},
  doi = {10.1145/3721146.3721947}
}

@article{jain2024intelligent,
  title = {Intelligent Router for {LLM} Workloads: Improving Performance Through Workload-Aware Load Balancing},
  author = {Jain, Kunal and Parayil, Anjaly and Mallick, Ankur and Choukse, Esha and Qin, Xiaoting and Zhang, Jue and Goiri, {\'I}{\~n}igo and Wang, Rujia and Bansal, Chetan and R{\"u}hle, Victor and Kulkarni, Anoop and Kofsky, Steve and Rajmohan, Saravan},
  journal = {arXiv preprint arXiv:2408.13510},
  year = {2024},
  doi = {10.48550/arXiv.2408.13510},
  url = {https://arxiv.org/abs/2408.13510},
  note = {First posted in 2024; title follows version 2, revised January 7, 2025}
}

@manual{vllm-ascend,
  title = {{vLLM} {Ascend} Plugin},
  author = {{vLLM Community}},
  organization = {vLLM Project},
  url = {https://docs.vllm.ai/projects/ascend/en/latest/},
  year = {2025}
}

@misc{vllm-router,
  title = {{vLLM} Router},
  author = {{vLLM Project}},
  year = {2025},
  howpublished = {\url{https://github.com/vllm-project/router}},
  note = {Open-source request router for the {vLLM} serving stack}
}

@inproceedings{mooncake,
  title = {{Mooncake}: Trading More Storage for Less Computation: A {KVCache}-centric Architecture for Serving {LLM} Chatbot},
  author = {Qin, Ruoyu and Li, Zheming and He, Weiran and Cui, Jialei and Ren, Feng and Zhang, Mingxing and Wu, Yongwei and Zheng, Weimin and Xu, Xinran},
  booktitle = {23rd USENIX Conference on File and Storage Technologies (FAST 25)},
  pages = {155-170},
  year = {2025},
  publisher = {USENIX Association},
  url = {https://www.usenix.org/conference/fast25/presentation/qin}
}

@misc{azure2024llmtrace,
  title = {Microsoft {Azure} Traces},
  author = {{Microsoft Azure Project}},
  year = {2025},
  howpublished = {\url{https://github.com/Azure/AzurePublicDataset}},
  note = {Open-source Microsoft Azure traces}
}

@inproceedings{stojkovic2025dynamollm,
  title = {{DynamoLLM}: Designing {LLM} Inference Clusters for Performance and Energy Efficiency},
  author = {Stojkovic, Jovan and Zhang, Chaojie and Goiri, {\'I}{\~n}igo and Torrellas, Josep and Choukse, Esha},
  booktitle = {2025 IEEE International Symposium on High Performance Computer Architecture (HPCA)},
  pages = {1348-1362},
  year = {2025},
  organization = {IEEE},
  doi = {10.1109/HPCA61900.2025.00102}
}

@article{mitzenmacher2001power,
  title = {The Power of Two Choices in Randomized Load Balancing},
  author = {Mitzenmacher, Michael},
  journal = {IEEE Transactions on Parallel and Distributed Systems},
  volume = {12},
  number = {10},
  pages = {1094-1104},
  year = {2001},
  publisher = {IEEE},
  doi = {10.1109/71.963420}
}

@misc{nvidia2023h200,
  author = {{NVIDIA}},
  title = {{NVIDIA} Supercharges Hopper, the World's Leading {AI} Computing Platform},
  year = {2023},
  month = nov,
  howpublished = {\href{https://nvidianews.nvidia.com/news/nvidia-supercharges-hopper-the-worlds-leading-ai-computing-platform}{NVIDIA Newsroom}},
  note = {Published November 13, 2023}
}

@inproceedings{zheng2026uncertainty,
  author = {Zheng, Haoyu and Zhang, Yongqiang and Fu, Fangcheng and Zhou, Xiaokai and Luo, Hao and Zhu, Hongchao and Zhu, Yuanyuan and Wang, Hao and Yan, Xiao and Jiang, Jiawei},
  title = {Scheduling {LLM} Inference with Uncertainty-Aware Output Length Predictions},
  booktitle = {Proceedings of the 43rd International Conference on Machine Learning},
  series = {Proceedings of Machine Learning Research},
  volume = {306},
  pages = {163894-163918},
  year = {2026},
  publisher = {PMLR},
  url = {https://proceedings.mlr.press/v306/zheng26ae.html}
}

@misc{BalanceRoute2026,
  author = {Bu, Tianci and Lyu, Yuan and Chen, Zixi and Song, Chendong and Liang, Hong and Gurung, Tsepten and Fan, Yuwei and Ye, Yinyu and Zhou, Zijie},
  title = {Tackling the Data-Parallel Load Balancing Bottleneck in {LLM} Serving: Practical Online Routing at Scale},
  year = {2026},
  howpublished = {arXiv preprint arXiv:2605.06113},
  note = {Accepted at NeurIPS 2026},
  doi = {10.48550/arXiv.2605.06113},
  url = {https://arxiv.org/abs/2605.06113}
}

@article{shevtsova2011,
  author = {Shevtsova, Irina},
  title = {On the absolute constants in the {Berry-Esseen} type inequalities for identically distributed summands},
  journal = {arXiv preprint arXiv:1111.6554},
  year = {2011},
  url = {https://arxiv.org/abs/1111.6554},
  doi = {10.48550/arXiv.1111.6554}
}

@article{mendelson2025sparse,
  author = {Mendelson, Gal and Kuang, Xu},
  title = {Load Balancing Using Sparse Communication},
  journal = {Operations Research},
  volume = {74},
  number = {2},
  pages = {1026-1046},
  year = {2025},
  doi = {10.1287/opre.2022.0359},
  note = {Published online June 19, 2025; volume 74, issue 2 appeared in March-April 2026}
}

@article{chen2012balanced,
  author = {Chen, Hong and Ye, Heng-Qing},
  title = {Asymptotic Optimality of Balanced Routing},
  journal = {Operations Research},
  volume = {60},
  number = {1},
  pages = {163-179},
  year = {2012},
  doi = {10.1287/opre.1110.0998},
  url = {https://pubsonline.informs.org/doi/10.1287/opre.1110.0998}
}

@article{gupta2019slowdown,
  author = {Gupta, Varun and Walton, Neil},
  title = {Load Balancing in the Nondegenerate Slowdown Regime},
  journal = {Operations Research},
  volume = {67},
  number = {1},
  pages = {281-294},
  year = {2019},
  doi = {10.1287/opre.2018.1768},
  url = {https://pubsonline.informs.org/doi/10.1287/opre.2018.1768}
}

@article{goyal2025reusable,
  author = {Goyal, Vineet and Iyengar, Garud and Udwani, Rajan},
  title = {Asymptotically Optimal Competitive Ratio for Online Allocation of Reusable Resources},
  journal = {Operations Research},
  volume = {73},
  number = {4},
  pages = {1897-1915},
  year = {2025},
  doi = {10.1287/opre.2021.0695},
  url = {https://pubsonline.informs.org/doi/10.1287/opre.2021.0695}
}

@article{jaillet2025online,
  author = {Jaillet, Patrick and Jiang, Jiashuo and Mellou, Konstantina and Molinaro, Marco and Podimata, Chara and Zhou, Zijie},
  title = {Online Scheduling for {LLM} Inference with {KV} Cache Constraints},
  journal = {arXiv preprint arXiv:2502.07115},
  year = {2025},
  note = {Version 5, January 15, 2026},
  doi = {10.48550/arXiv.2502.07115},
  url = {https://arxiv.org/abs/2502.07115v5}
}

@article{ao2025optimizing,
  author = {Ao, Ruicheng and Luo, Gan and Simchi-Levi, David and Wang, Xinshang},
  title = {Optimizing {LLM} Inference: Fluid-Guided Online Scheduling with Memory Constraints},
  journal = {arXiv preprint arXiv:2504.11320},
  year = {2025},
  note = {Version 4, June 13, 2026},
  doi = {10.48550/arXiv.2504.11320},
  url = {https://arxiv.org/abs/2504.11320v4}
}

@inproceedings{nie2026queueing,
  author = {Nie, Chengyi and Si, Nian and Zhou, Zijie},
  title = {A Queueing-Theoretic Framework for Stability Analysis of {LLM} Inference with {KV} Cache Memory Constraints},
  booktitle = {Proceedings of the 43rd International Conference on Machine Learning},
  series = {Proceedings of Machine Learning Research},
  volume = {306},
  pages = {93016-93031},
  year = {2026},
  publisher = {PMLR},
  url = {https://proceedings.mlr.press/v306/nie26a.html}
}

@inproceedings{chen2026onlinelp,
  author = {Chen, Zixi and Ye, Yinyu and Zhou, Zijie},
  title = {Online Linear Programming for Multi-Objective Routing in {LLM} Serving},
  booktitle = {Proceedings of the 43rd International Conference on Machine Learning},
  series = {Proceedings of Machine Learning Research},
  volume = {306},
  pages = {17892-17912},
  year = {2026},
  publisher = {PMLR},
  url = {https://proceedings.mlr.press/v306/chen26ft.html}
}

@manual{vllm2026dp,
  author = {{vLLM Project}},
  title = {Data Parallel Deployment},
  year = {2026},
  note = {Official documentation, accessed October 10, 2026},
  url = {https://docs.vllm.ai/en/stable/serving/data_parallel_deployment/}
}

@article{chen2025adaptively,
  title={Adaptively robust llm inference optimization under prediction uncertainty},
  author={Chen, Zixi and Ye, Yinyu and Zhou, Zijie},
  journal={arXiv preprint arXiv:2508.14544},
  year={2025}
}

@article{wang2025llm,
  title={Llm serving optimization with variable prefill and decode lengths},
  author={Wang, Meixuan and Ye, Yinyu and Zhou, Zijie},
  journal={arXiv preprint arXiv:2508.06133},
  year={2025}
}

@article{ao2026service,
  title={Service-Induced Congestion in Memory-Constrained LLM Serving},
  author={Ao, Ruicheng and Dong, Jing and Luo, Gan and Simchi-Levi, David},
  journal={arXiv preprint arXiv:2606.15555},
  year={2026}
}

@misc{deepseek2024v3,
  author = {{DeepSeek-AI}},
  title = {{DeepSeek-V3} Technical Report},
  year = {2024},
  eprint = {2412.19437},
  archivePrefix = {arXiv},
  primaryClass = {cs.CL},
  doi = {10.48550/arXiv.2412.19437},
  url = {https://arxiv.org/abs/2412.19437},
  note = {Version 2, revised February 18, 2025}
}

@misc{alvarez2025wideep,
  author = {Alvarez, Eduardo and Chen, Xiaoming and Yang, Jun and Xie, Kaiyu and Yang, Dongxu},
  title = {Scaling Large {MoE} Models with Wide Expert Parallelism on {NVL72} Rack Scale Systems},
  year = {2025},
  month = oct,
  howpublished = {\href{https://developer.nvidia.com/blog/scaling-large-moe-models-with-wide-expert-parallelism-on-nvl72-rack-scale-systems/}{NVIDIA Technical Blog}},
  note = {Published October 20, 2025}
}

@article{kong2026geometry,
  title={Geometry-aware online scheduling for llm serving: From theoretical bound to system practice},
  author={Kong, Li and Qi, Qi and Ye, Yinyu and Zhou, Zijie},
  journal={arXiv preprint arXiv:2606.22327},
  year={2026}
}

@article{feng2026general,
  title={General Non-Clairvoyant KV-Cache Scheduling via Regime-Aware Routing},
  author={Feng, Yiding and Liu, Siyu and Yang, Zonghan and Zhang, Yuhao},
  journal={arXiv preprint arXiv:2607.09248},
  year={2026}
}

@article{feng2026competitive,
  title={Competitive Non-Clairvoyant KV-Cache Scheduling for LLM Inference},
  author={Feng, Yiding and Yang, Zonghan and Zhang, Yuhao},
  journal={arXiv preprint arXiv:2601.22996},
  year={2026}
}

@article{wang2023openchat,
  title={Openchat: Advancing open-source language models with mixed-quality data},
  author={Wang, Guan and Cheng, Sijie and Zhan, Xianyuan and Li, Xiangang and Song, Sen and Liu, Yang},
  journal={arXiv preprint arXiv:2309.11235},
  year={2023}
}

@article{wang2024burstgpt,
  title={Burstgpt: A real-world workload dataset to optimize llm serving systems},
  author={Wang, Yuxin and Chen, Yuhan and Li, Zeyu and Kang, Xueze and Tang, Zhenheng and He, Xin and Guo, Rui and Wang, Xin and Wang, Qiang and Zhou, Amelie Chi and others},
  journal={arXiv preprint arXiv:2401.17644},
  year={2024}
}

@article{zheng2023lmsys,
  title={Lmsys-chat-1m: A large-scale real-world llm conversation dataset},
  author={Zheng, Lianmin and Chiang, Wei-Lin and Sheng, Ying and Li, Tianle and Zhuang, Siyuan and Wu, Zhanghao and Zhuang, Yonghao and Li, Zhuohan and Lin, Zi and Xing, Eric P and others},
  journal={arXiv preprint arXiv:2309.11998},
  year={2023}
}

@article{zhao2024wildchat,
  title={Wildchat: 1m chatgpt interaction logs in the wild},
  author={Zhao, Wenting and Ren, Xiang and Hessel, Jack and Cardie, Claire and Choi, Yejin and Deng, Yuntian},
  journal={arXiv preprint arXiv:2405.01470},
  year={2024}
}
\clearpage
\setcounter{section}{0}
\setcounter{subsection}{0}
\setcounter{subsubsection}{0}
\renewcommand{\thesection}{EC.\arabic{section}}
\renewcommand{\theHsection}{supplement.\arabic{section}}
\renewcommand{\theHsubsection}{\theHsection.\arabic{subsection}}
\renewcommand{\theHsubsubsection}{\theHsubsection.\arabic{subsubsection}}
\setcounter{equation}{0}
\renewcommand{\theequation}{EC.\arabic{equation}}
\renewcommand{\theHequation}{supplement.\arabic{equation}}
\setcounter{figure}{0}
\renewcommand{\thefigure}{EC.\arabic{figure}}
\renewcommand{\theHfigure}{supplement.\arabic{figure}}
\setcounter{table}{0}
\renewcommand{\thetable}{EC.\arabic{table}}
\renewcommand{\theHtable}{supplement.\arabic{table}}
\setcounter{theorem}{0}
\renewcommand{\thetheorem}{EC.\arabic{theorem}}
\providecommand{\theHtheorem}{\arabic{theorem}}
\renewcommand{\theHtheorem}{supplement.\arabic{theorem}}
\setcounter{lemma}{0}
\renewcommand{\thelemma}{EC.\arabic{lemma}}
\providecommand{\theHlemma}{\arabic{lemma}}
\renewcommand{\theHlemma}{supplement.\arabic{lemma}}
\setcounter{proposition}{0}
\renewcommand{\theproposition}{EC.\arabic{proposition}}
\providecommand{\theHproposition}{\arabic{proposition}}
\renewcommand{\theHproposition}{supplement.\arabic{proposition}}
\setcounter{corollary}{0}
\renewcommand{\thecorollary}{EC.\arabic{corollary}}
\providecommand{\theHcorollary}{\arabic{corollary}}
\renewcommand{\theHcorollary}{supplement.\arabic{corollary}}
\setcounter{definition}{0}
\renewcommand{\thedefinition}{EC.\arabic{definition}}
\providecommand{\theHdefinition}{\arabic{definition}}
\renewcommand{\theHdefinition}{supplement.\arabic{definition}}
\setcounter{remark}{0}
\renewcommand{\theremark}{EC.\arabic{remark}}
\providecommand{\theHremark}{\arabic{remark}}
\renewcommand{\theHremark}{supplement.\arabic{remark}}
\section*{Supplementary Materials}
\noindent This companion gives the rigorous proof, benchmark extension, algorithm details, prediction methods, and supplementary numerical evidence. The numerical tables and plots retain their main-text or appendix placement from the conference study.
\section{Proof of Theorem~\ref{thm:geometric} }\label{app:geometric-proof}\label{append:geometric-proof}
This appendix proves Theorem~\ref{thm:geometric} for the original exact current-imbalance policy $\IO$. It also proves the common stationary benchmark law used in Proposition~\ref{prop:benchmark-extension}.

Throughout, $n=GB$, $p\in(0,1)$, $q=1-p$, and the fixed integer $s\ge2$ bounds the iid prompt marks $S_i\sim F$. The iid hidden lifetimes $O_i\sim\Geo(p)$, on $\{1,2,\ldots\}$, are independent of the entire prompt sequence. Define
\begin{equation}\label{app:input-supply}
 \rho=1-\max_{t\in\{1,\ldots,s\}}\Pp(S_i=t)>0,
 \qquad \liminf_{k\to\infty}\frac{A(k)}k>\frac{np}{\rho}.
\end{equation}
The system is empty at epoch zero. The exogenous nondecreasing finite-prefix release counts $A(k)$, with $A(0)=0$, are fixed before input marks are sampled. Every unadmitted request remains in the global queue. Workload at age $j$ is $S_i+j$; a newly admitted request starts at age zero and receives no completion test until the next epoch. Permanent physical slots and sticky assignments are used in the constructions below. For brevity, write $I^\pi(k)=I(k;\pi)$.

For completeness, define all finite-parameter constants used in the proof. Let $A_*\sim\Geo(p)-1$, independent of $S\sim F$, and put
\begin{align}
 \sigma^2&=\Var(S)+q/p^2,\label{app:sigma}\\
 \mu_2&=\E[(s+1+A_*)^2]
       =\left(s+\frac1p\right)^2+\frac q{p^2},\label{app:mu2}\\
 \theta_0&=\min\left\{\frac1{s+1},-\frac12\log q\right\},\qquad
 V=pq\E[(s+1+A_*)^2e^{\theta_0(s+1+A_*)}],\label{app:V}\\
 \tau&=\min\{\theta_0,p/(2V)\},\qquad \kappa=p\tau/4,\label{app:kappa}\\
 \varepsilon_{B,G}&=\sqrt{8npq\mu_2\min\{1,4Ge^{-\kappa B}\}}.\label{app:epsilon}
\end{align}
All these quantities are finite, and $V,\tau,\kappa$ are positive. The lower-bound constants are
\begin{equation}\label{app:lower-constants}
 \lambda_0=\min\left\{\frac{\log2}{2s},\log\frac{1+q}{2q}\right\},\quad
 R_{\rm BE}=\frac{4\sqrt2}{\sigma}\max\{s,96/p\},\quad
 B_0=\left\lceil(16R_{\rm BE})^2\right\rceil.
\end{equation}
The explicit finite condition for the benchmark lower bound is
\begin{equation}\label{app:finite-condition}
 B\ge B_0,\qquad \log G\le64\sigma^2\lambda_0^2B.
\end{equation}

\subsection{Ghost clocks, age moments, and cumulative supply}
The constructions below are made separately for each policy and preserve its original marked-input law. A comparison of expectations does not require one shared slot-clock realization to reproduce every policy under identical fixed request lifetimes.

\begin{lemma}[Auxiliary clocks and workload moments]\label{app:lem-ghost}
For any admissible policy, one may use an iid array
\[
 \xi_{g,b}(k)\sim\operatorname{Bernoulli}(p),\quad
 g\in[G],\ b\in[B],\ k\ge1,
\]
such that an occupied slot completes when its clock equals one. Coins in empty slots are auxiliary ghost events. Define
\[
 W_k=\sum_{j=1}^k\sum_{g,b}\xi_{g,b}(j),\qquad
 X_{g,b}(0)=0,\quad
 X_{g,b}(k)=(1-\xi_{g,b}(k))(X_{g,b}(k-1)+1).
\]
Then $W_k\sim\operatorname{Binomial}(nk,p)$, cumulative actual completions are at most $W_k$, and actual ages are at most their corresponding ghost ages. At each fixed $k$, the $X_{g,b}(k)$ are independent and each is stochastically dominated by $A_*$. In particular,
\begin{equation}\label{app:eq-secondmoment}
 \sup_k\E[I^\pi(k)^2]
 \le C_{B,G}^2,
 \qquad
 C_{B,G}=(G-1)B\sqrt{2s^2+4n/p^2}.
\end{equation}
\end{lemma}
\proof{Proof.}
Geometric memorylessness gives an independent Bernoulli$(p)$ completion hazard for every currently occupied slot, conditional on observable history. One can generate these hazards with independent fresh slot coins and supply independent coins for idle slots. Equivalently, unrevealed request lifetimes can be sampled by deferred decisions at admission; this reproduces their independent geometric law because the policy cannot inspect them before admission. Past auxiliary coins may be included in the filtration without changing independence of future coins.

Every actual completion is a clock event, giving the completion bound. For age domination, a surviving job and its ghost both increment; a completing job resets its ghost, and any newly admitted job starts at age zero. A previously empty slot's new job also has age zero, no greater than its ghost age. Induction therefore gives the claimed domination from the empty initial state. Ghost ages use independent clock histories in different slots, so they are independent. For $r\ge1$, the event $X_{g,b}(k)\ge r$ requires $r$ consecutive zero coins and has probability at most $q^r$; for $r>k$ it is impossible. This is the tail of $A_*$.

Write $M_k=\max_{g,b}X_{g,b}(k)$. Since each workload is at most $s+M_k$ and a worker holds at most $B$ jobs,
\[
 I^\pi(k)\le(G-1)B(s+M_k).
\]
Furthermore $M_k^2\le\sum_{g,b}X_{g,b}(k)^2$ and
\[
 \E[A_*^2]=\frac{q(1+q)}{p^2}\le\frac2{p^2}.
\]
Using $(s+M_k)^2\le2s^2+2M_k^2$ proves \eqref{app:eq-secondmoment}.
\Halmos\endproof

\begin{lemma}[Explicit vanishing supply failures]\label{app:lem-supply}
For each fixed schedule satisfying \eqref{eq:theory-supply}, there are events $H_k$ and deterministic numbers $f_k\in[0,1]$ with $\sum_k f_k<\infty$ such that $\Pp(H_k^c)\le f_k$, and on $H_k$ the pre-admission pool satisfies
\begin{equation}\label{app:eq-diversity}
 |R_k\setminus R_k(t)|\ge n\quad(t=1,\ldots,s).
\end{equation}
Consequently all slots are full after admission on $H_k$. These conclusions hold without conditioning the input law on any diversity or saturation event.
\end{lemma}
\proof{Proof.}
Choose a number $\ell$ strictly between $np/\rho$ and the liminf in \eqref{eq:theory-supply}. For some deterministic $k_0$, $A(k)\ge\ell k$ for every $k\ge k_0$. Set
\[
 \Delta=\rho\ell-np>0,\quad a=np+\Delta/3,\quad b=np+2\Delta/3,
\]
and $k_1=\max\{k_0,\lceil6n/\Delta\rceil,1\}$. Let $d=|\supp(F)|$ and
\[
 N_t(m)=\sum_{i=1}^m\1\{S_i\ne t\}.
\]
For $k\ge k_1$, define $H_k$ by $N_t(A(k))\ge bk$ for every supported $t$ and $W_{k-1}\le ak$. Set $H_k=\varnothing$ at earlier epochs, including $k=0$.

The number of requests admitted before epoch $k$ equals previous completions plus the number then active, hence is at most $n+W_{k-1}$. Charging every previous admission against each prompt class gives the pathwise lower bound
\[
 |R_k\setminus R_k(t)|\ge N_t(A(k))-n-W_{k-1}.
\]
On $H_k$, its right side is at least $\Delta k/3-n\ge n$. Unsupported classes are easier, since their outside class is the entire pool. This proves \eqref{app:eq-diversity}. The waiting pool thus has at least $n$ jobs, while free capacity is at most $n$; maximal feasible admission fills every slot.

For supported $t$, $N_t(A(k))$ is binomial with success probability at least $\rho$. Since
\[
 \frac{bk}{A(k)}\le\frac b\ell=\rho-\frac\Delta{3\ell},
\]
Hoeffding's inequality yields
\[
 \Pp(N_t(A(k))<bk)
 \le\exp\left(-\frac{2\Delta^2k}{9\ell}\right).
\]
Also $W_{k-1}\le W_k$ and $W_k\sim\operatorname{Binomial}(nk,p)$, so the same inequality gives
\[
 \Pp(W_{k-1}>ak)\le\exp\left(-\frac{2\Delta^2k}{9n}\right).
\]
A union bound proves the claim with $f_k=1$ for $k<k_1$ and
\begin{equation}\label{app:eq-fk}
 f_k=\min\left\{1,
 d e^{-2\Delta^2k/(9\ell)}+e^{-2\Delta^2k/(9n)}\right\}
 \quad(k\ge k_1).
\end{equation}
This sequence is summable. Borel-Cantelli also shows that the good supply and full occupancy hold eventually almost surely, but subsequent arguments use unconditional probabilities and moments.
\Halmos\endproof

\subsection{A deterministic comparator for the original objective}
\begin{lemma}[Separation, including zero vacancies]\label{app:lem-comparator}
Consider a finite waiting pool of sizes in $\{1,\ldots,s\}$, pre-admission loads $a_g\ge0$, and integer free capacities $c_g\ge0$. Put $C=\sum_g c_g$ and suppose
\begin{equation}\label{app:eq-poollemma}
 |R\setminus R(t)|\ge C\quad(t=1,\ldots,s).
\end{equation}
Among allocations admitting exactly $c_g$ jobs to each worker, choose a comparator minimizing lexicographically
\[
 (D,N,I),\qquad
 D=\max_g L_g-\min_g L_g,\quad
 N=\#\{g:L_g=\max_hL_h\},\quad I=G\max_gL_g-\sum_gL_g.
\]
Let $r$ be any comparator maximum worker. For any $g$ with $L_r-L_g>s$, there is $\theta\in\{0,\ldots,s\}$ for which
\begin{equation}\label{app:eq-threshold}
 \sum_{i\in T_r}S_i-\sum_{i\in T_g}S_i
 \le\theta c_r-(\theta+1)c_g,
\end{equation}
where $T_h$ is the admitted set on worker $h$. The original $I$-minimizing allocation has objective no greater than this comparator's objective, with no prescribed tie-breaking on the original allocation.
\end{lemma}
\proof{Proof.}
The feasible set is finite and nonempty by \eqref{app:eq-poollemma}; when $C=0$, the empty allocation is feasible. Suppose first that $c_r,c_g>0$. If admitted sizes contain $x>y$ on workers $r,g$, respectively, exchange these jobs. The two new loads are $L_r-(x-y)$ and $L_g+(x-y)$. Because $0<x-y\le s-1<L_r-L_g$, both lie strictly between the old pair's endpoints. The maximum cannot increase and the minimum cannot decrease. If $D$ decreases, comparator optimality fails. If $D$ ties, the maximum and minimum must separately retain their values, and worker $r$ ceases to be a maximizer while worker $g$ remains below the maximum. Then $N$ decreases, again contradicting optimality. Thus
\[
 \max_{i\in T_r}S_i\le\min_{i\in T_g}S_i.
\]
If equality occurs at $t$, some selected size-$t$ job occurs on each worker. Since the selected set has cardinality $C$ and contains size-$t$ jobs, \eqref{app:eq-poollemma} implies that an unselected job has size $z\ne t$. If $z<t$, replace a size-$t$ job on $r$ by it. The modified load remains strictly above $L_g$, so the minimum is unchanged. The maximum decreases, or its value ties and its number of maximizers decreases. Either contradicts the first two comparator criteria. If $z>t$, replace a size-$t$ job on $g$. Its load remains strictly below $L_r$; the maximum and $N$ remain unchanged and the minimum cannot decrease. If $D$ ties, total load increases by $z-t$, strictly decreasing $I$. This contradicts the third criterion. Hence equality is impossible, and an integer threshold $\theta\le s-1$ separates the two admitted sets strictly. This yields \eqref{app:eq-threshold}.

If $c_r=0$, take $\theta=0$: the left side is at most $-c_g$, including when $c_g=0$. If $c_r>0$ and $c_g=0$, take $\theta=s$: the left side is at most $sc_r$. These cases exhaust zero capacities. Finally, the actual policy minimizes $I$ over the same feasible set, so it is no worse in objective value than the comparator.
\Halmos\endproof

\begin{lemma}[Unified endpoint-noise drift]\label{app:lem-drift}
Suppose the preceding state is full and \eqref{app:eq-poollemma} holds at the current epoch. Let $w_{g,b}$ be the preceding post-admission slot workloads, and define
\begin{align*}
 E_g(t)&=\sum_{b=1}^B(\xi_{g,b}-p)(t-w_{g,b}-1),\\
 u_*&=\max_{g,\,t\in\{0,s+1\}}E_g(t)
       -\min_{g,\,t\in\{0,s+1\}}E_g(t).
\end{align*}
For the original exact policy, pathwise,
\begin{equation}\label{app:eq-unifieddrift}
 I^{\IO}_{\mathrm{new}}
 \le q I^{\IO}_{\mathrm{old}}+(G-1)\bigl[s+\pos{u_*-Bp}\bigr].
\end{equation}
\end{lemma}
\proof{Proof.}
Let $c_g=\sum_b\xi_{g,b}$, $L_g^{\rm old}=\sum_b w_{g,b}$, and
\[
 \zeta_g=-\sum_b(\xi_{g,b}-p)(w_{g,b}+1).
\]
The pre-admission load satisfies
\[
 a_g=\sum_b(1-\xi_{g,b})(w_{g,b}+1)
     =q(L_g^{\rm old}+B)+\zeta_g,
 \qquad E_g(t)=\zeta_g+t(c_g-Bp).
\]
Apply Lemma~\ref{app:lem-comparator}. For its maximum worker $r$ and a worker $g$ with new deficit greater than $s$, \eqref{app:eq-threshold} gives
\begin{align*}
 L_r^{c}-L_g^{c}
 &\le a_r-a_g+\theta c_r-(\theta+1)c_g\\
 &=q(L_r^{\rm old}-L_g^{\rm old})-Bp
     +E_r(\theta)-E_g(\theta+1).
\end{align*}
The functions $E_g(t)$ are affine. Since $\theta,\theta+1\in[0,s+1]$, their values lie between their respective endpoint values, and the last difference is at most $u_*$. If $M^{\rm old}=\max_hL_h^{\rm old}$, every large deficit is therefore at most
\[
 q(M^{\rm old}-L_g^{\rm old})+u_*-Bp.
\]
All other deficits are at most $s$. In either case a deficit is bounded above by
\[
 q(M^{\rm old}-L_g^{\rm old})+s+\pos{u_*-Bp}.
\]
Sum over $g\ne r$. The omitted old deficit $M^{\rm old}-L_r^{\rm old}$ is nonnegative, so the sum is at most the right side of \eqref{app:eq-unifieddrift}. The original IP has no greater objective than the comparator. This proves the result even when some or all $c_g$ are zero.
\Halmos\endproof

\subsection{Unconditional exponential moments and the IO upper bound}
For the following concentration lemma, define $w_{g,b}=0$ in an empty slot and use the ghost coins in \emph{all} $B$ slots. Thus the endpoint noises are defined at every state. Their moment bounds are unconditional; ghost-age independence is never claimed after conditioning on full occupancy or on $H_k$.

\begin{lemma}[Endpoint-noise concentration]\label{app:lem-mgf}
At every epoch, with the global endpoint definition just given,
\begin{equation}\label{app:eq-mgf}
 \E[e^{\theta E_g(t)}]\le e^{VB\theta^2/2}
 \quad(t\in\{0,s+1\},\ |\theta|\le\theta_0).
\end{equation}
Moreover,
\begin{align}
 \Pp(u_*>Bp)&\le\min\{1,4G e^{-\kappa B}\},\label{app:eq-noiseprob}\\
 \E[u_*^2]&\le8n pq\mu_2,\label{app:eq-noisesecond}\\
 \E\pos{u_*-Bp}&\le\varepsilon_{B,G}.\label{app:eq-noiseexcess}
\end{align}
\end{lemma}
\proof{Proof.}
Let $X_j$ denote a preceding-epoch ghost age and $Z_j=s+1+X_j$ for a slot $j$. Its endpoint coefficient $b_j=t-w_j-1$ satisfies $|b_j|\le Z_j$, both for occupied slots and for $w_j=0$ in idle slots. Conditional on the observable and ghost-coin past, coefficients are fixed and current coins are independent Bernoulli$(p)$. The inequality
\[
 e^x\le1+x+\tfrac12 x^2e^{|x|}
\]
and conditional centering yield
\begin{align*}
 \E[e^{\theta(\xi_j-p)b_j}\mid\mathcal F_{k-1}]
 &\le1+\frac{pq\theta^2}{2}b_j^2e^{|\theta b_j|}\\
 &\le1+\frac{pq\theta^2}{2}Z_j^2e^{|\theta|Z_j}.
\end{align*}
Multiply these conditional bounds over the $B$ slots of worker $g$, then take unconditional expectations. The $X_j$ are independent across slots and dominated by $A_*$; the displayed functions of $Z_j$ are increasing and nonnegative. Consequently, for $|\theta|\le\theta_0$,
\[
 \E[e^{\theta E_g(t)}]
 \le\left(1+\frac{pq\theta^2}{2}
      \E[(s+1+A_*)^2e^{\theta_0(s+1+A_*)}]\right)^B
 \le e^{VB\theta^2/2}.
\]
This factorization concerns the independent ghost ages, not the potentially dependent adaptive workloads. Finiteness follows because $q e^{\theta_0}\le\sqrt q<1$.

Markov's inequality with $\theta=\tau$ and $\tau\le p/(2V)$ gives
\[
 \Pp(E_g(t)>Bp/2)
 \le\exp(-\tau Bp/2+VB\tau^2/2)
 \le e^{-p\tau B/4}=e^{-\kappa B}.
\]
Use $-\tau$ for the other tail. Since $u_*\le2\max_{g,t}|E_g(t)|$, a union bound over $2G$ endpoints proves \eqref{app:eq-noiseprob}.

Conditional centering and independence also give
\[
 \E[E_g(t)^2]=pq\sum_{b=1}^B\E[(t-w_{g,b}-1)^2]
 \le Bpq\mu_2.
\]
Using $u_*^2\le4\sum_{g,t}E_g(t)^2$ proves \eqref{app:eq-noisesecond}. Finally,
\[
 \E\pos{u_*-Bp}
 \le\E[u_*\1\{u_*>Bp\}]
 \le\sqrt{\E[u_*^2]\Pp(u_*>Bp)},
\]
which is \eqref{app:eq-noiseexcess}.
\Halmos\endproof

\proof{Proof of the upper bound \eqref{eq:geometric-upper}.}
At epoch $k$, $H_{k-1}\cap H_k$ ensures full preceding occupancy and the comparator's pool-diversity condition: \eqref{app:eq-diversity} gives at least $n\ge C$ jobs outside every class. Lemma~\ref{app:lem-drift} applies on that event. Elsewhere, the nonnegative right side of its drift inequality can be augmented by the full current cost. Thus, with $e_k=\E[I^{\IO}(k)]$,
\[
 e_k\le q e_{k-1}+(G-1)(s+\varepsilon_{B,G})+b_k,
\]
where Lemmas~\ref{app:lem-ghost} and \ref{app:lem-supply} imply
\[
 b_k=\E[I^{\IO}(k)\1\{H_{k-1}^c\cup H_k^c\}]
 \le C_{B,G}\sqrt{f_{k-1}+f_k}\longrightarrow0.
\]
Sum this recursion for $k=1,\ldots,K$ and use $e_0=0$, $e_K\ge0$:
\[
 p\sum_{k=1}^K e_k
 \le K(G-1)(s+\varepsilon_{B,G})+\sum_{k=1}^K b_k-qe_K
 \le K(G-1)(s+\varepsilon_{B,G})+\sum_{k=1}^K b_k.
\]
Divide by $pK$. The Cesaro average of $b_k$ tends to zero for fixed $B,G$ and the fixed schedule, proving \eqref{eq:geometric-upper} for all finite $B,G$.
\Halmos\endproof

\paragraph{Computable exponential-moment constant.}
For $r=q e^{\theta_0}<1$, summing the geometric series and its first two derivatives gives
\begin{equation}\label{app:eq-Vexplicit}
 V=p^2q e^{\theta_0(s+1)}
 \left[\frac{(s+1)^2}{1-r}
       +\frac{2(s+1)r}{(1-r)^2}
       +\frac{r(1+r)}{(1-r)^3}\right].
\end{equation}
These are explicit proof constants. The exponent $\kappa$ can be small for long mean lifetimes or large prompt bounds; a finite production-scale guarantee requires evaluating the displayed bound, not only its asymptotic order.

\subsection{The common stationary snapshot of the benchmark class}
\begin{lemma}[Stationary limit under prompt-blind routing]\label{app:refresh}
For every policy $\pi$ in Definition~\ref{def:benchmark-class}, the post-admission physical-slot array converges to the product law of independent copies of
\[
 U=S+A_*,\qquad S\sim F,\quad A_*\sim\Geo(p)-1,\quad S\ \hbox{independent of}\ A_*.
\]
Moreover,
\begin{equation}\label{app:benchmark-limit}
 \lim_{K\to\infty}\overline I_K^\pi
 =\E\left[G\max_g\sum_{b=1}^B U_{g,b}-\sum_{g,b}U_{g,b}\right],
\end{equation}
where the $U_{g,b}$ are mutually independent. In particular, this formula applies to $\JSQ$.
\end{lemma}
\proof{Proof.}
Use the policy-specific slot-clock construction of Lemma~\ref{app:lem-ghost}. Let $Z$ be the external routing seed in Definition~\ref{def:benchmark-class}. It is independent of all input marks. Geometric memorylessness permits a clock array whose conditional law given the entire seed is the iid Bernoulli$(p)$ law; equivalently, generate that array independently of $Z$ and sample hidden lifetimes by deferred decisions. Thus the clocks are independent of $Z$ and of the entire prompt sequence.

Condition on the entire slot-clock array and $Z$. The fixed arrival counts, occupancy evolution, global FCFS admission indices, and destinations of successive admissions are now determined without inspecting prompts. This statement follows recursively over epochs and over admissions within an epoch: free capacity is determined by occupied slots and current clocks; the next request is the first waiting stream index; each destination is selected from available slots using only the permitted histories and the seed. Distinct actual admissions consequently receive independent $F$ prompt marks under this conditioning.

Assign the mark of each actual admission into a slot with $\xi=1$ to the reset-prompt coordinate $T_{g,b}(k)$. Assign admissions into previously empty slots with $\xi=0$ to a separate array. A request is admitted only once and a physical slot receives at most one new request per epoch, so both mappings are injective and their source-index sets are disjoint. Fill unused coordinates with fresh independent auxiliary $F$ marks. Conditional on the clocks and seed, every finite collection in these arrays is independent with law $F$. The complete reset array is therefore iid and independent of the clocks and seed. This is a reveal-at-admission construction for the proof; it does not assert fresh conditional iid marks given the larger operational filtration that already reveals queued prompts. Conditioning on the entire clock array is used only for this prompt-mapping argument, not to assert conditional independence of completed or future lifetimes.

Construct a virtual permanently full chain, initialized with fresh auxiliary iid $F$ prompts independent of all clocks, seed, and reset marks, and update it by
\begin{equation}\label{app:virtual-benchmark}
 Y^*_{g,b}(k)=(1-\xi_{g,b}(k))(Y^*_{g,b}(k-1)+1)
              +\xi_{g,b}(k)T_{g,b}(k).
\end{equation}
Whenever the actual preceding and current states are full, the actual slot follows this recursion with the same primitives. Let $r_k=\lfloor k/2\rfloor$. If the actual post-admission states are full at every epoch from $k-r_k$ through $k$, and every slot has a clock reset in epochs $k-r_k+1,\ldots,k$, then both arrays agree at epoch $k$. Identical updates in that window erase each possibly different initial load at its first reset. Lemma~\ref{app:lem-supply}, which holds for every maximal-admission policy, gives
\[
 \Pp(Y(k)\ne Y^*(k))
 \le\sum_{j=k-r_k}^{k}f_j+nq^{r_k}\longrightarrow0.
\]
Including epoch $k$ ensures its vacancies are refilled. No last-shortage time is treated as a stopping time. Actual and virtual imbalances have uniformly bounded second moments by Lemma~\ref{app:lem-ghost}; Cauchy-Schwarz shows that their expected imbalances differ by a quantity tending to zero.

The single-slot invariant law of \eqref{app:virtual-benchmark} is $U=S+A_*$. A reset supplies a fresh independent prompt, while a nonreset increments age by one. The law $\Pp(A_*=a)=pq^a$ is invariant for that age recursion, and prompt and age remain independent. Independent primitives give the product invariant law across slots.

Couple the fresh virtual chain to a chain initialized in that product invariant law using the same future resets and prompts. The initial stationary array is independent of those future primitives. After a slot's first reset its copies agree; before that reset their load difference is their initial difference. Hence
\[
 \E|Y^*_{g,b}(k)-Y^{\rm st}_{g,b}(k)|
 =q^k\E|Y^*_{g,b}(0)-Y^{\rm st}_{g,b}(0)|\longrightarrow0.
\]
The same coupling gives probability of any unmatched slot at most $nq^k$, proving product-law convergence, even in total variation on the discrete state space. For arbitrary slot arrays $y,y'$,
\[
 |I(y)-I(y')|\le(G+1)\sum_{g,b}|y_{g,b}-y'_{g,b}|.
\]
Initial differences are integrable. Thus virtual expected imbalance converges to its stationary value. Combining the two couplings and taking Cesaro averages proves \eqref{app:benchmark-limit}. Prompt-age independence yields $\Var(U)=\sigma^2$.
\Halmos\endproof

\subsection{A finite stationary lower bound by exponential tilting}
\begin{lemma}[Tilted moments]\label{app:lem-tiltmoments}
Let $U=S+A_*$, $\mu=\E U$, and tilt its law by $e^{\lambda U}/\E e^{\lambda U}$ for $0\le\lambda\le\lambda_0$. Write $\sigma_\lambda^2=\Var_\lambda(U)$ and $m_{3,\lambda}=\E_\lambda|U-\E_\lambda U|^3$. Then
\begin{equation}\label{app:eq-tiltedbounds}
 \tfrac12\sigma^2\le\sigma_\lambda^2\le8\sigma^2,
 \qquad \frac{m_{3,\lambda}}{\sigma_\lambda^3}\le R_{\rm BE}.
\end{equation}
\end{lemma}
\proof{Proof.}
The independent components remain independent under the tilt. The prompt law is proportional to $e^{\lambda S}\,dF(S)$. The age law is geometric on nonnegative integers with $q_\lambda=q e^\lambda$ and $p_\lambda=1-q_\lambda$. The choice of $\lambda_0$ ensures
\[
 e^\lambda\le\sqrt2,\quad q\le q_\lambda\le2q,
 \quad p_\lambda\ge p/2,
 \quad 2\lambda(s-1)\le\log2.
\]
For two independent prompt copies, the tilt likelihood ratio of their pair lies between $e^{-2\lambda(s-1)}$ and $e^{2\lambda(s-1)}$. Applying the identity $\Var(S)=\frac12\E(S-S')^2$ yields
\[
 \tfrac12\Var(S)\le\Var_\lambda(S)\le2\Var(S).
\]
For the age component, $q_\lambda/p_\lambda^2$ increases with $q_\lambda$, and
\[
 \Var(A_*)\le\Var_\lambda(A_*)\le8\Var(A_*).
\]
Adding variances gives the first part of \eqref{app:eq-tiltedbounds}.

The centered prompt is bounded in absolute value by $s$, so its third absolute moment is at most $s\Var_\lambda(S)$. For the tilted geometric age,
\[
 \E_\lambda[A_*^3]
 =\frac{q_\lambda(1+4q_\lambda+q_\lambda^2)}{p_\lambda^3}
 \le\frac{6q_\lambda}{p_\lambda^3}.
\]
Using $|x-y|^3\le4(x^3+y^3)$ for $x,y\ge0$ and Jensen's inequality,
\[
 \E_\lambda|A_*-\E_\lambda A_*|^3
 \le8\E_\lambda[A_*^3]
 \le\frac{48}{p_\lambda}\Var_\lambda(A_*).
\]
Finally $|x+y|^3\le4(|x|^3+|y|^3)$ and $p_\lambda\ge p/2$ imply
\[
 m_{3,\lambda}\le4\max\{s,96/p\}\sigma_\lambda^2.
\]
Divide by $\sigma_\lambda^3$ and use $\sigma_\lambda\ge\sigma/\sqrt2$ to obtain the stated $R_{\rm BE}$.
\Halmos\endproof

\begin{lemma}[One-worker right tail]\label{app:lem-tail}
Let $L=\sum_{b=1}^B U_b$ for iid copies of $U$. If $B\ge B_0$ and $\log G\le64\sigma^2\lambda_0^2B$, then, for
\[
 t=\frac\sigma{16}\sqrt{B\log G},\qquad a_0=\frac{e^{-1/4}}5,
\]
one has
\begin{equation}\label{app:eq-workerprob}
 \Pp(L\ge B\mu+t)\ge a_0G^{-1/4},
 \qquad \Pp(L\le B\mu)\ge\tfrac14.
\end{equation}
\end{lemma}
\proof{Proof.}
For $X=U-\mu$, let $K(\lambda)=\log\E e^{\lambda X}$. The moment generating function is finite throughout the chosen tilt interval. One has $K(0)=K'(0)=0$, $K''(\lambda)=\sigma_\lambda^2$, and Jensen gives $K(\lambda)\ge0$. Lemma~\ref{app:lem-tiltmoments} therefore implies
\begin{equation}\label{app:eq-Kprime}
 \tfrac12\sigma^2\lambda\le K'(\lambda)\le8\sigma^2\lambda.
\end{equation}
For $z>0$ set $t_z=\sigma\sqrt B z$ and $\lambda=2z/(\sigma\sqrt B)$. Suppose $\lambda\le\lambda_0$. Under this tilt the centered sum $Z_B=L-B\mu$ has mean $BK'(\lambda)$ and variance $B\sigma_\lambda^2$.

The iid Berry-Esseen inequality with constant $1$ is conservative; a smaller absolute constant is available in \cite{shevtsova2011}. Its hypotheses hold by Lemma~\ref{app:lem-tiltmoments}, and it gives
\[
 \sup_x\left|\Pp_\lambda\left(
 \frac{Z_B-BK'(\lambda)}{\sigma_\lambda\sqrt B}\le x\right)-\Phi(x)\right|
 \le\frac{R_{\rm BE}}{\sqrt B}\le\frac1{16}.
\]
Thus the tilted event
\[
 \mathcal E_\lambda=\{BK'(\lambda)\le Z_B
                 \le BK'(\lambda)+\sigma_\lambda\sqrt B\}
\]
has probability at least $\Phi(1)-\Phi(0)-2/16>1/5$. For example,
\[
 \Phi(1)-\Phi(0)
 \ge\frac1{\sqrt{2\pi}}\int_0^1(1-x^2/2)\,dx
 =\frac5{6\sqrt{2\pi}}>\frac{13}{40},
\]
which proves the strict bound after subtracting $1/8$. Closed interval endpoints are harmless, since their probability is at least the difference of the two distribution-function values.

By \eqref{app:eq-Kprime}, $Z_B\ge t_z$ on $\mathcal E_\lambda$. Changing measure and bounding $Z_B$ by the upper endpoint of that event yields
\begin{align*}
 \Pp(Z_B\ge t_z)
 &\ge\E_\lambda[e^{-\lambda Z_B+BK(\lambda)}\1_{\mathcal E_\lambda}]\\
 &\ge\tfrac15\exp\{-\lambda[BK'(\lambda)+\sigma_\lambda\sqrt B]+BK(\lambda)\}\\
 &\ge\tfrac15\exp\{-8B\sigma^2\lambda^2-\sqrt8\sigma\sqrt B\lambda\}\\
 &=\tfrac15\exp\{-32z^2-4\sqrt2 z\}.
\end{align*}
Now choose $z=\sqrt{\log G}/16$. The finite condition guarantees $\lambda\le\lambda_0$, and
\[
 32z^2=\tfrac18\log G,\qquad
 4\sqrt2z\le\tfrac18\log G+\tfrac14,
\]
the latter by completing a square. This proves the first part of \eqref{app:eq-workerprob}. Applying Berry-Esseen at $\lambda=0$, $x=0$ gives
\[
 \Pp(L\le B\mu)\ge\tfrac12-\tfrac1{16}\ge\tfrac14,
\]
proving the second part.
\Halmos\endproof

\begin{lemma}[Expected stationary imbalance]\label{app:lem-expectedmax}
Under \eqref{eq:geometric-finite-condition}, independent worker loads $L_g=\sum_{b=1}^B U_{g,b}$ satisfy
\begin{equation}\label{app:eq-snapshotlower}
 \E\left[G\max_gL_g-\sum_gL_g\right]
 \ge\frac{G\sigma}{1024}\sqrt{B\log G}.
\end{equation}
\end{lemma}
\proof{Proof.}
Let $M_{-1}=\max_{2\le g\le G}L_g$. Independence and Lemma~\ref{app:lem-tail} give
\begin{align*}
 \Pp(M_{-1}\ge B\mu+t)
 &\ge1-(1-a_0G^{-1/4})^{G-1}\\
 &\ge1-\exp\{-a_0(G-1)G^{-1/4}\}\\
 &\ge h:=1-\exp\{-a_0 2^{-1/4}\}>0.
\end{align*}
Here $(G-1)G^{-1/4}$ increases for $G\ge2$, so the same bound covers fixed small $G$. Use the exact identity
\[
 \E[\max_gL_g]-B\mu=\E[(M_{-1}-L_1)_+].
\]
On the independent events $M_{-1}\ge B\mu+t$ and $L_1\le B\mu$, the positive part is at least $t$. Therefore
\[
 \E[\max_gL_g]-B\mu\ge\frac h4t
 =\frac h{64}\sigma\sqrt{B\log G}.
\]
Finally, $e^{-1/4}\ge3/4$ and $2^{1/4}\le2$ imply $a_02^{-1/4}\ge3/40$. Since $1-e^{-x}\ge x/(1+x)$ for $x\ge0$,
\[
 h\ge\frac{3}{43}>\frac1{16},\qquad \frac h{64}>\frac1{1024}.
\]
Multiplying the expected excess maximum by $G$, using $\E\sum_gL_g=GB\mu$, proves \eqref{app:eq-snapshotlower}.
\Halmos\endproof

\subsection{Completion of the comparison theorem and scope}
\proof{Proof of the remaining statements of Theorem~\ref{thm:geometric}.}
Lemmas~\ref{app:refresh} and \ref{app:lem-expectedmax} prove \eqref{eq:geometric-lower}. Write
\[
 L=\frac{G\sigma}{1024}\sqrt{B\log G}>0,
 \qquad U=\frac{G-1}{p}(s+\varepsilon_{B,G})>0.
\]
For every $\delta>0$, the FCFS+Count-JSQ average is eventually at least $L-\delta$ and the IO average at most $U+\delta$. Their ratio is consequently at least $(L-\delta)/(U+\delta)$ whenever its denominator is positive; zero denominators give the infinite convention. Take the lower limit and then $\delta\downarrow0$ to obtain \eqref{eq:geometric-ratio}.

For fixed $p,F,s$, the constants $\kappa,\mu_2,\sigma,\lambda_0,B_0$ are fixed. From \eqref{app:epsilon},
\[
 \varepsilon_{B,G}^2
 \le32nGpq\mu_2 e^{-\kappa B}.
\]
The logarithm of its right side is a constant plus $\log n+\log G-\kappa B$, which tends to $-\infty$ when $B\to\infty$ and $\log n=o(B)$. Thus $\varepsilon_{B,G}\to0$. Also $B\ge B_0$ eventually and $\log G/B\to0$, so \eqref{eq:geometric-finite-condition} eventually holds. Once $\varepsilon_{B,G}\le s$, \eqref{eq:geometric-ratio} yields \eqref{eq:geometric-asymptotic} for each admissible schedule. The bound's constants do not depend on that schedule; taking the infimum after its time limit preserves it.
\Halmos\endproof

\section{Capacity-aware FCFS benchmark policies}\label{append:benchmark-policies}
This appendix defines the benchmark class used in the geometric comparison. Global FCFS determines which request enters service; the destination rule determines its worker and free physical slot. Request counts below mean the number of active requests, rather than the number waiting in a separate worker queue.

\begin{definition}[Prompt-blind global FCFS class]\label{def:benchmark-class}
A policy belongs to the class $\mathcal B$ if it satisfies the following conditions at every epoch.
\begin{enumerate}
\item It selects waiting requests in global stream order and admits
\[
 U_k=\min\left\{|R_k|,\sum_g c_g(k)\right\}
\]
requests. In particular, admission is maximal whenever both a waiting request and free capacity remain.
\item It assigns each admitted request to a free permanent physical slot and does not migrate active requests.
\item Its worker and slot choices, including all tie-breaking, are determined by the exogenous arrival counts, current occupancy and active-request counts, past admissions and completions, and an external routing seed independent of the entire prompt and lifetime sequences. They do not inspect prompt sizes or workloads, or unrevealed lifetimes.
\end{enumerate}
\end{definition}

The class includes deterministic routing by taking the external seed to be constant. Choices may be adaptive within an admission batch as active-request counts change. Since all policies in $\mathcal B$ are maximal, the supply events in Lemma~\ref{app:lem-supply} make every slot full after admission for each of them.

\paragraph{FCFS with active-request-count JSQ.}
For each next request in the global FCFS queue, let $N_g$ be the current number of active requests on worker $g$ and let
\[
 \mathcal E=\{g:N_g<B\}
\]
be the eligible workers. Policy $\JSQ$ chooses a worker in $\arg\min_{g\in\mathcal E}N_g$, assigns the request to a free slot there, and updates the active counts before admitting the next request. Ties and the choice among free slots are prompt-blind. The loop stops only when the queue is empty or $\mathcal E$ is empty. This is the main benchmark in Theorem~\ref{thm:geometric}.

\paragraph{Capacity-aware round robin.}
A pointer cycles through workers in a fixed order. For each next global FCFS request, the policy skips full workers, assigns the request to the next worker with a free slot, and advances the pointer. Prompt-blind slot choice and maximal admission put this policy in $\mathcal B$. The fixed-physical-slot-order FCFS policy is another member of the same class.

\paragraph{Capacity-aware random routing.}
For each next global FCFS request, the policy samples a worker from $\mathcal E$ using independent routing randomness and chooses a free slot on that worker without inspecting prompts. Uniform sampling among eligible workers is one example. Alternatively, the policy may sample uniformly among all free physical slots. These two sampling distributions differ when eligible workers have different amounts of free capacity, but both belong to $\mathcal B$. The policy repeats until all admissible requests have entered service.

\paragraph{Capacity-aware power of two choices.}
For each next global FCFS request, sample two workers from $\mathcal E$ with independent routing randomness, and select one with the smaller active-request count. Sampling with replacement is allowed; if there is only one eligible worker, select it. Ties and free-slot choices are prompt-blind, and counts are updated after each admission. A variant sampling from all workers is included only if full candidates are discarded and the rule falls back to an eligible worker, or retries until an eligible worker is selected, so that admission remains maximal almost surely. Write $\mathrm{P2C}$ for such a capacity-aware variant.

\begin{proposition}[Common benchmark snapshot and comparison]\label{prop:benchmark-extension}
Under the geometric marked-input and cumulative supply assumptions of Theorem~\ref{thm:geometric}, every $\pi\in\mathcal B$ has the same limiting physical-slot product law and the same long-run expected imbalance as $\JSQ$. Specifically, with independent copies $U_{g,b}$ of $U=S+A_*$, where $S\sim F$ and $A_*\sim\Geo(p)-1$ are independent,
\begin{equation}\label{eq:benchmark-common-limit}
 \lim_{K\to\infty}\overline I_K^\pi
 =\E\left[G\max_g\sum_{b=1}^B U_{g,b}-\sum_{g,b}U_{g,b}\right]
 =\lim_{K\to\infty}\overline I_K^{\JSQ}.
\end{equation}
Therefore, under \eqref{eq:geometric-finite-condition},
\begin{equation}\label{eq:benchmark-common-lower}
 \lim_{K\to\infty}\overline I_K^\pi
 \ge\frac{G\sigma}{1024}\sqrt{B\log G},
 \qquad \sigma^2=\Var(S)+q/p^2.
\end{equation}
The finite comparison \eqref{eq:geometric-ratio} and asymptotic comparison \eqref{eq:geometric-asymptotic} consequently hold with any $\pi\in\mathcal B$ in place of $\JSQ$. This includes the capacity-aware round-robin, random, and P2C rules defined above.
\end{proposition}
\proof{Proof.}
The definitions above give global FCFS selection, prompt-blind destinations, and maximal admission. Lemma~\ref{app:refresh} proves the product-law limit and \eqref{eq:benchmark-common-limit} for the entire class, using the conditional prompt mapping given all slot clocks and the independent routing seed. Lemma~\ref{app:lem-expectedmax} bounds that common snapshot expectation. The upper bound for $\IO$ is independent of the benchmark's destination rule, so the same limit comparison proves the remaining claims.
\Halmos\endproof

\paragraph{Scope of the benchmark extension.}
The proposition concerns global FCFS selection and capacity-aware maximal admission. Routing to separate worker queues, stopping because a sampled worker is full, or rejecting an assignment without trying available capacity need not satisfy Definition~\ref{def:benchmark-class}. Their occupancy or stationary laws require a separate argument. Prompt-size or KV-workload decisions also lie outside this benchmark class. The policies in $\mathcal B$ can have different imbalance during the initial shortage period; the proposition equates their long-run limits under the stated supply model.

\section{Routing Algorithm Pseudocode}
\label{append:alg}

The pseudocode separates the allocation rule from the prediction method. Within one routing step, loads include admissions already made during that step. Updating state means removing admitted requests from the pool, recording the worker assignments, reducing free capacity, and updating loads and margins. Worker priorities are refreshed when margins change. Ties not specified below use a fixed order.

The waiting pool carries the candidate order. The source BR-0 procedure orders candidates by prefill length; the BR-H procedure uses the head of its waiting pool. The choice of pool order affects the bounded candidate window and should be fixed when reproducing an experimental configuration. Neither algorithm treats an exact search within that window as a solution of the joint program~\eqref{eq:bf-ip}.

\begin{algorithm2e}[htbp]
\caption{BR-0: Routing from Current Loads}
\label{alg:br0}
\KwIn{Ordered waiting pool $R$, active sets $\{\mathcal A_g(k)\}$, current loads $\{\ell_g\}$, batch limit $B$, thresholds $S_{\mathrm{greedy}}\geq0$ and $R_{\max}\geq1$}
\KwOut{Admissions $\{S_g(k)\}_{g\in[G]}$}
Set $S_g(k)\gets\emptyset$ and $c_g\gets B-|\mathcal A_g(k)|$ for each worker; set $S_{\mathrm{tot}}\gets\sum_gc_g$\;
\lIf{$S_{\mathrm{tot}}=0$ \textbf{or} $R=\emptyset$}{\textbf{return}}
Set $M\gets\max_g\ell_g$ and $m_g\gets M-\ell_g$ for every worker\;
Use the score $F_g^0$ in~\eqref{eq:br0-Fscore}\;
\textbf{Stage 1: Single-request admissions}\;
\While{$S_{\mathrm{tot}}>S_{\mathrm{greedy}}$ \textbf{and} $R\neq\emptyset$}{
Choose $g\in\operatorname*{arg\,max}_{j\in[G]}c_j$, breaking ties by smallest $\ell_j$\;
Choose $i^\star\in\operatorname*{arg\,max}_{i\in R}F_g^0(\{i\})$\;
Admit $i^\star$ to $g$; add it to $S_g(k)$ and remove it from $R$\;
Reduce $c_g$ and $S_{\mathrm{tot}}$ by one; increase $\ell_g$ by $s_{i^\star}$\;
Update $M$ and all margins $m_j$\;
}
\textbf{Stage 2: Candidate-subset admissions}\;
\While{$S_{\mathrm{tot}}>0$ \textbf{and} $R\neq\emptyset$}{
Choose a worker $g$ with $c_g>0$ by descending lexicographic priority $(c_g,m_g)$\;
Let $\mathcal C_g$ contain the first $\min(R_{\max},|R|)$ pool entries\;
Choose $Q^\star\in\operatorname*{arg\,max}_{Q\subseteq\mathcal C_g,\ |Q|\leq c_g}F_g^0(Q)$\;
\If{$F_g^0(Q^\star)\leq0$}{
Choose $i^\dagger\in\operatorname*{arg\,max}_{i\in R}F_g^0(\{i\})$; set $Q^\star\gets\{i^\dagger\}$\;
}
Admit $Q^\star$ to $g$; add it to $S_g(k)$ and remove it from $R$\;
Reduce $c_g$ and $S_{\mathrm{tot}}$ by $|Q^\star|$; increase $\ell_g$ by $b(Q^\star)$\;
Update $M$, all margins $m_j$, and worker priorities\;
}
\end{algorithm2e}

\begin{algorithm2e}[htbp]
\caption{BR-H: Routing from Projected Loads}
\label{alg:brh}
\KwIn{Ordered waiting pool $R$, active sets $\{\mathcal A_g(k)\}$, current loads $\{\ell_g\}$, batch limit $B$, horizon $H\geq1$, prediction period $\Delta T\in\mathbb N$, thresholds $S_{\mathrm{greedy}}\geq0$ and $R_{\max}\geq1$, coefficients $\beta>0$ and $\gamma\in(0,1]$}
\KwOut{Admissions $\{S_g(k)\}_{g\in[G]}$}
Set $S_g(k)\gets\emptyset$ and $c_g\gets B-|\mathcal A_g(k)|$ for each worker; set $S_{\mathrm{tot}}\gets\sum_gc_g$\;
\lIf{$S_{\mathrm{tot}}=0$ \textbf{or} $R=\emptyset$}{\textbf{return}}
Refresh ongoing-request predictions when due, using the interface in Appendix~\ref{subsec:brh-prediction} and~\eqref{eq:composite}\;
Form projected loads $\widehat\ell_{g,h}$ for $h=0,\ldots,H$, with $\widehat\ell_{g,0}=\ell_g$\;
Set $M_h\gets\max_g\widehat\ell_{g,h}$ and $m_{g,h}\gets M_h-\widehat\ell_{g,h}$\;
Use the score $F_g^H$ in~\eqref{eq:brh-Fscore}\;
\textbf{Stage 1: Single-request admissions}\;
\While{$S_{\mathrm{tot}}>S_{\mathrm{greedy}}$ \textbf{and} $R\neq\emptyset$}{
Choose $g\in\operatorname*{arg\,max}_{j\in[G]}c_j$, breaking ties by smallest $\widehat\ell_{j,0}$\;
Choose $i^\star\in\operatorname*{arg\,max}_{i\in R}F_g^H(\{i\})$\;
Admit $i^\star$ to $g$; add it to $S_g(k)$ and remove it from $R$\;
Reduce $c_g$ and $S_{\mathrm{tot}}$ by one; update projected loads using the candidate-load approximation in Section~\ref{subsec:brh-fscore}\;
Update $\{M_h\}$ and all margins $m_{j,h}$\;
}
\textbf{Stage 2: Candidate-subset admissions}\;
\While{$S_{\mathrm{tot}}>0$ \textbf{and} $R\neq\emptyset$}{
Choose a worker $g$ with $c_g>0$ by descending lexicographic priority $(c_g,\min_hm_{g,h})$\;
Let $\mathcal C_g$ contain the first $\min(R_{\max},|R|)$ pool entries\;
Choose $Q^\star\in\operatorname*{arg\,max}_{Q\subseteq\mathcal C_g,\ |Q|\leq c_g}F_g^H(Q)$\;
\If{$F_g^H(Q^\star)\leq0$}{
Choose $i^\dagger\in\operatorname*{arg\,max}_{i\in R}F_g^H(\{i\})$; set $Q^\star\gets\{i^\dagger\}$\;
}
Admit $Q^\star$ to $g$; add it to $S_g(k)$ and remove it from $R$\;
Reduce $c_g$ and $S_{\mathrm{tot}}$ by $|Q^\star|$; update projected loads using the same candidate-load approximation\;
Update $\{M_h\}$, all margins $m_{j,h}$, and worker priorities\;
}
\end{algorithm2e}

BR-0 is also obtained by running the second rule with $H=0$ and $\beta=G$, using the same pool order and tie-breaking. In that case prediction refresh is omitted and the projected state consists only of the current loads. BR-H refreshes predictions every $\Delta T$ decode steps; its prediction settings are reported with the numerical results and in Appendix~\ref{subsec:brh-prediction}. The fallback to a single request keeps both rules work conserving at the routing step. Selecting that request by score does not provide an individual starvation guarantee.

\clearpage
\subsection{An alternative exact solver for the candidate window}
\label{append:routing-reachable-sums}

Both scores depend on a candidate subset only through its total prefill size. For a window $\mathcal C_g$ and cardinality limit $c=\min(c_g,|\mathcal C_g|)$, let $\mathcal D_j$ be the set of token sums attainable with exactly $j$ requests. Initialize $\mathcal D_0=\{0\}$ and $\mathcal D_j=\emptyset$ for $1\leq j\leq c$. For each candidate $i$, update
\[
\mathcal D_j\gets\mathcal D_j\cup
\{t+s_i:t\in\mathcal D_{j-1}\},
\qquad j=c,c-1,\ldots,1.
\]
The descending order prevents a candidate from being used twice. The admissible totals are $\mathcal T=\bigcup_{j=0}^c\mathcal D_j$; predecessor records indexed by the candidate prefix recover a subset for each selected total. Integer token sizes permit a bitset representation of these sets. This construction is pseudopolynomial in the sum of the candidate token sizes.

For BR-0, the score increases up to $m_g$ and decreases thereafter when $G>1$. An exact maximizer over $\mathcal T$ is therefore found by comparing the largest reachable total at most $m_g$ with the smallest reachable total at least $m_g$, omitting a probe when that side is empty. At $G=1$ the overflow portion is constant, and the same comparison still yields a maximizer.

For BR-H, the breakpoints are all distinct projected margins $m_{g,h}$, rather than only their minimum. One exact approach evaluates~\eqref{eq:brh-Fscore} at every total in $\mathcal T$. Another compares the nearest reachable totals on either side of each breakpoint together with the smallest and largest reachable totals. Between adjacent breakpoints the score is affine, so these probes include a maximizer over the reachable totals. The nonpositive-score fallback in the pseudocode is applied after either search.

Enumeration and the reachable-sum construction solve the same candidate-window subproblem. With $R_{\max}=4$, enumeration examines at most $16$ subsets per selected worker.

\section{Prediction specification and evaluation}
\label{app:prediction}

\subsection{Short-horizon prediction interface}
\label{subsec:brh-prediction}

At a predictor refresh, let $a_i$ be the number of tokens already generated and $r_i=o_i-a_i$ the remaining output length. Let $\mathcal H_i$ denote the causal information supplied to the estimator. The prediction target is the capped remaining length
\begin{equation}
\label{eq:horizon-contribution}
c_i=\mathbb E[\min\{r_i,H\}\mid\mathcal H_i]\in[0,H].
\end{equation}
It is a bounded scalar summary, rather than a prediction of the complete output length. Decomposition on the event $\{r_i\le H\}$ gives
\begin{equation}
\label{eq:two-stage}
c_i=(1-p_{\mathrm{fin}}^{(i)})H+p_{\mathrm{fin}}^{(i)}\mu_{\mathrm{rem}}^{(i)},
\end{equation}
where $p_{\mathrm{fin}}^{(i)}=\Pr(r_i\le H\mid\mathcal H_i)$ and $\mu_{\mathrm{rem}}^{(i)}=\mathbb E[r_i\mid\mathcal H_i,r_i\le H]$. A termination classifier estimates the first quantity; a conditional-mean regressor estimates the second on finish-positive samples. Cross-entropy and squared error are the respective training losses for a learned implementation. The composite estimate is
\begin{equation}
\label{eq:composite}
\widehat c_i=(1-\widehat p_{\mathrm{fin}}^{(i)})H+\widehat p_{\mathrm{fin}}^{(i)}\widehat\mu_{\mathrm{rem}}^{(i)},
\end{equation}
clipped to $[0,H]$ before applying the active-request floor described below.

BR-H uses $\widehat c_i$ as a departure cutoff in a plug-in KV-load scenario. The KV contribution grows by one token per subsequent decode iteration until the predicted departure; it is then set to zero. The current worker load is observed, and future worker loads are formed by summing these request-level scenarios. The resulting maximum loads and safe margins are updated as admissions are made. This projection is an approximation: the conditional mean of a capped duration does not determine all conditional survival probabilities, so the cutoff trajectory is not an exact conditional expected-load trajectory. Nor does it provide a theoretical guarantee for BR-H. Section~\ref{sec:experiments} evaluates this implemented prediction interface.

\subsection{Estimator realizations}
\label{app:predictor}
\label{app:predictor-realizations}

\paragraph{Survival.}
Given the empirical CDF $\widehat F$ of training output lengths, the termination estimate at age $a_i$ is
\[
\widehat p_{\mathrm{fin}}^{(i)}=
\frac{\widehat F(a_i+H)-\widehat F(a_i)}{1-\widehat F(a_i)}.
\]
The conditional remaining-length mean is
\[
\widehat\mu_{\mathrm{rem}}^{(i)}=
\frac{\sum_{j:a_i<o_j\le a_i+H}(o_j-a_i)}{|\{j:a_i<o_j\le a_i+H\}|}.
\]
Both estimates use a sorted training-output history with $O(\log n)$ inference cost. They require no content features. The displayed formulas apply when their conditioning sets are nonempty.

\paragraph{ExactMatch.}
A prompt-hash-specific empirical CDF uses the same formulas within its matching history and falls back to the marginal CDF when the prompt has no matching key. The interface also admits prompt-class keys based on a prompt-length quantile bucket. The ExactMatch rows in the numerical tables use the prompt-hash implementation. Their benefits concern recurrence in the evaluated workload, rather than an assumption that every new prompt has a matching history. Inference cost remains $O(\log n)$ per call.

\paragraph{Learned alternative.}
A learned alternative uses a classifier and regressor with prompt size, current age, exponentially weighted moving averages of observed output lengths, and rolling statistics of past inputs and outputs. Its gradient-boosted trees have $200$ estimators of depth $3$. The classification target is $\mathbf 1\{r_i\le H\}$; the regression target is $r_i$ on finish-positive samples. These are causal features. The reported router and predictor-accuracy tables use Survival and ExactMatch; no additional learned-model performance result is inferred from this specification.

\subsection{Training and refresh procedure}
\label{app:predictor-training}

For the empirical estimators, preparing the history requires sorting and takes $O(n\log n)$. For the learned alternative, each historical request $(s_j,o_j)$ supplies samples at ages $T=0,\Delta T,2\Delta T,\ldots<o_j$. Each sample contains the causal features, the termination label $\mathbf 1\{o_j\le T+H\}$, and remaining length $o_j-T$. The classifier uses all samples and the regressor uses the finish-positive subset. The refresh interval is $\Delta T=H$. Training and evaluation segments are disjoint in time. The reported evaluation measures include ROC-AUC for classification and conditional MAE for the finish-positive regression task; Brier score is also specified as a calibration measure. Periodic re-fitting accommodates distribution changes.

\subsection{Refresh safeguards}
\label{app:predictor-refresh}

Predictions are refreshed every $\Delta T=H$ generated tokens. At each refresh the two estimates are recomputed and combined through Equation~\eqref{eq:composite}. Between refreshes the cached cutoff is decremented by one per generated token. The refresh update is accepted when $\widehat p_{\mathrm{fin}}^{(i)}\ge0.5$; otherwise the cutoff is reset to $H$. This is a confidence gate intended to avoid treating a weak short-tail estimate as an imminent departure. While a request remains active, its cutoff has a floor of $1$. Reaching the floor between scheduled refreshes triggers an immediate refresh, preventing an expired prediction from silently removing an active request from the projected load.

\subsection{Prediction computation time}
\label{app:predictor-latency}

Reported per-call cost is approximately $1\,\mu\mathrm{s}$ for the empirical Survival and memorization estimators and approximately $1$ ms for the learned alternative, including its feature-processing cost. Learned-model refreshes at the same $\Delta T$ boundary are batched. The two empirical estimators are also reported as having per-call cost below $10$ ms. Refreshes run on the proxy's token-telemetry path and are amortized across generated tokens. Dispatcher latency, including its measured stage contributions, is reported separately in Appendix~\ref{subsec:exp_overhead}.

\subsection{Offline predictor accuracy}
\label{subsec:exp_predictor}

Table~\ref{tab:predictor_accuracy} reports the offline results for the two deployed estimators on both workloads. ROC-AUC measures ranking separation for the termination event; conditional MAE measures remaining-length error among requests that finish within the horizon. On Azure-2024, both estimators have conditional MAE of $5.4$ tokens. On 30002, prompt-specific history improves AUC from $0.700$ to $0.974$ and conditional MAE from $20.1$ to $2.9$ tokens. These prediction measures and the operating measures in Table~\ref{tab:main_combined} describe different aspects of the deployed routers.

\begin{table}[htbp]
\centering
\caption{\textbf{Offline predictor accuracy.} Stage~1: ROC-AUC on $\{r_i(k) \le H\}$. Stage~2: conditional MAE in tokens on the finish-positive subsample.}
\label{tab:predictor_accuracy}
\begin{tabular}{llcc}
\toprule
Dataset & Realization & Stage-1 AUC $\uparrow$ & Stage-2 MAE $\downarrow$ \\
\midrule
Azure-2024 & Empirical-survival       & 0.993 &  5.4 \\
Azure-2024 & Per-prompt memorization  & 0.995 &  5.4 \\
\midrule
30002      & Empirical-survival       & 0.700 & 20.1 \\
30002      & Per-prompt memorization  & 0.974 &  2.9 \\
\bottomrule
\end{tabular}
\end{table}

\section{Measurement details and dispatch overhead}
\label{app:experimental-measurement}

\subsection{Measurement procedure}
\label{app:engineering-instrumentation}

The worker-load traces use an optional patch to the vLLM engine-step loop. It timestamps each step and records \texttt{step\_cost\_ms}, \texttt{running\_reqs}, \texttt{total\_tokens}, \texttt{waiting\_reqs}, and \texttt{kv\_cache\_usage} in a file for each DP rank. The rank is identified from process metadata, with launch arguments and process identifiers as fallbacks. The instrumentation is used for experimental measurement and is disabled in production. The dispatcher study separately records per-loop and per-stage latency in the proxy. Additional proxy and engine implementation details appear in the system paper \citep{BalanceRoute2026}.

All empirical imbalance values below are trace means of the aggregate load imbalance $I(k)$ defined in Section~\ref{subsec:exp_setup}. The tables and plots use the measurements reported in that evaluation.

\subsection{Dispatch overhead}
\label{subsec:exp_overhead}

Figure~\ref{fig:overhead} reports dispatcher latency on DeepSeek-V3 $671$B and Qwen3-30B-A3B, both with $G=8$ and $R_{\max}=4$. On DeepSeek-V3 the median, mean, and $99$th percentile are $1.20$, $1.34$, and $2.77$ ms. On Qwen3 they are $1.03$, $1.15$, and $2.52$ ms. The distributions differ by approximately $17\%$ at the median and $10\%$ at the $99$th percentile. An engine step is approximately $60$ ms, making the dispatcher approximately $50$ times smaller at the median and $22$ times smaller at the $99$th percentile. A small number of Stage~2 observations exceed $10$ ms but remain below one engine step.

Stage~2 contributes approximately $69\%$ of total dispatcher time, with mean $924\,\mu\mathrm{s}$ on DeepSeek-V3 and $794\,\mu\mathrm{s}$ on Qwen3. Stage~1 has median $0.8$ to $1.2\,\mu\mathrm{s}$ and $99$th percentiles of $0.54$ ms and $1.37$ ms, respectively. Preparation has a $99$th-percentile-to-median ratio of approximately $1.8$. These measurements describe the implemented dispatcher and do not measure the cost of solving the ideal joint integer program.

\begin{figure}[t]
\centering
\includegraphics[width=\linewidth]{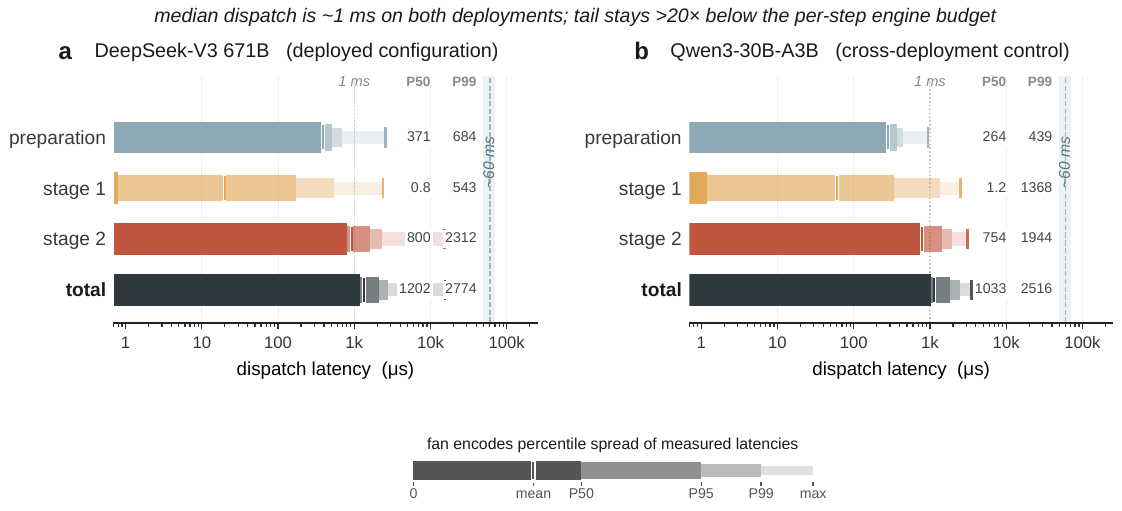}
\caption{\textbf{Per-tick dispatch overhead, two deployments.} $G{=}8$, $R_{\max}{=}4$. \textbf{(a)} DeepSeek-V3 671B. \textbf{(b)} Qwen3-30B-A3B (cross-deployment control). Each row displays latency percentiles ($0\to P_{50}$ solid, then to $P_{95}$, $P_{99}$, max); the white tick marks the mean. Shaded band ($\sim\!60$\,ms) marks the per-step engine latency on DeepSeek-V3.}
\label{fig:overhead}
\end{figure}

\section{Supplementary Survival results on 30002}
\label{app:survival-30002}

The Survival predictor uses the marginal output-length CDF. Figure~\ref{fig:kv-workload-30002-survival} shows its worker loads over the same $1{,}500$-step segment as Figure~\ref{fig:kv-workload-grid}. Its mean imbalance is $40{,}757$ tokens, compared with $36{,}277$ for ExactMatch, $51{,}927$ for BR-0, and $215{,}110$ for JSQ. The Survival imbalance is approximately $1.1$ times the ExactMatch imbalance; BR-0's imbalance is approximately $1.3$ times the Survival imbalance.

\begin{figure}[htbp]
\centering
\includegraphics[width=0.5\linewidth]{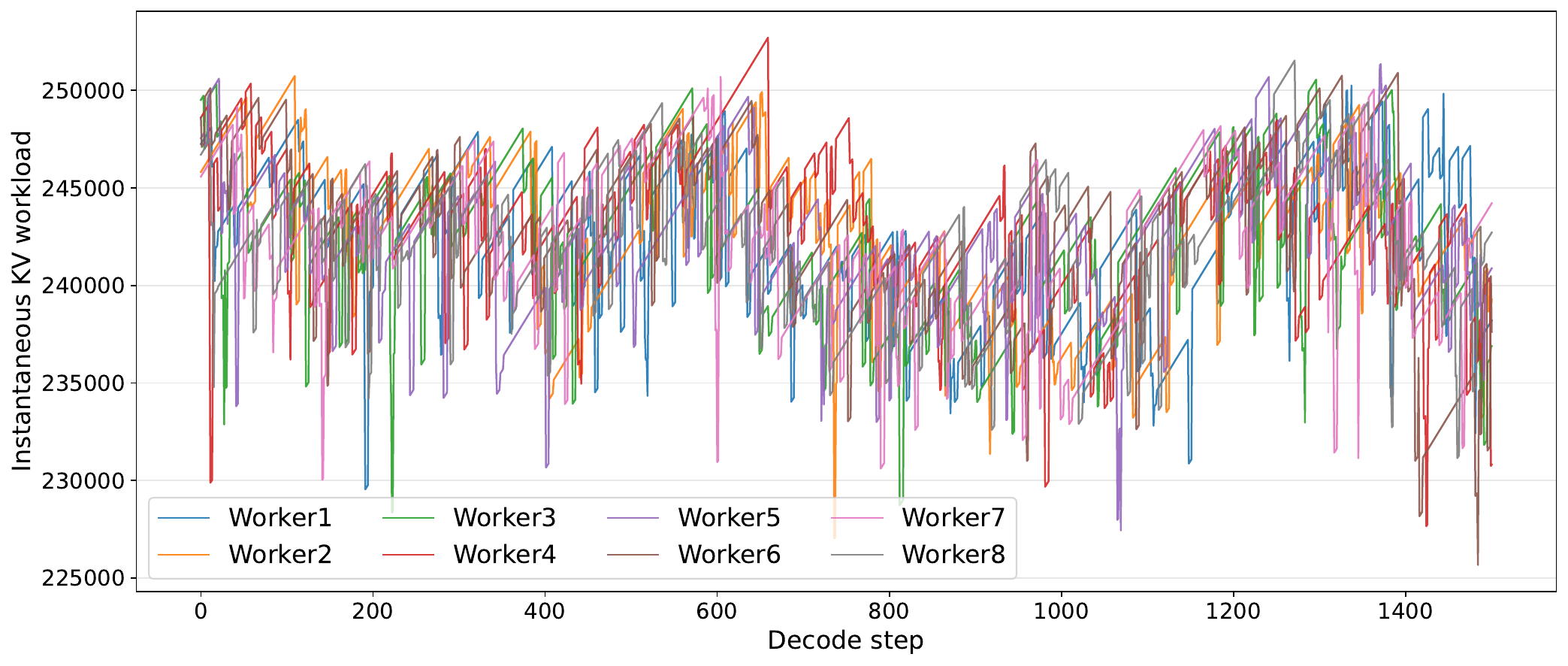}
\caption{\textbf{Per-worker instantaneous KV-cache workload on 30002 under BR-H with the Survival predictor}, $H{=}80$, $G{=}8$ decode workers, same $1{,}500$-step trace segment as Figure~\ref{fig:kv-workload-grid}. Trace-mean imbalance is $40{,}757$ tokens. Reference mean imbalances are $51{,}927$ for BR-0, $215{,}110$ for JSQ, and $36{,}277$ for ExactMatch. The predictor refresh and confidence gate are described in Appendix~\ref{app:predictor-refresh}.}
\label{fig:kv-workload-30002-survival}
\end{figure}

Relative to BR-0, Survival reduces mean imbalance by $21.5\%$ ($51{,}927$ to $40{,}757$), reduces TPOT P95 by approximately $0.7\%$ ($79.32$ to $78.75$ ms), and increases throughput by $3.9\%$ ($942.5$ to $979.2$ tok/s). Relative to JSQ, the reported imbalance reduction is approximately $5.3$ times and the throughput increase approximately $16.2\%$. Survival's imbalance is approximately $12\%$ above ExactMatch, while their throughputs are close. These observations show that Survival improves on BR-0 for this trace despite its weaker offline prediction accuracy. The confidence gate is described in Appendix~\ref{app:predictor-refresh}; the data do not establish a general guarantee that poorer predictions cannot reduce routing performance.

\section{Additional worker-load panels}
\label{app:full-kv-grids}

The following figures show the additional worker-load panels. All $G=8$ panels show the same $1{,}500$-step segment of the relevant trace. Their annotations give rounded mean imbalance values of $I$.

\begin{figure}[htbp]
\centering
\begin{subfigure}[t]{0.48\linewidth}
    \centering
    \includegraphics[width=\linewidth]{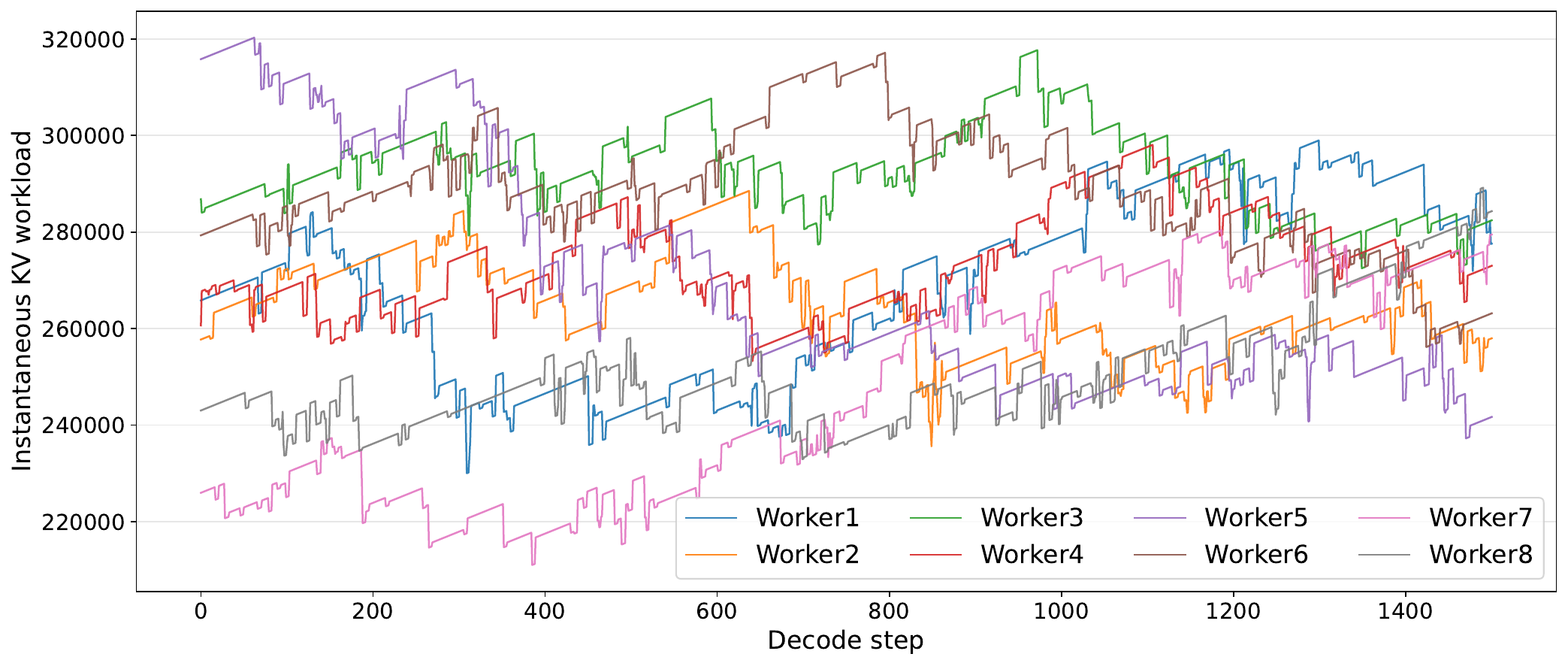}
    \caption*{\centering\footnotesize Round-Robin \\ \textit{mean $I$\ 255k}}
\end{subfigure}\hfill
\begin{subfigure}[t]{0.48\linewidth}
    \centering
    \includegraphics[width=\linewidth]{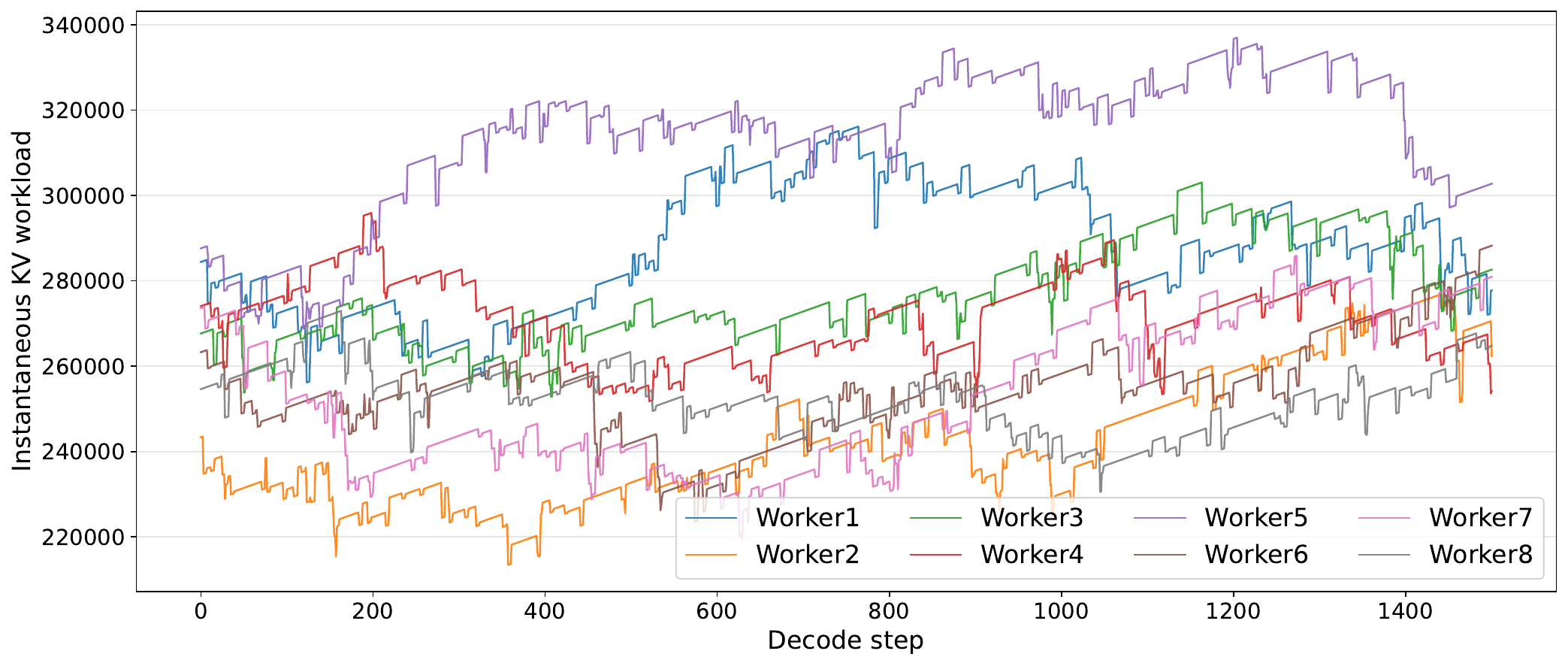}
    \caption*{\centering\footnotesize P2C \\ \textit{mean $I$\ 263k}}
\end{subfigure}
 
\vspace{0.7em}
\begin{subfigure}[t]{0.48\linewidth}
    \centering
    \includegraphics[width=\linewidth]{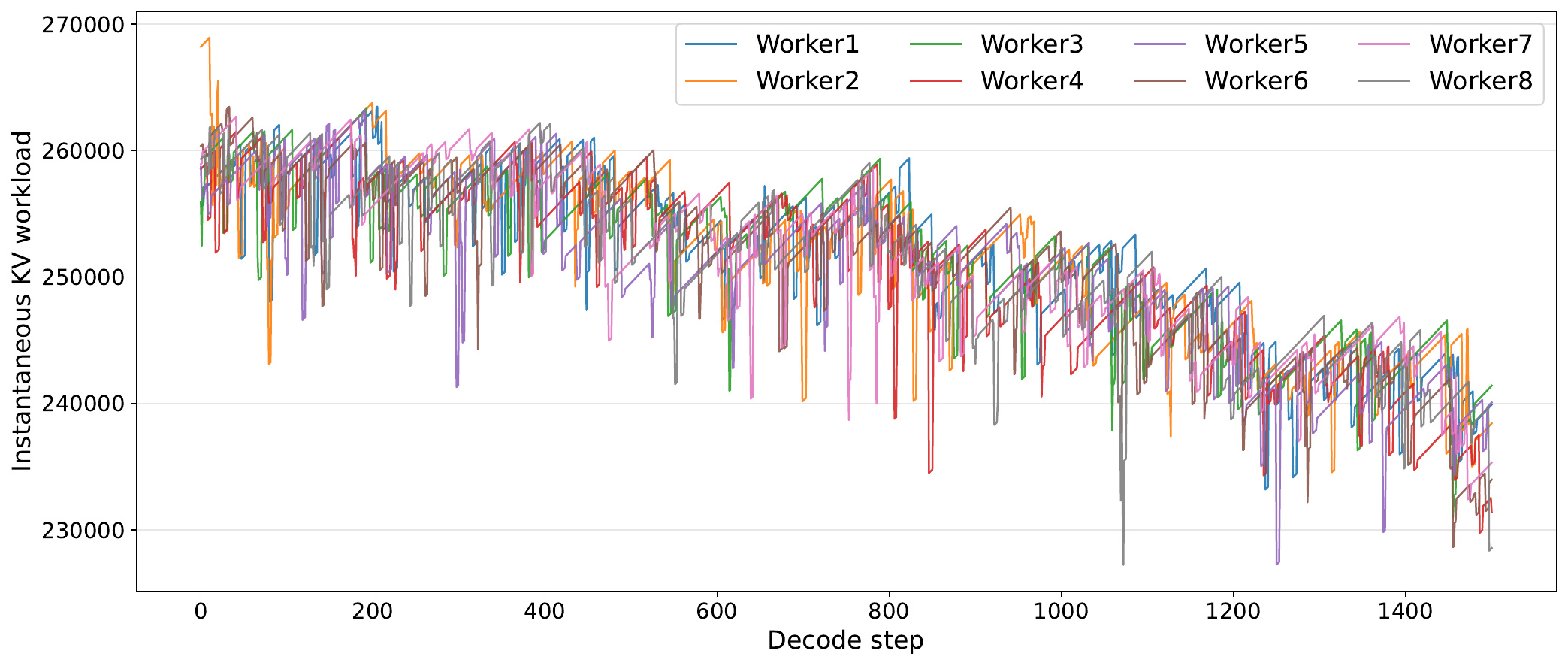}
    \caption*{\centering\footnotesize \textbf{BR-H} oracle, $(\beta,\gamma){=}(43, 0.86)$ \\ \textit{mean $I$\ 24.0k}}
\end{subfigure}\hfill
\begin{subfigure}[t]{0.48\linewidth}
    \centering
    \includegraphics[width=\linewidth]{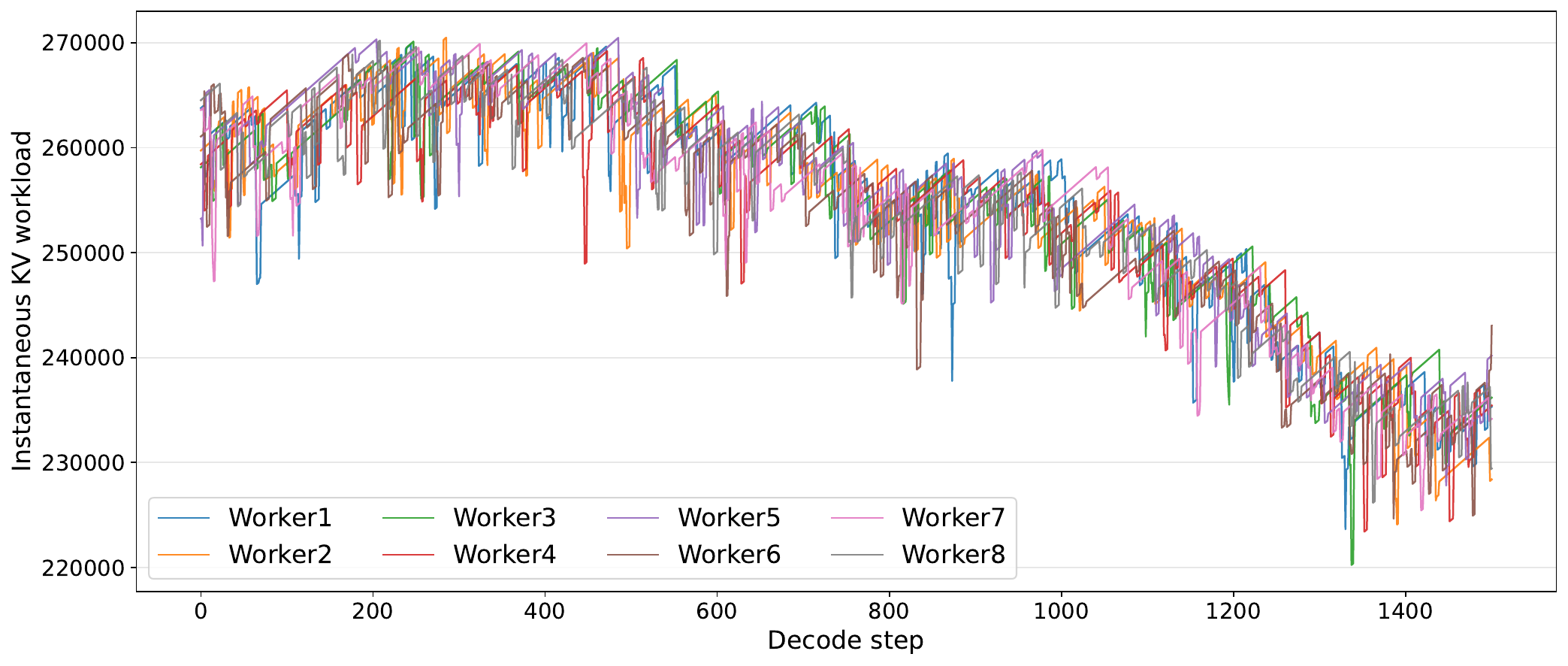}
    \caption*{\centering\footnotesize \textbf{BR-H} oracle, $(\beta,\gamma){=}(14.67, 0.64)$ \\ \textit{mean $I$\ 23.6k}}
\end{subfigure}
 
\caption{\textbf{Remaining per-worker panels for 30002, $G{=}8$ (companion to Figure~\ref{fig:kv-workload-grid}).} Top: RR and P2C baselines. Bottom: two BR-H oracle operating points.}
\label{fig:kv-workload-grid-app}
\end{figure}

\begin{figure}[htbp]
\centering
\begin{subfigure}[t]{0.48\linewidth}
    \centering
    \includegraphics[width=\linewidth]{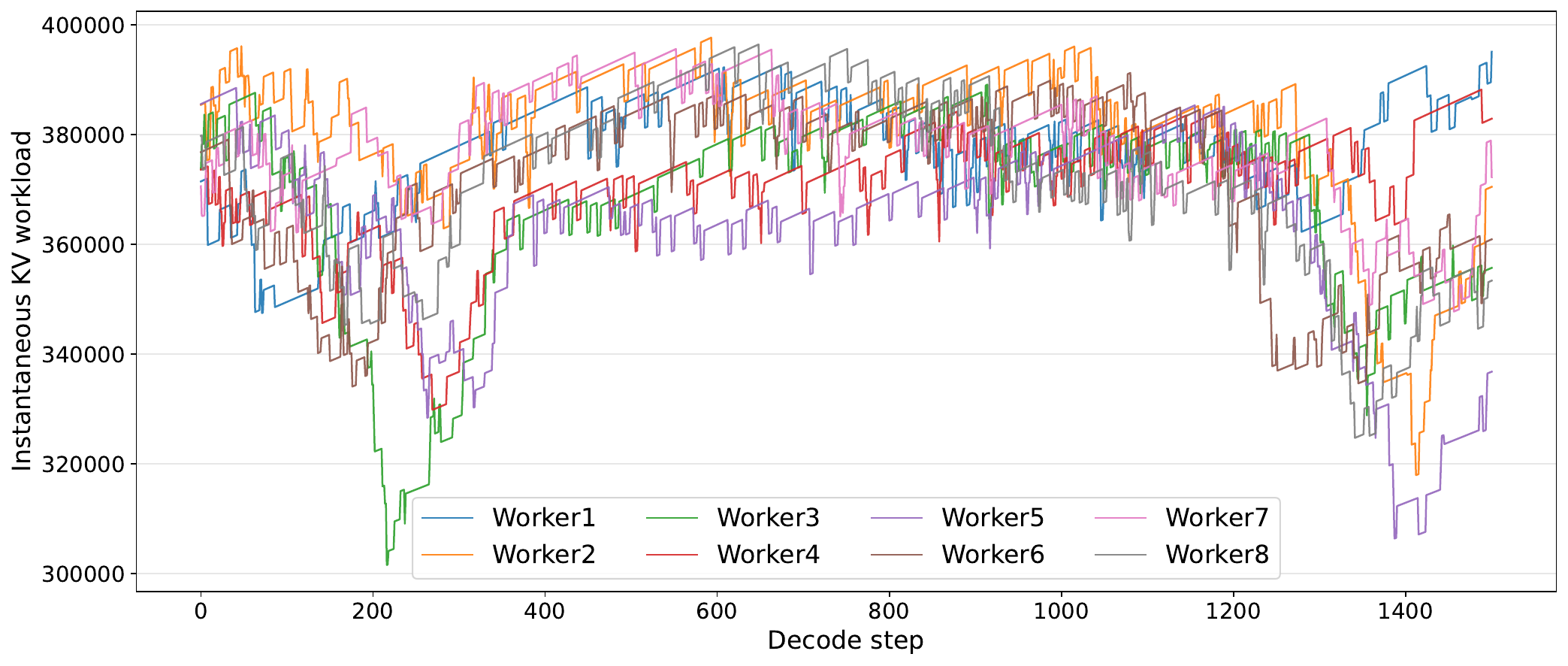}
    \caption*{\centering\footnotesize Random \\ \textit{mean $I$\ 147k}}
\end{subfigure}\hfill
\begin{subfigure}[t]{0.48\linewidth}
    \centering
    \includegraphics[width=\linewidth]{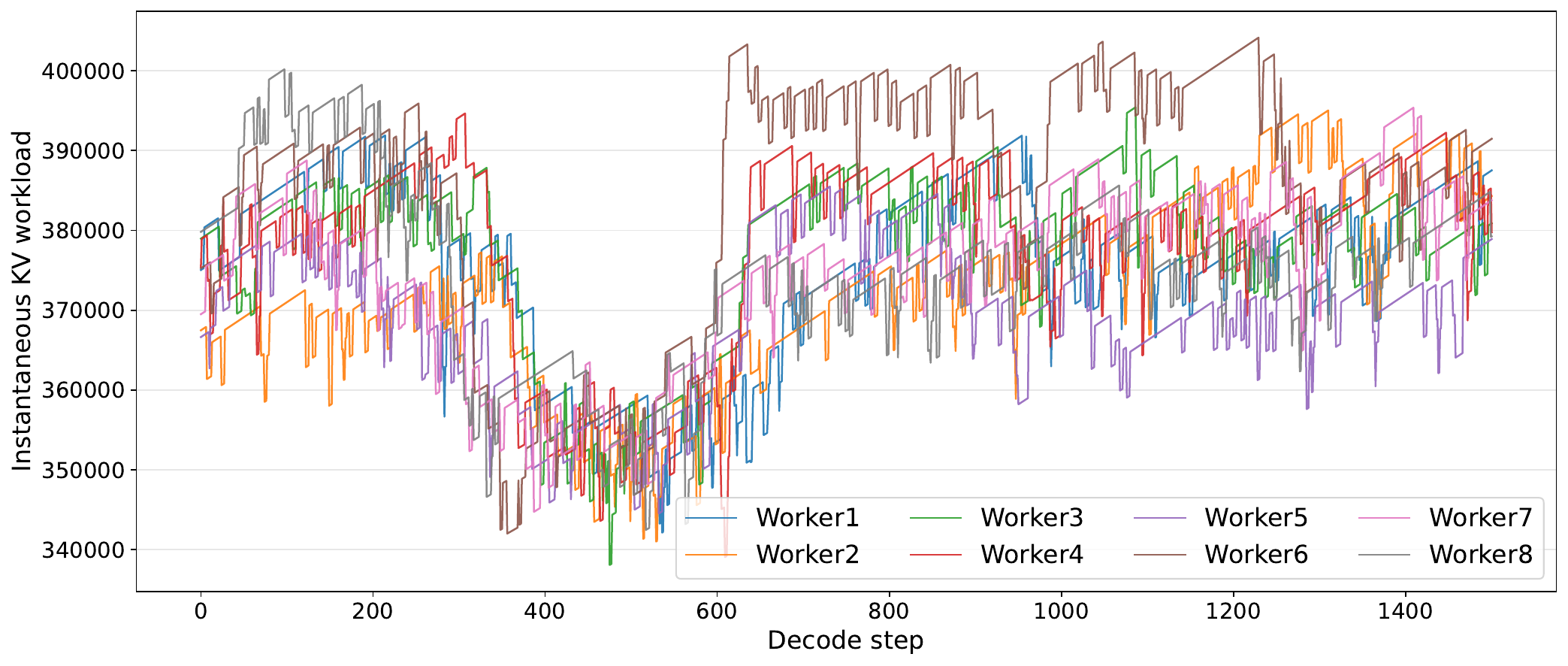}
    \caption*{\centering\footnotesize JSQ \\ \textit{mean $I$\ 105k}}
\end{subfigure}

\vspace{0.7em}
\begin{subfigure}[t]{0.48\linewidth}
    \centering
    \includegraphics[width=\linewidth]{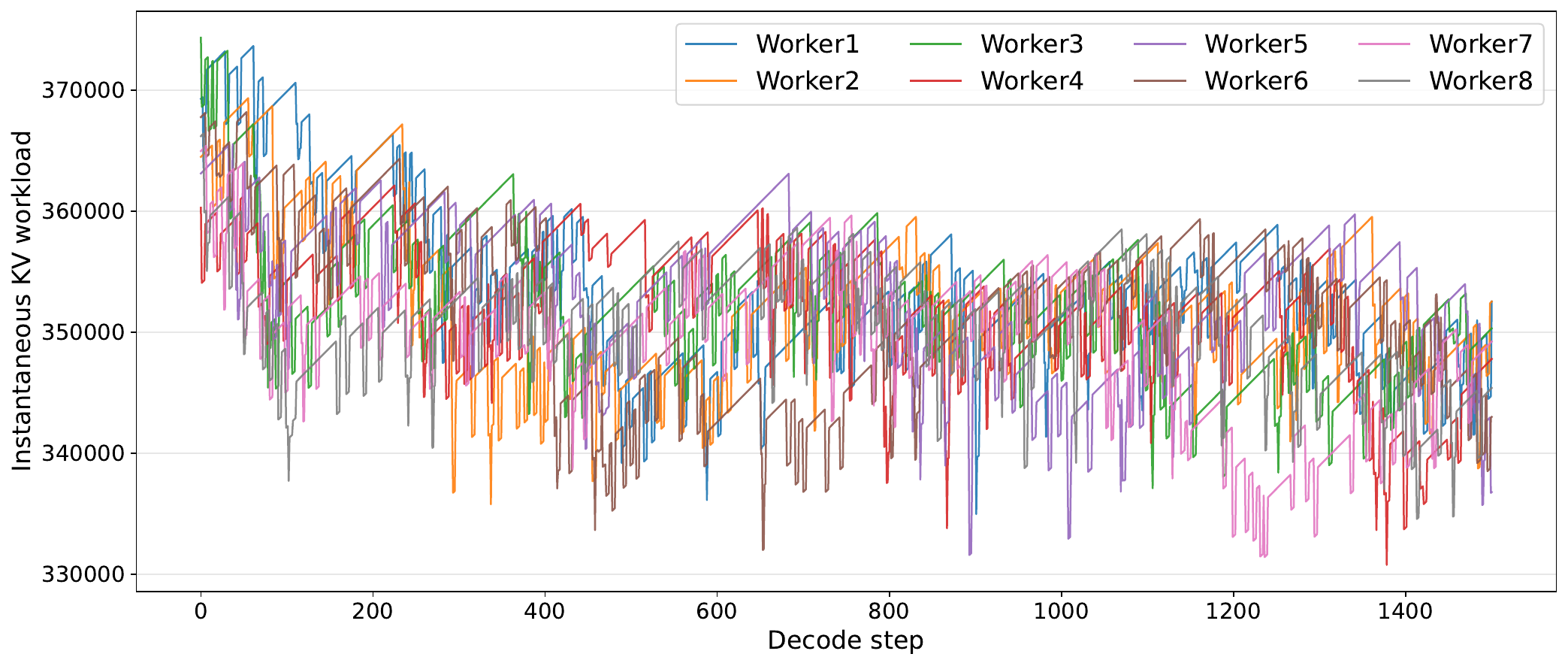}
    \caption*{\centering\footnotesize \textbf{BR-0} (prediction-free) \\ \textit{mean $I$\ 54k}}
\end{subfigure}\hfill
\begin{subfigure}[t]{0.48\linewidth}
    \centering
    \includegraphics[width=\linewidth]{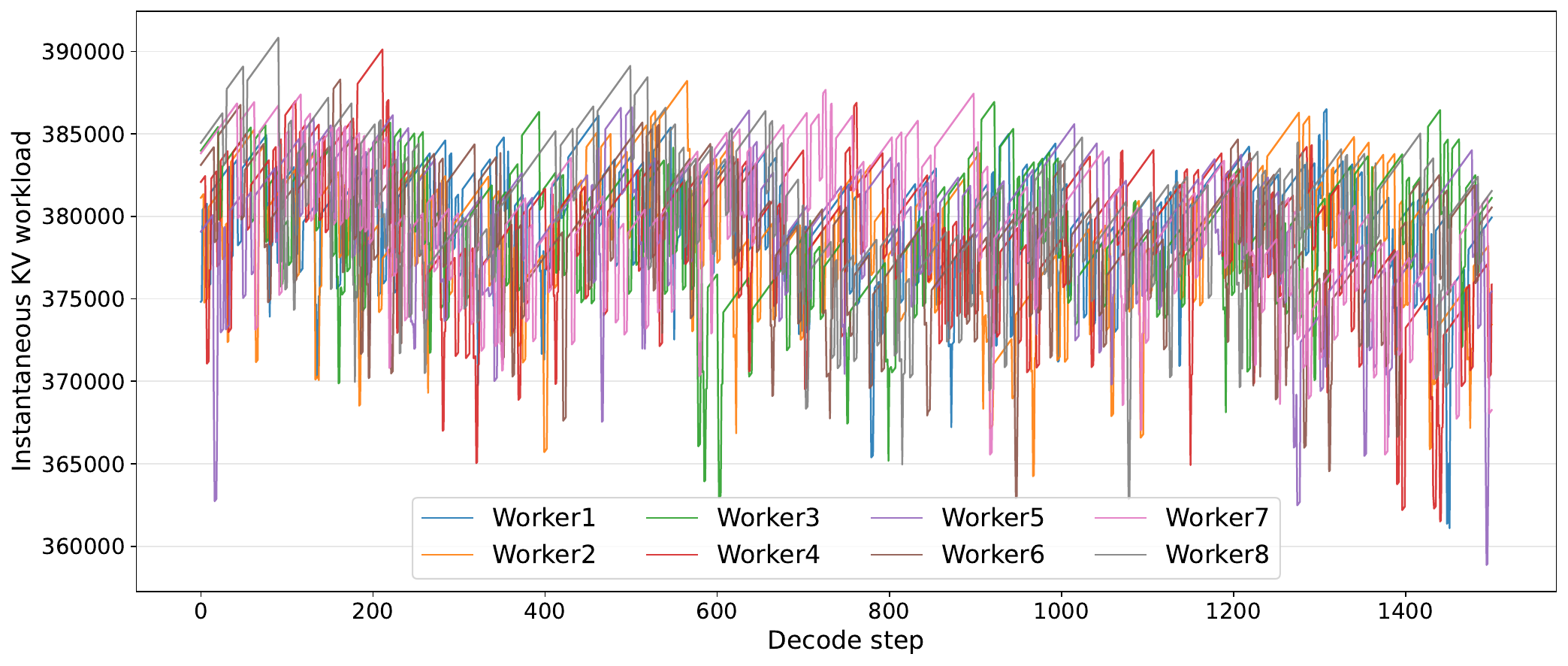}
    \caption*{\centering\footnotesize \textbf{BR-H} predicted, ExactMatch \\ \textit{mean $I$\ 38.5k}}
\end{subfigure}

\caption{\textbf{Per-worker KV-cache workload on Azure-2024, $G{=}8$.} Same layout as Figure~\ref{fig:kv-workload-grid}. }
\label{fig:kv-workload-azure-grid-paper}
\end{figure}

\begin{figure}[htbp]
\centering
\begin{subfigure}[t]{0.48\linewidth}
    \centering
    \includegraphics[width=\linewidth]{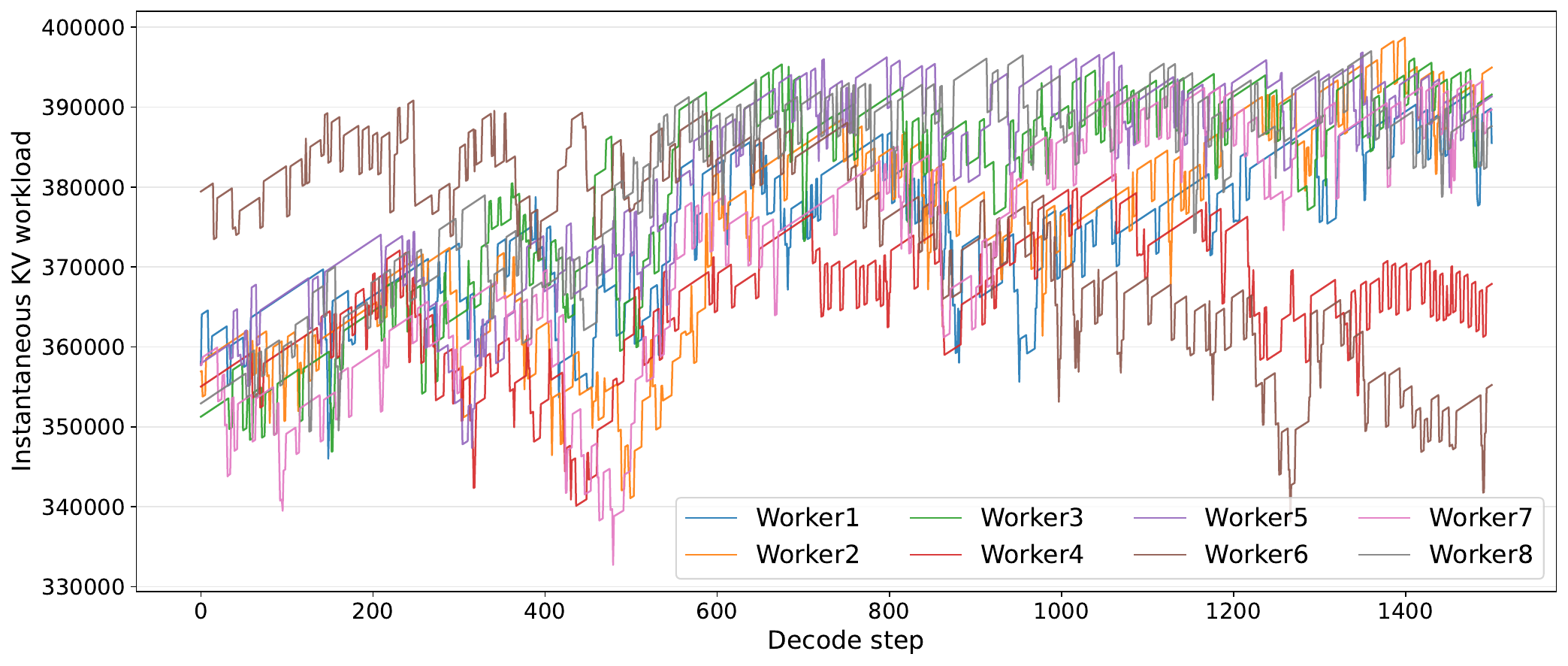}
    \caption*{\centering\footnotesize Round-Robin \\ \textit{mean $I$\ 110k}}
\end{subfigure}\hfill
\begin{subfigure}[t]{0.48\linewidth}
    \centering
    \includegraphics[width=\linewidth]{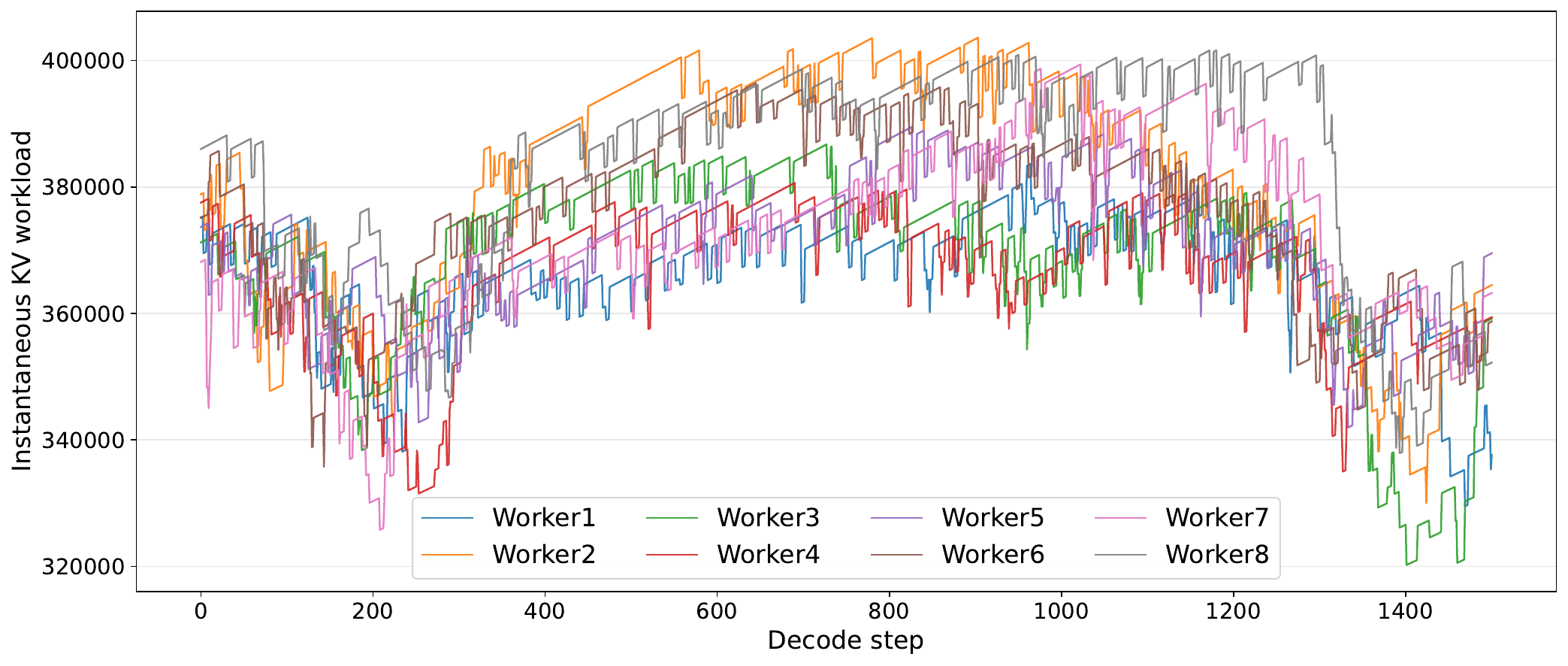}
    \caption*{\centering\footnotesize P2C \\ \textit{mean $I$\ 109k}}
\end{subfigure}
 
\vspace{0.7em}
\begin{subfigure}[t]{0.48\linewidth}
    \centering
    \includegraphics[width=\linewidth]{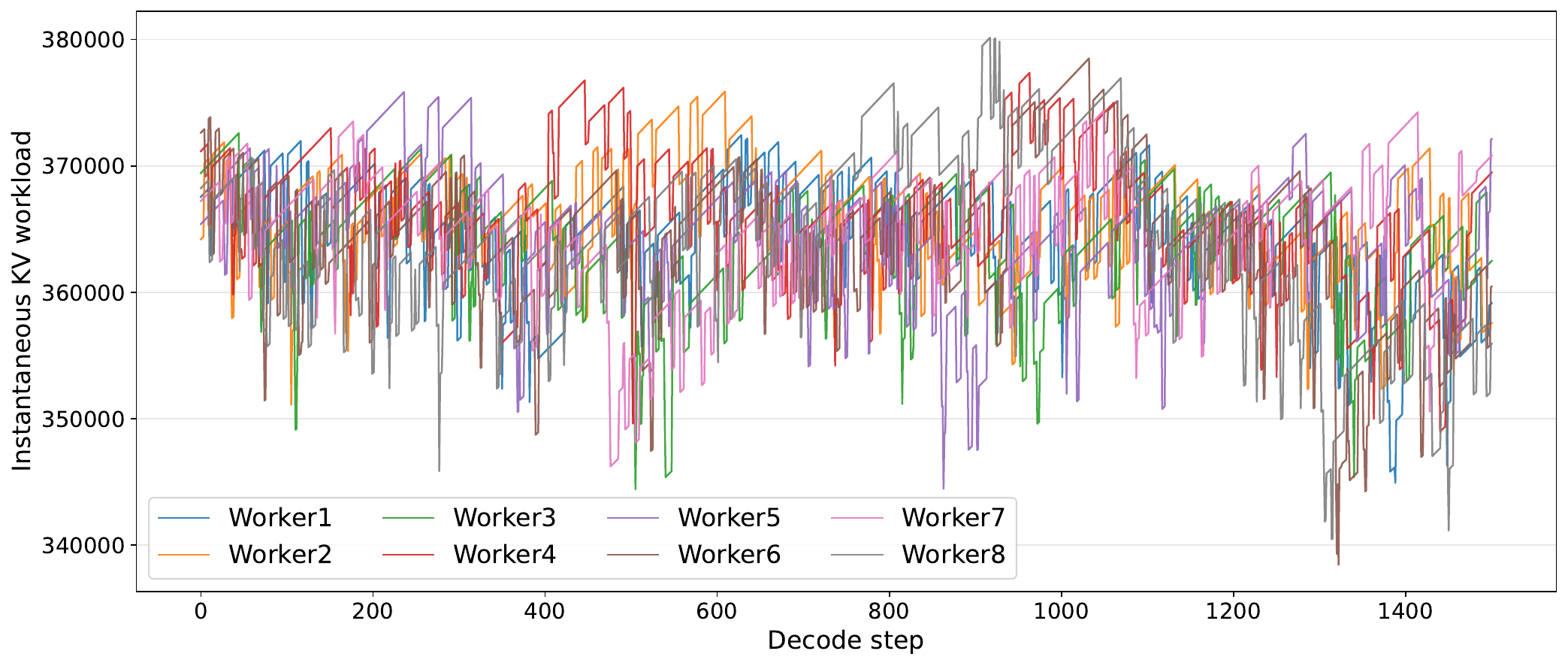}
    \caption*{\centering\footnotesize \textbf{BR-H} oracle, $(\beta,\gamma){=}(43, 0.86)$ \\ \textit{mean $I$\ 37.8k}}
\end{subfigure}\hfill
\begin{subfigure}[t]{0.48\linewidth}
    \centering
    \includegraphics[width=\linewidth]{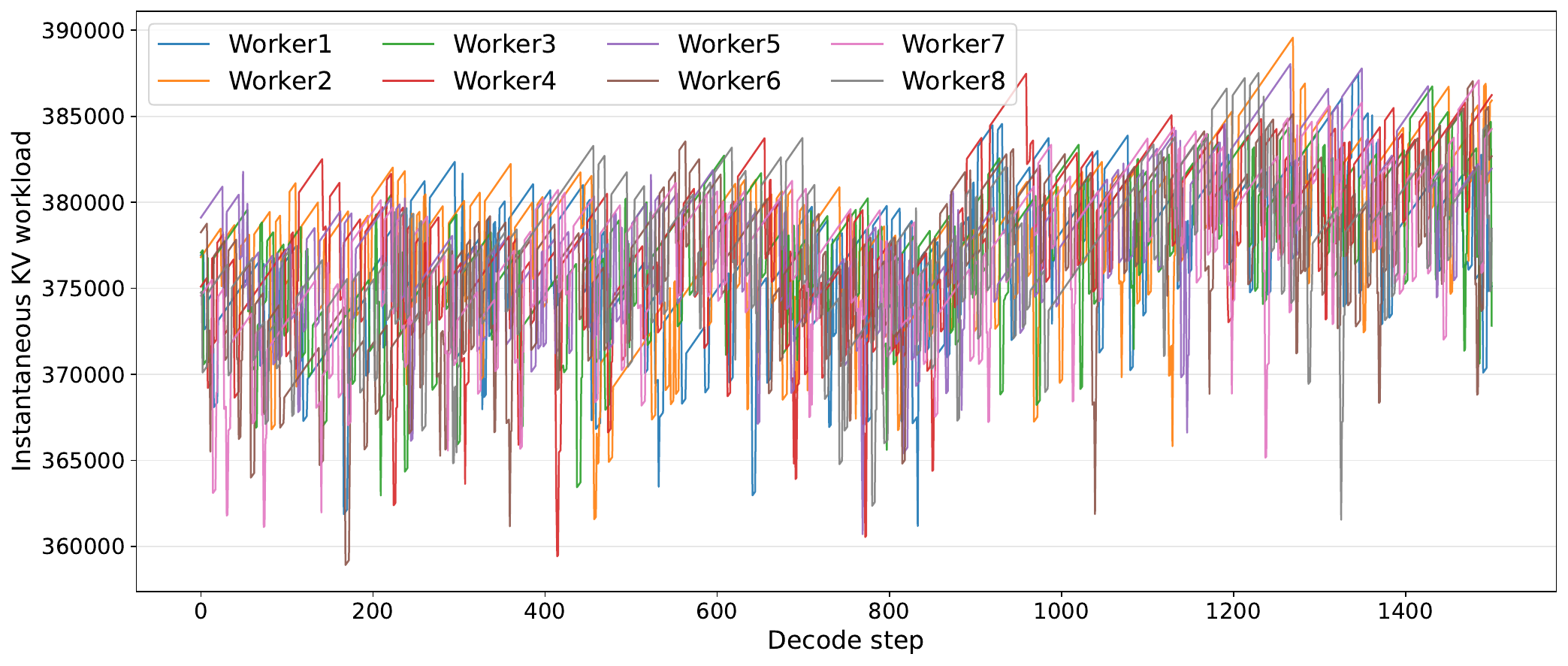}
    \caption*{\centering\footnotesize \textbf{BR-H} oracle, $(\beta,\gamma){=}(14.67, 0.64)$ \\ \textit{mean $I$\ 35.3k}}
\end{subfigure}
 
\caption{\textbf{Remaining per-worker panels for Azure-2024, $G{=}8$ (companion to Figure~\ref{fig:kv-workload-azure-grid-paper}).} Top: RR and P2C baselines. Bottom: two BR-H oracle operating points.}
\label{fig:kv-workload-azure-grid-app}
\end{figure}

\section{Additional scaling results}
\label{app:scaling}

Tables~\ref{tab:scaling_g4_app} and~\ref{tab:scaling_g16_app} give the individual $G=4$ and $G=16$ results underlying Table~\ref{tab:scaling_summary}. Their accompanying figures show the worker trajectories. BR-H uses oracle predictions in these configurations. Deployed predictor comparisons are available at $G=8$ only, so the scaling results do not quantify how deployed prediction error changes with $G$.

At $G=4$, the baseline mean imbalances are approximately $64$ to $80$ thousand tokens and baseline throughput is approximately $365$ tok/s. BR-0 has imbalance approximately $9.7$ thousand tokens and throughput approximately $407$ tok/s. Oracle BR-H has imbalance approximately $6.6$ thousand tokens and throughput approximately $415$ tok/s, with the lowest TPOT P95 in the table.

\begin{table}[htbp]
\centering
\caption{\textbf{All methods on 30002 at $G{=}4$ (2P1D, 48 NPUs).} Reproduces the $G{=}4$ column of Table~\ref{tab:scaling_summary} as a self-contained companion to Figure~\ref{fig:scaling-traces-g4}. Best in each column in \textbf{bold}.}
\label{tab:scaling_g4_app}
\begin{tabular}{lccc}
\toprule
Method & Mean $I$ $\downarrow$ & TPOT P95 (ms) $\downarrow$ & Throughput (tok/s) $\uparrow$ \\
\midrule
Random       &  80{,}118 & 66.30 & 364.7 \\
Round-Robin  &  68{,}108 & 65.27 & 365.7 \\
P2C          &  64{,}012 & 65.40 & 366.4 \\
JSQ          &  75{,}326 & 66.45 & 365.4 \\
\midrule
BR-0         &   9{,}711 & 63.19 & 407.1 \\
BR-H \emph{oracle}, $H{=}80$ & \textbf{6{,}576} & \textbf{62.52} & \textbf{415.3} \\
\bottomrule
\end{tabular}
\end{table}

\begin{figure}[htbp]
\centering
\begin{subfigure}[t]{0.48\linewidth}
    \centering
    \includegraphics[width=\linewidth]{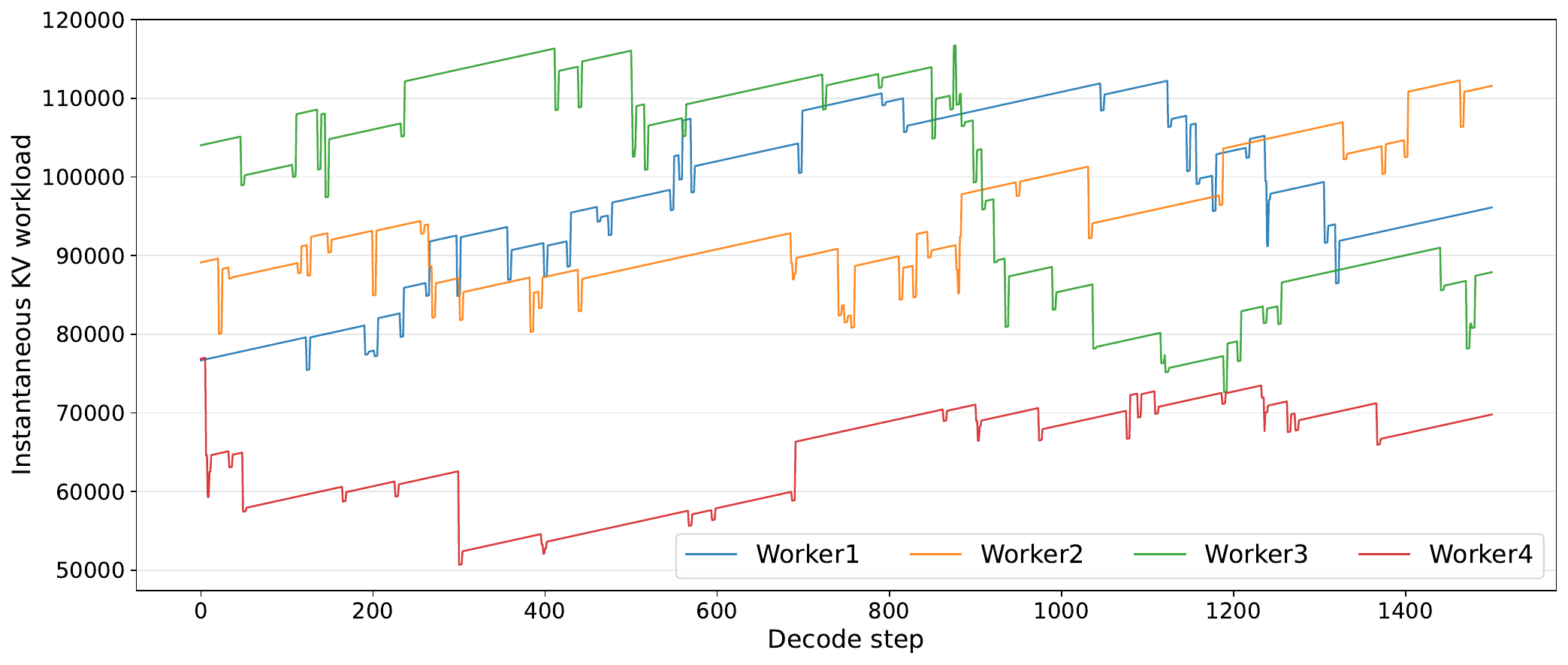}
    \caption*{\centering\footnotesize Random \\ \textit{mean $I$\ 80k}}
\end{subfigure}\hfill
\begin{subfigure}[t]{0.48\linewidth}
    \centering
    \includegraphics[width=\linewidth]{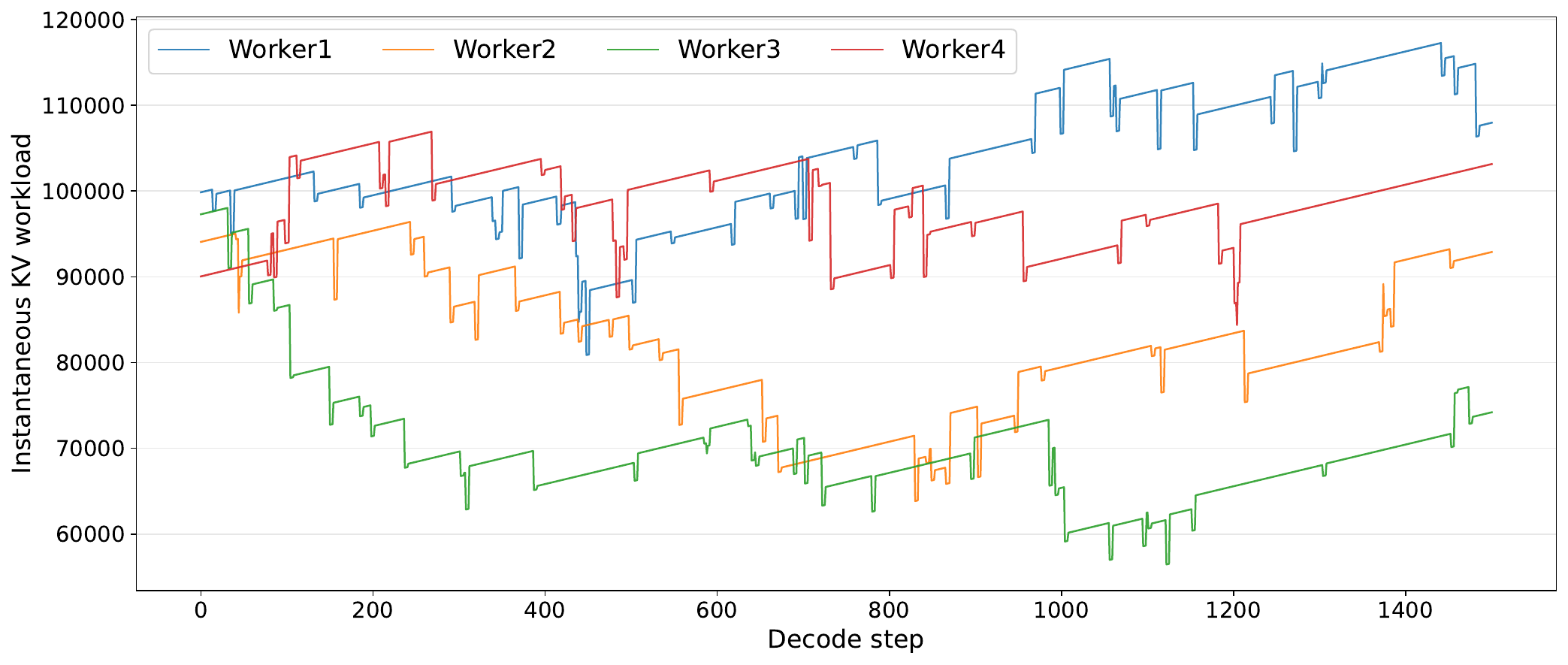}
    \caption*{\centering\footnotesize Round-Robin \\ \textit{mean $I$\ 68k}}
\end{subfigure}\hfill
\begin{subfigure}[t]{0.48\linewidth}
    \centering
    \includegraphics[width=\linewidth]{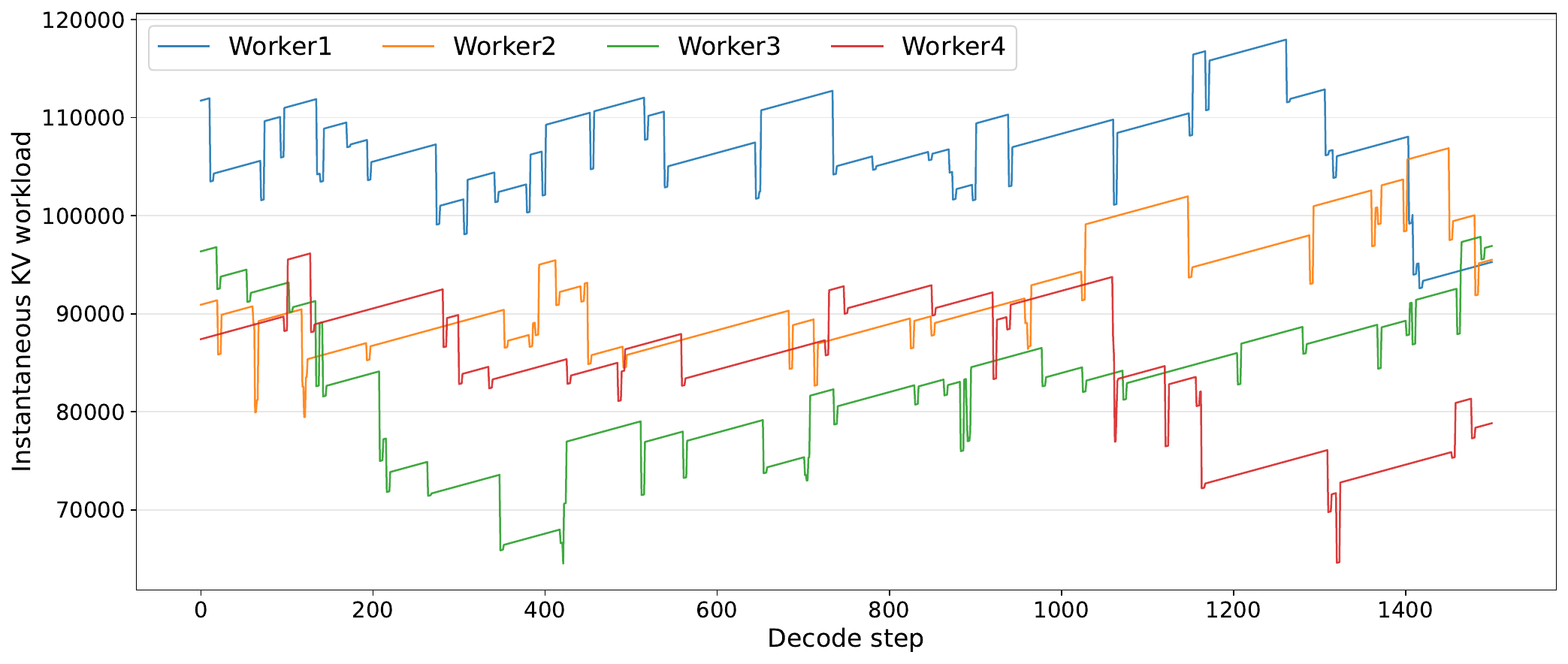}
    \caption*{\centering\footnotesize P2C \\ \textit{mean $I$\ 64k}}
\end{subfigure}\hfill
\begin{subfigure}[t]{0.48\linewidth}
    \centering
    \includegraphics[width=\linewidth]{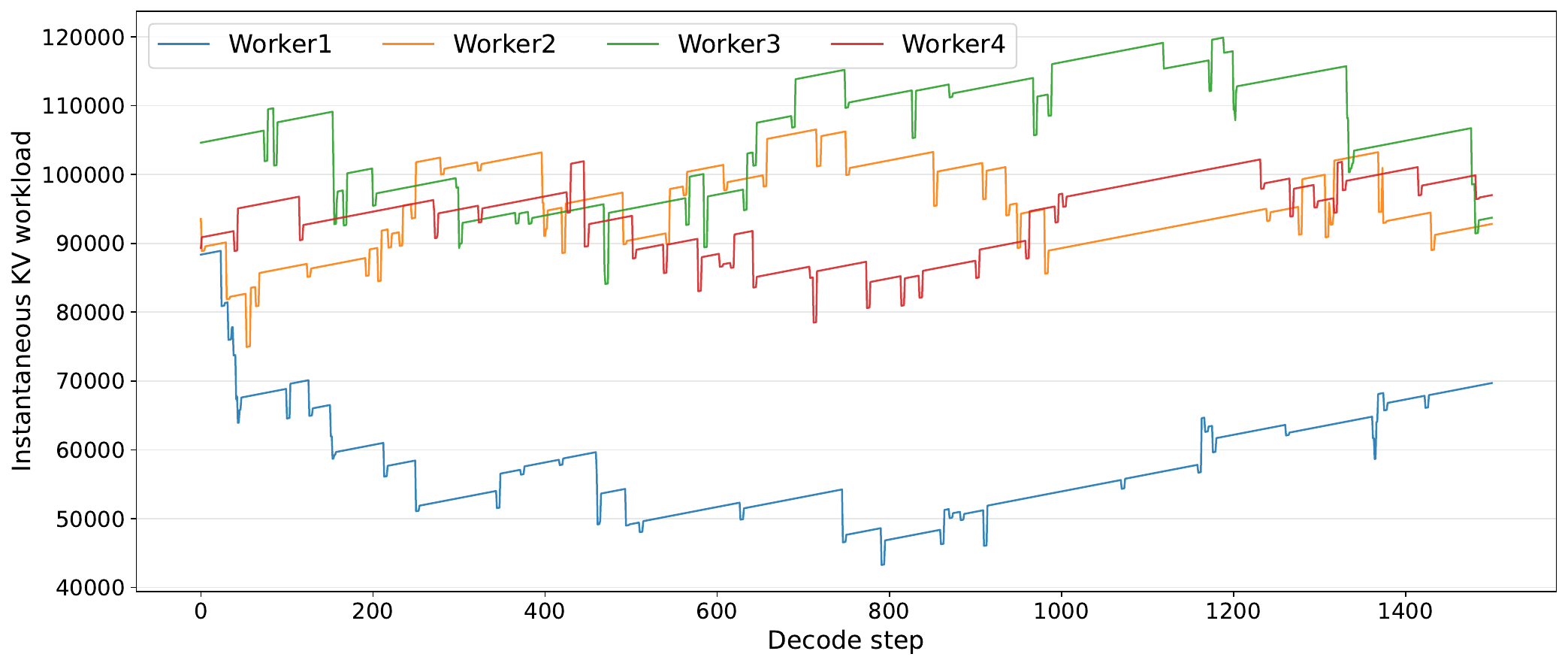}
    \caption*{\centering\footnotesize JSQ \\ \textit{mean $I$\ 75k}}
\end{subfigure}
 
\vspace{0.7em}
 
\begin{subfigure}[t]{0.48\linewidth}
    \centering
    \includegraphics[width=\linewidth]{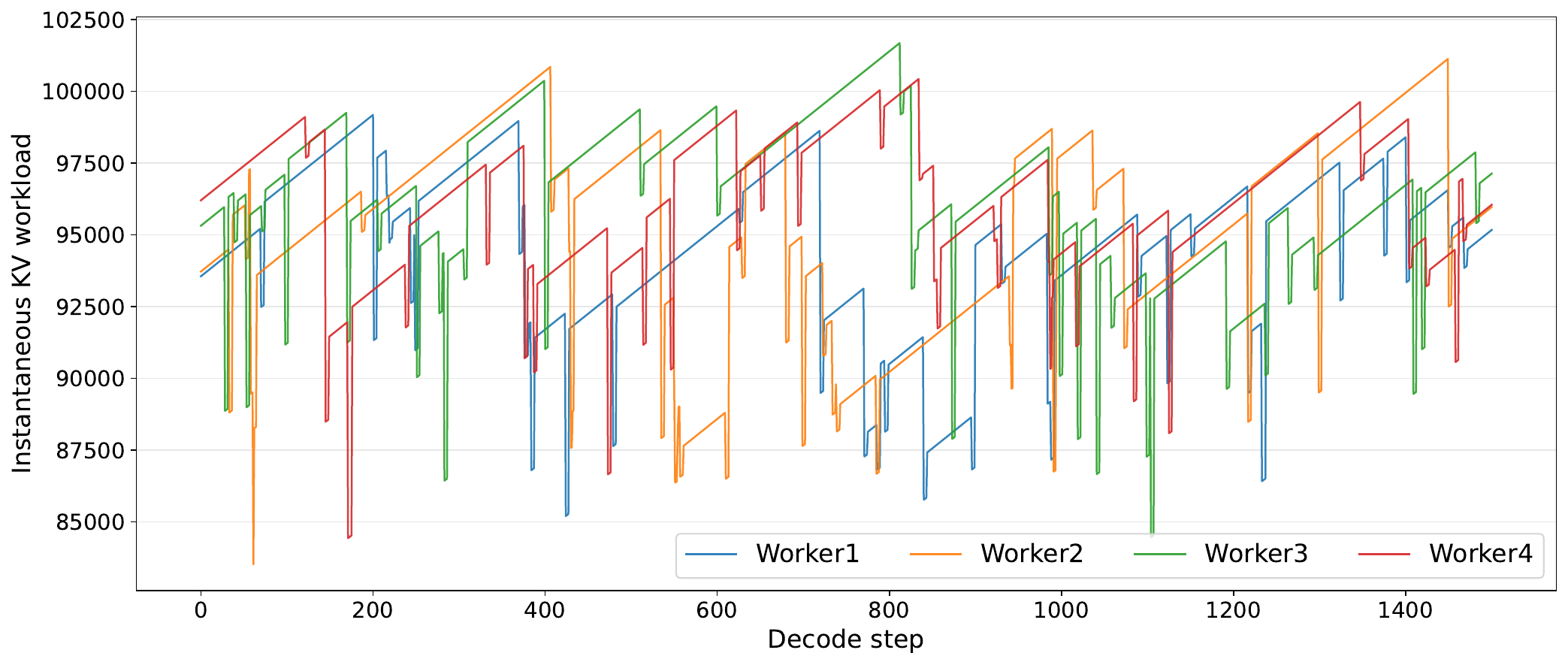}
    \caption*{\centering\footnotesize \textbf{BR-0} (prediction-free) \\ \textit{mean $I$\ 9.7k}}
\end{subfigure}\hfill
\begin{subfigure}[t]{0.48\linewidth}
    \centering
    \includegraphics[width=\linewidth]{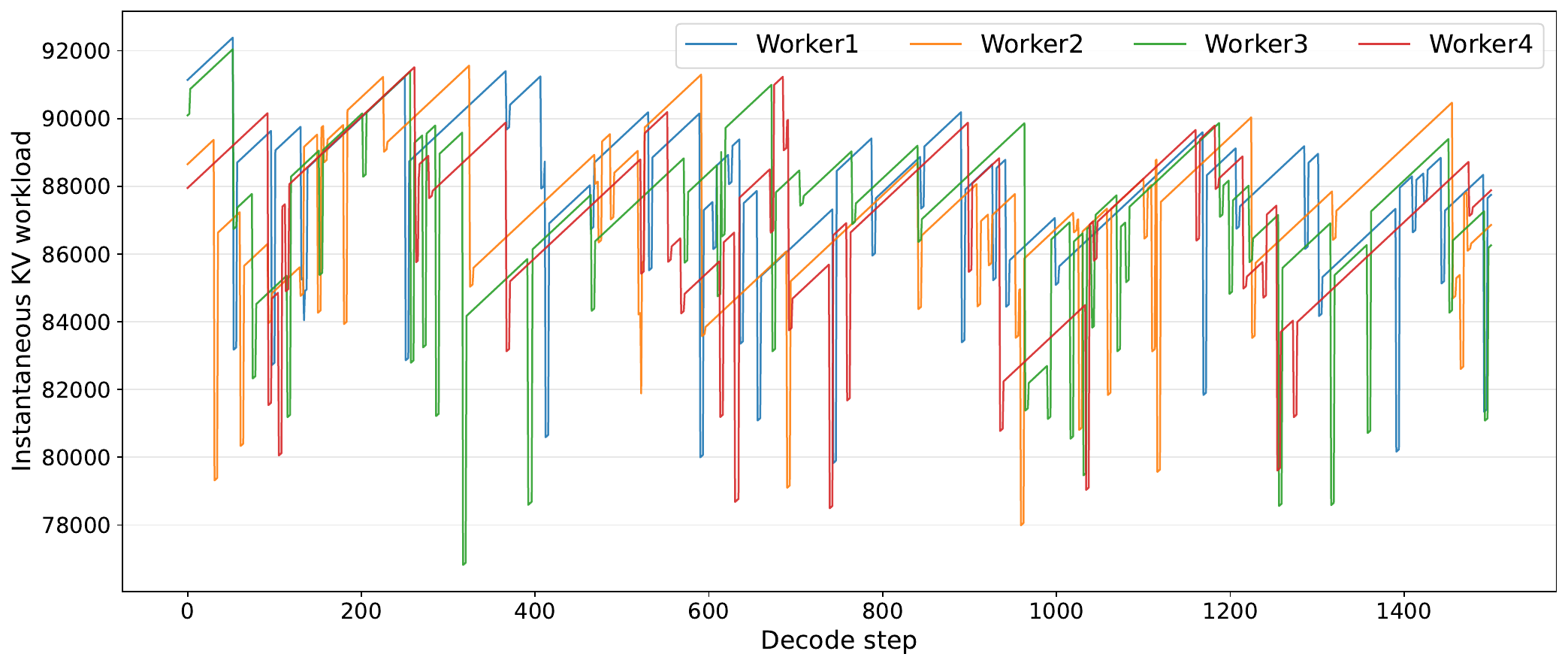}
    \caption*{\centering\footnotesize \textbf{BR-H} (oracle, $H{=}80$) \\ \textit{mean $I$\ 6.6k}}
\end{subfigure}
 
\caption{\textbf{Per-worker KV-cache traces on 30002 at $G{=}4$ (2P1D, 48 NPUs).} The four DP workers' KV-cache footprints over a $1{,}500$-step decode segment under each method; trace-mean imbalance $I$ under each panel matches the $G{=}4$ column of Table~\ref{tab:scaling_summary}.}
\label{fig:scaling-traces-g4}
\end{figure}

At $G=16$, baseline mean imbalances are approximately $616$ to $688$ thousand tokens and throughput approximately $926$ tok/s. BR-0 has imbalance approximately $215$ thousand tokens and throughput approximately $1{,}118$ tok/s. The two oracle BR-H settings have imbalances approximately $117$ to $124$ thousand tokens and throughputs approximately $1{,}247$ to $1{,}251$ tok/s. The setting $(14.67,0.64)$ has lower imbalance and TPOT P95, while $(43,0.86)$ has throughput higher by approximately $4$ tok/s. The main scaling figure uses $(14.67,0.64)$.

\begin{table}[htbp]
\centering
\caption{\textbf{All methods on 30002 at $G{=}16$ (5P1D, 144 NPUs).} Reproduces the $G{=}16$ column of Table~\ref{tab:scaling_summary} as a self-contained companion to Figure~\ref{fig:scaling-traces-g16-main}, with both oracle BR-H operating points reported. Best in each column in \textbf{bold}.}
\label{tab:scaling_g16_app}
\begin{tabular}{lccc}
\toprule
Method & Mean $I$ $\downarrow$ & TPOT P95 (ms) $\downarrow$ & Throughput (tok/s) $\uparrow$ \\
\midrule
Random       & 638{,}062 & 77.80 &  923.6 \\
Round-Robin  & 616{,}089 & 77.91 &  925.3 \\
P2C          & 688{,}283 & 77.46 &  927.3 \\
JSQ          & 675{,}607 & 77.40 &  926.6 \\
\midrule
BR-0                                                            & 215{,}170 & 74.92 & 1{,}117.8 \\
BR-H \emph{oracle}, $(43, 0.86)$              & 123{,}700 & 72.15 & \textbf{1{,}251.1} \\
BR-H \emph{oracle}, $(14.67, 0.64)$ & \textbf{117{,}067} & \textbf{72.04} & 1{,}247.4 \\
\bottomrule
\end{tabular}
\end{table}

\begin{figure}[htbp]
\centering
\begin{subfigure}[t]{0.48\linewidth}
    \centering
    \includegraphics[width=\linewidth]{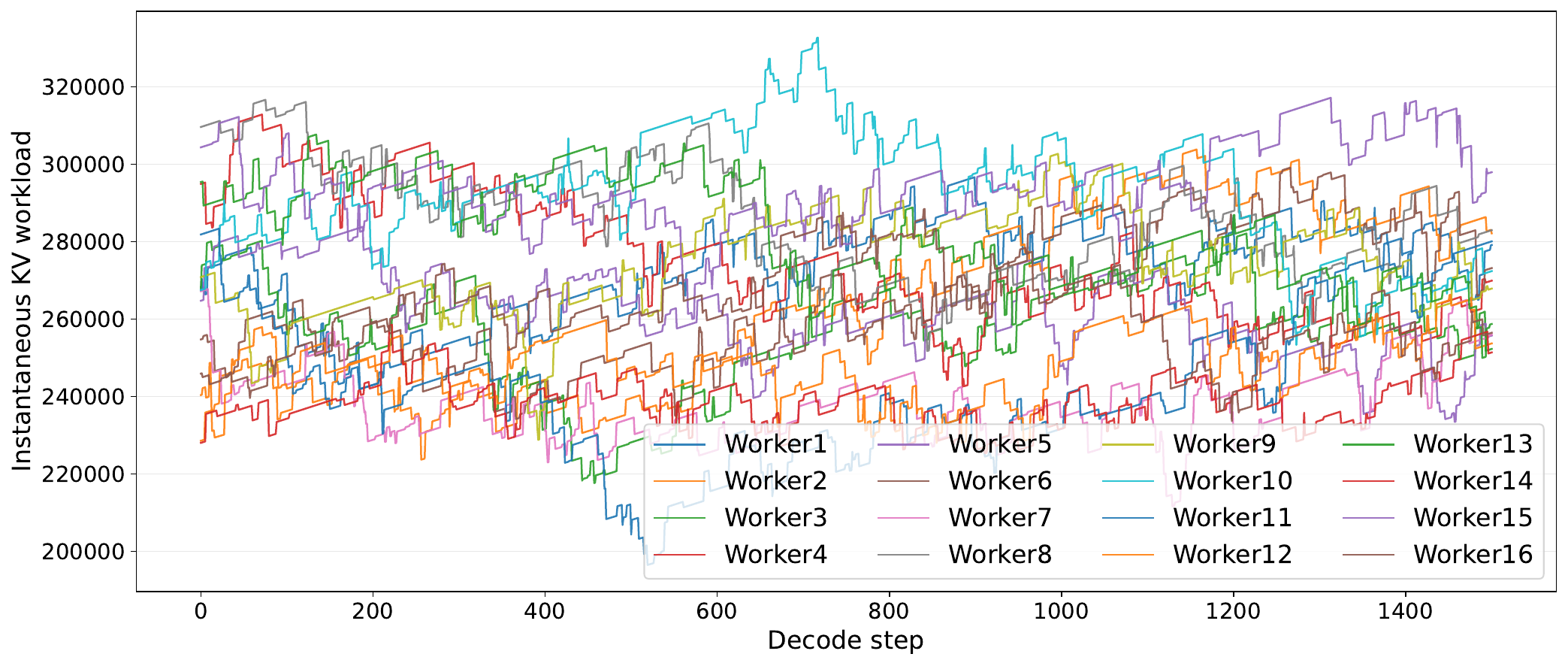}
    \caption*{\centering\footnotesize Random \\ \textit{mean $I$\ 638k}}
\end{subfigure}\hfill
\begin{subfigure}[t]{0.48\linewidth}
    \centering
    \includegraphics[width=\linewidth]{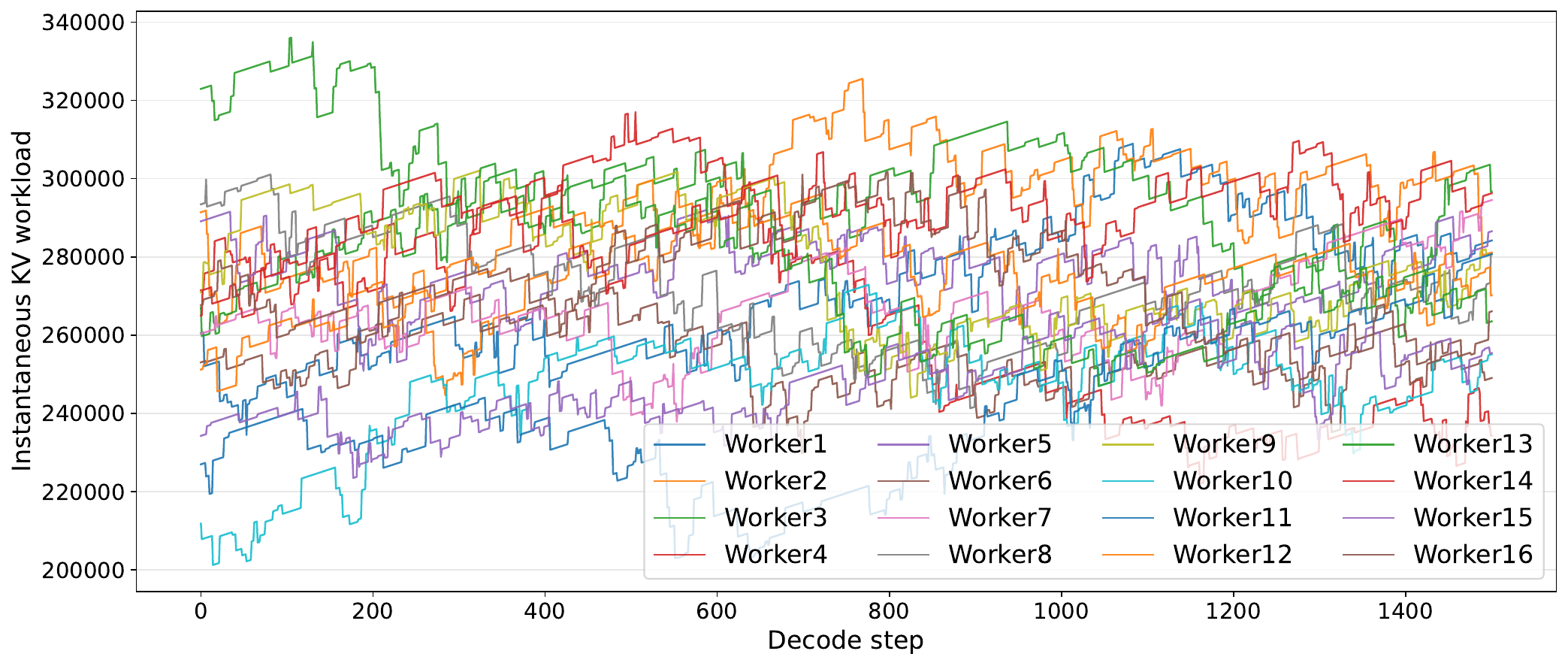}
    \caption*{\centering\footnotesize Round-Robin \\ \textit{mean $I$\ 616k}}
\end{subfigure}\hfill
\begin{subfigure}[t]{0.48\linewidth}
    \centering
    \includegraphics[width=\linewidth]{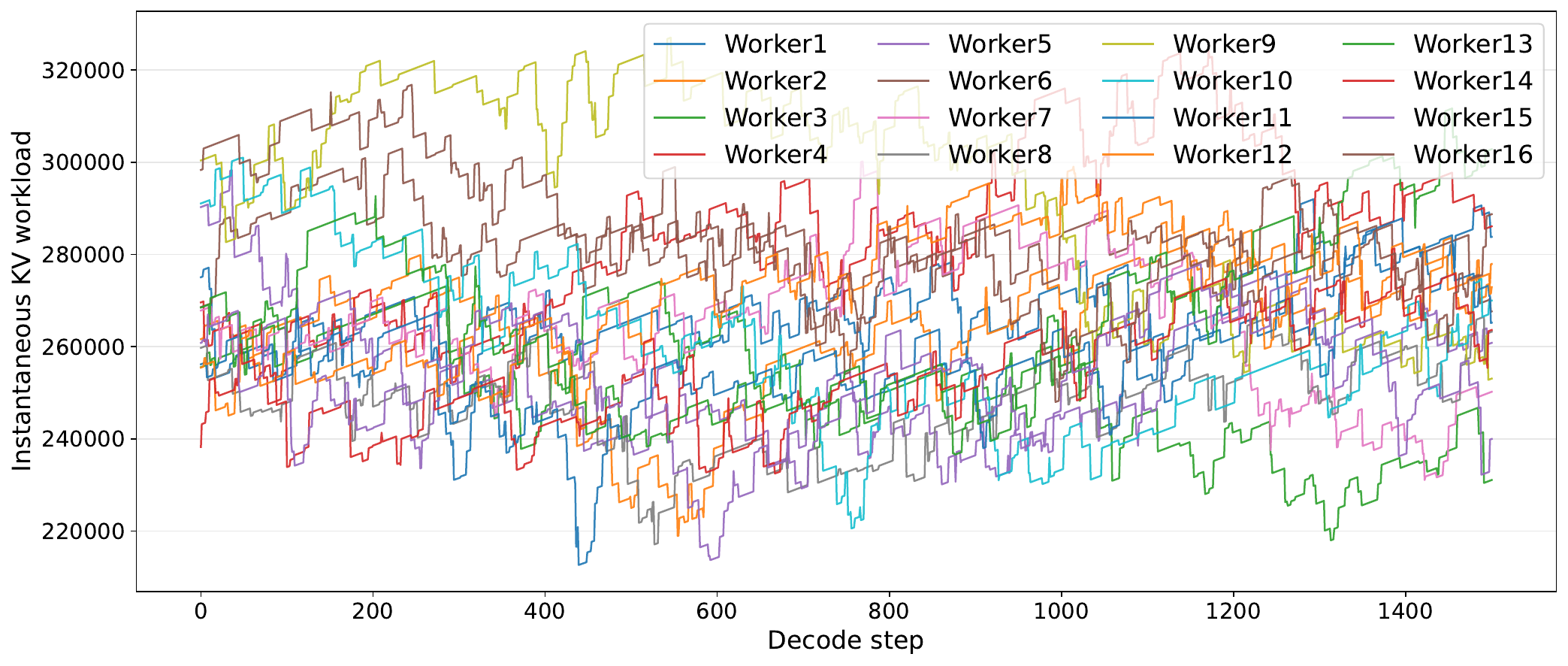}
    \caption*{\centering\footnotesize P2C \\ \textit{mean $I$\ 688k}}
\end{subfigure}\hfill
\begin{subfigure}[t]{0.48\linewidth}
    \centering
    \includegraphics[width=\linewidth]{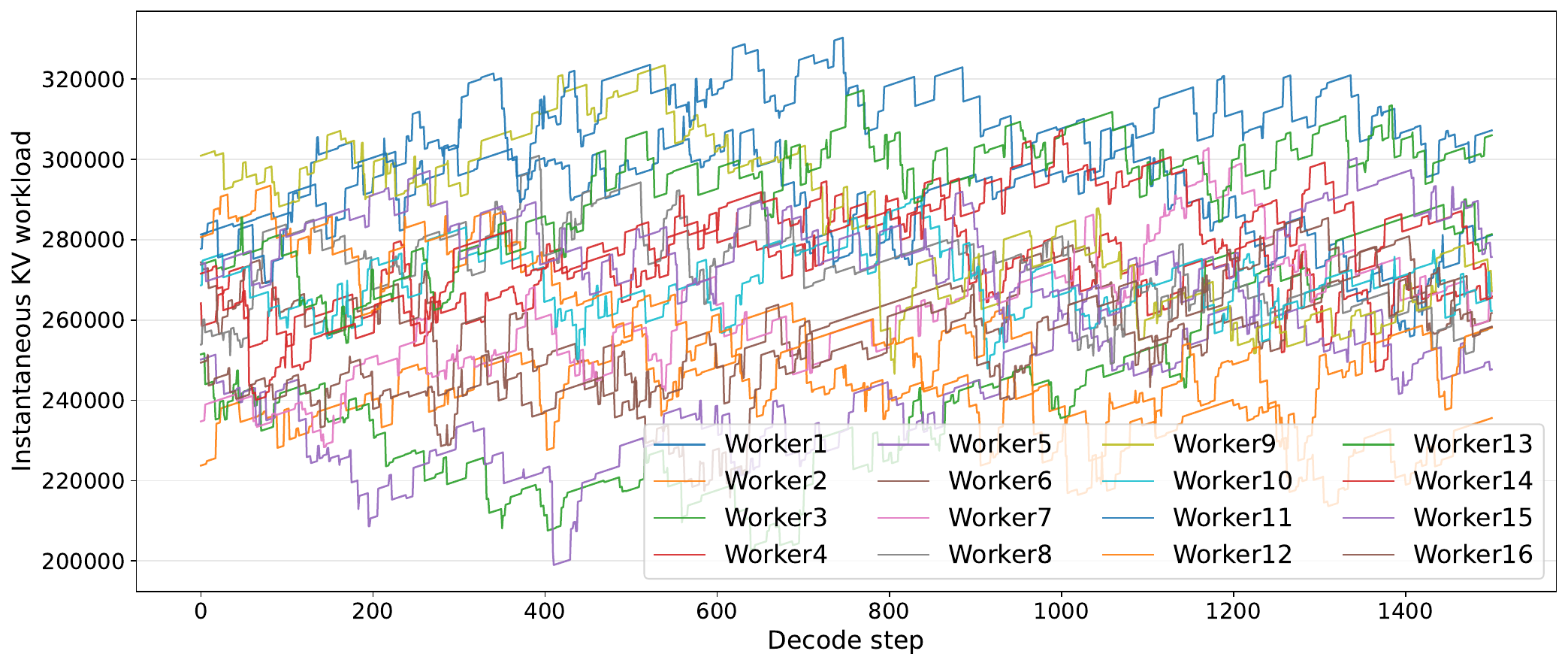}
    \caption*{\centering\footnotesize JSQ \\ \textit{mean $I$\ 676k}}
\end{subfigure}
 
\vspace{0.7em}
 
\begin{subfigure}[t]{0.48\linewidth}
    \centering
    \includegraphics[width=\linewidth]{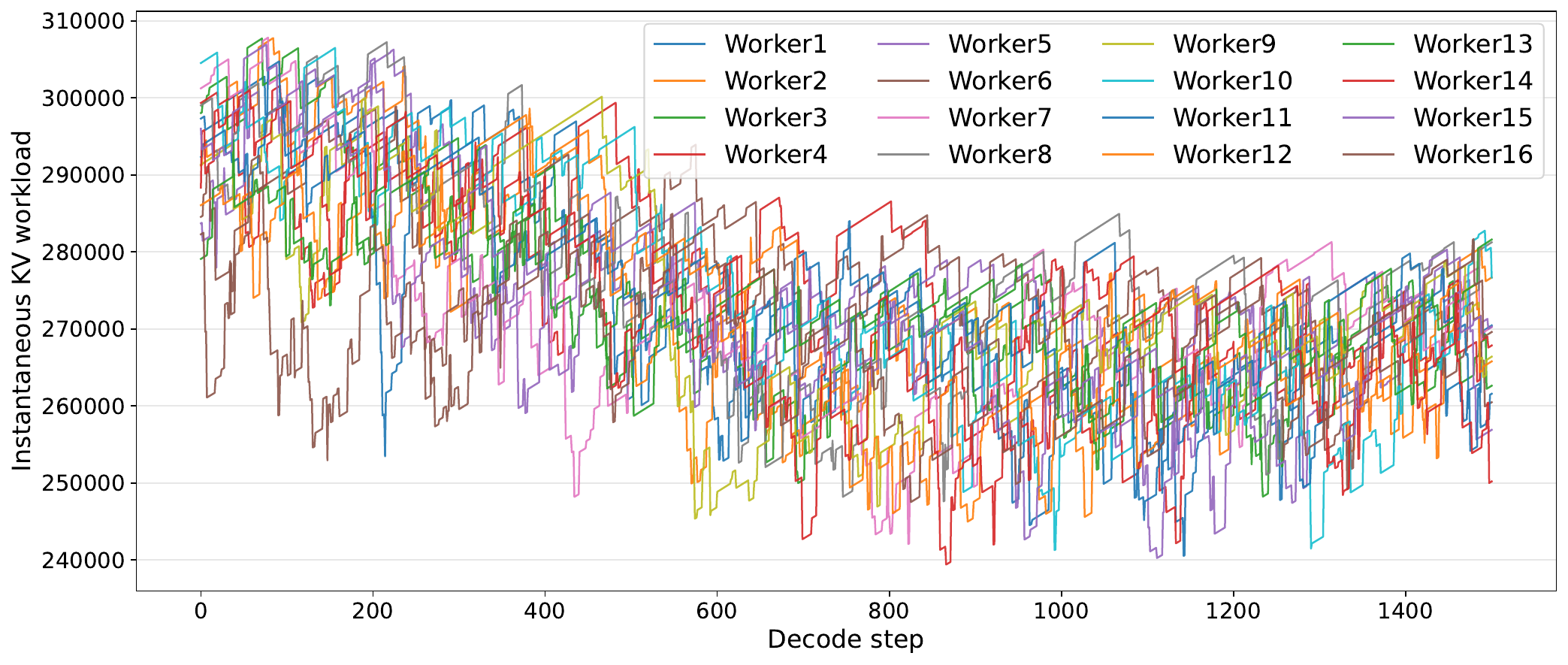}
    \caption*{\centering\footnotesize \textbf{BR-0} (prediction-free) \\ \textit{mean $I$\ 215k}}
\end{subfigure}\hfill
\begin{subfigure}[t]{0.48\linewidth}
    \centering
    \includegraphics[width=\linewidth]{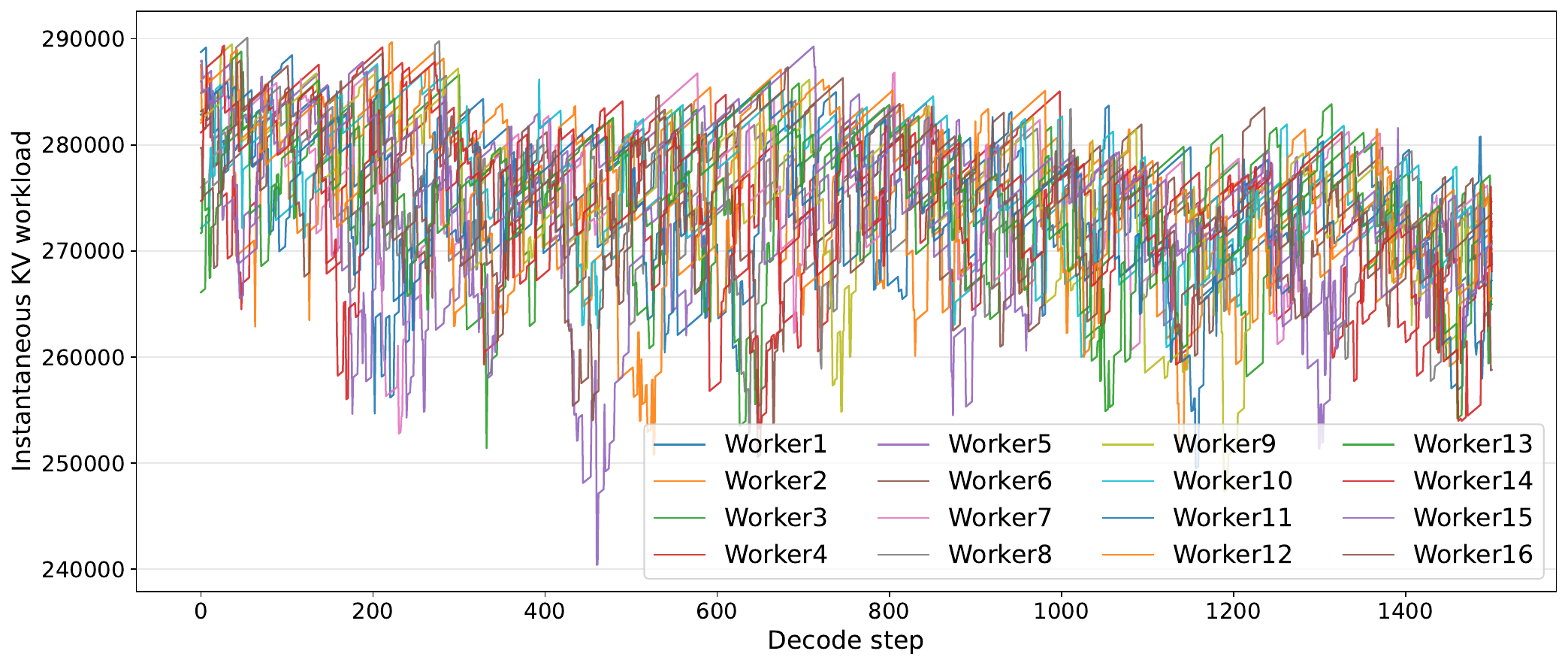}
    \caption*{\centering\footnotesize \textbf{BR-H} oracle \\ $\beta{=}43, \gamma{=}0.86$ \\ \textit{mean $I$\ 124k}}
\end{subfigure}\hfill
\begin{subfigure}[t]{0.48\linewidth}
    \centering
    \includegraphics[width=\linewidth]{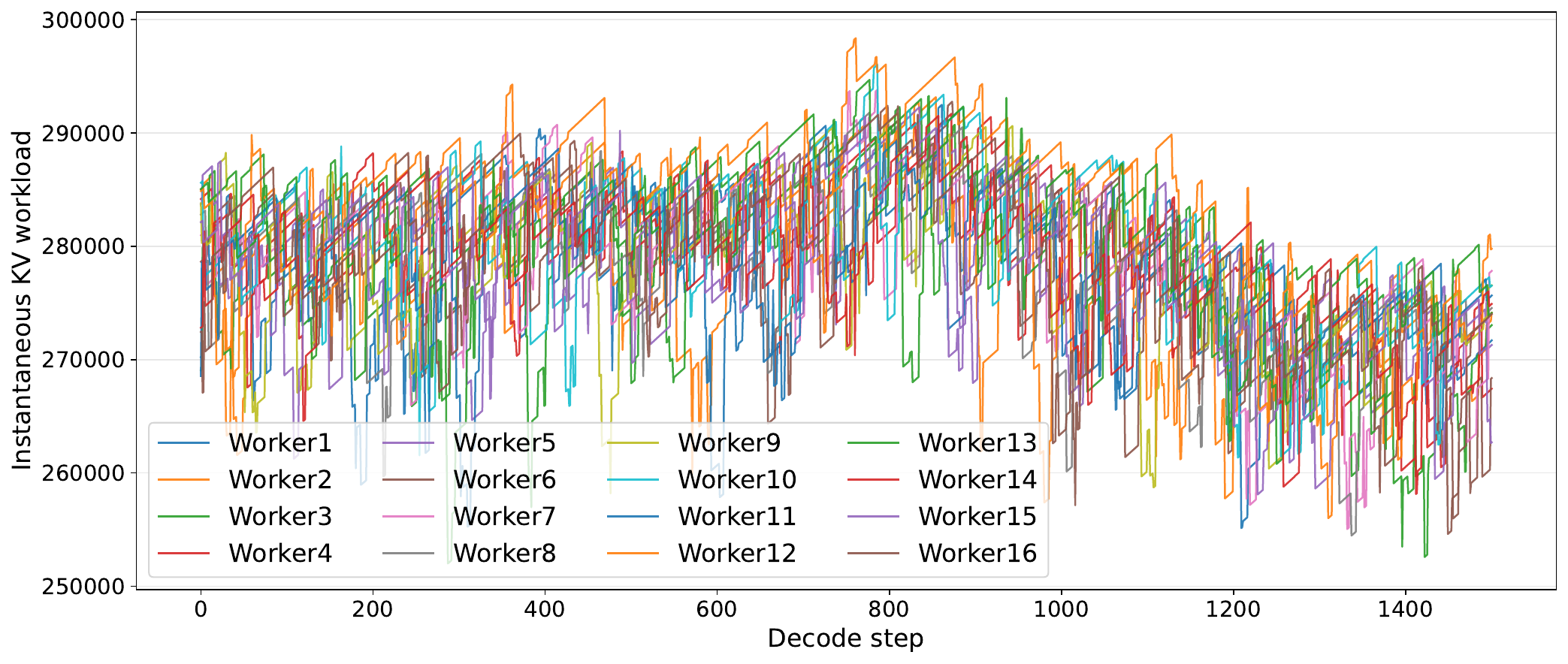}
    \caption*{\centering\footnotesize \textbf{BR-H} oracle \emph{(main)} \\ $\beta{=}14.67, \gamma{=}0.64$ \\ \textit{mean $I$\ 117k}}
\end{subfigure}\hfill
 
\caption{\textbf{Per-worker KV-cache traces on 30002 at $G{=}16$ (5P1D, 144 NPUs).} Layout matches Figure~\ref{fig:scaling-traces-g4}; trace-mean imbalance $I$ under each panel matches the $G{=}16$ column of Table~\ref{tab:scaling_summary}. Both oracle BR-H operating points are shown for direct comparison; $(\beta, \gamma){=}(14.67, 0.64)$ is the configuration plotted in the main-text Figure~\ref{fig:scaling}.}
\label{fig:scaling-traces-g16-main}
\end{figure}

\section{Additional sensitivity results}
\label{app:sensitivity-traces}

The sweep uses oracle BR-H at $H=80$ on 30002. It varies $\beta\in\{1,24,48,96\}$ at $\gamma=0.9$ and $\gamma\in\{0.5,0.7,0.9,1.0\}$ at $\beta=48$. The common setting $(48,0.9)$ appears in both table blocks and both rows of the galleries, giving seven distinct settings rather than eight independent runs. Each worker-load panel shows the same $1{,}500$-step segment used for its cluster configuration; panel axes are independently scaled.

At $G=8$, mean imbalance is approximately $24.7$ to $32.3$ thousand tokens, throughput approximately $945$ to $1{,}042$ tok/s, and TPOT P95 approximately $77.3$ to $79.5$ ms. The baseline imbalances are approximately $215$ to $438$ thousand tokens and throughputs approximately $810$ to $847$ tok/s. Table~\ref{tab:sensitivity_sweep_g8_app} gives the exact values.

\begin{table}[htbp]
\centering
\caption{\textbf{$(\beta, \gamma)$ sensitivity sweep at $G{=}8$ (4P1D, 96 NPUs) on 30002, BR-H oracle, $H{=}80$.} Cross-shaped sweep around $(\beta{=}48, \gamma{=}0.9)$.}
\label{tab:sensitivity_sweep_g8_app}
\begin{tabular}{lccccc}
\toprule
Sweep & $\beta$ & $\gamma$ & Mean $I$ $\downarrow$ & TPOT P95 (ms) $\downarrow$ & Throughput (tok/s) $\uparrow$ \\
\midrule
\emph{$\beta$ at $\gamma{=}0.9$} &  1 & 0.9 & 32{,}304 & 78.120 &  986.41 \\
                                 & 24 & 0.9 & 24{,}706 & 77.933 & 1{,}022.67 \\
                                 & 48 & 0.9 & 25{,}608 & 77.517 & 1{,}042.61 \\
                                 & 96 & 0.9 & 31{,}862 & 79.460 &  945.41 \\
\midrule
\emph{$\gamma$ at $\beta{=}48$}  & 48 & 0.5 & 28{,}327 & 78.040 & 1{,}024.04 \\
                                 & 48 & 0.7 & 26{,}599 & 77.611 & 1{,}020.78 \\
                                 & 48 & 0.9 & 25{,}608 & 77.517 & 1{,}042.61 \\
                                 & 48 & 1.0 & 29{,}318 & 77.310 & 1{,}026.77 \\
\bottomrule
\end{tabular}
\end{table}

\begin{figure}[htbp]
\centering

\begin{subfigure}[t]{0.48\linewidth}
    \centering
    \includegraphics[width=\linewidth]{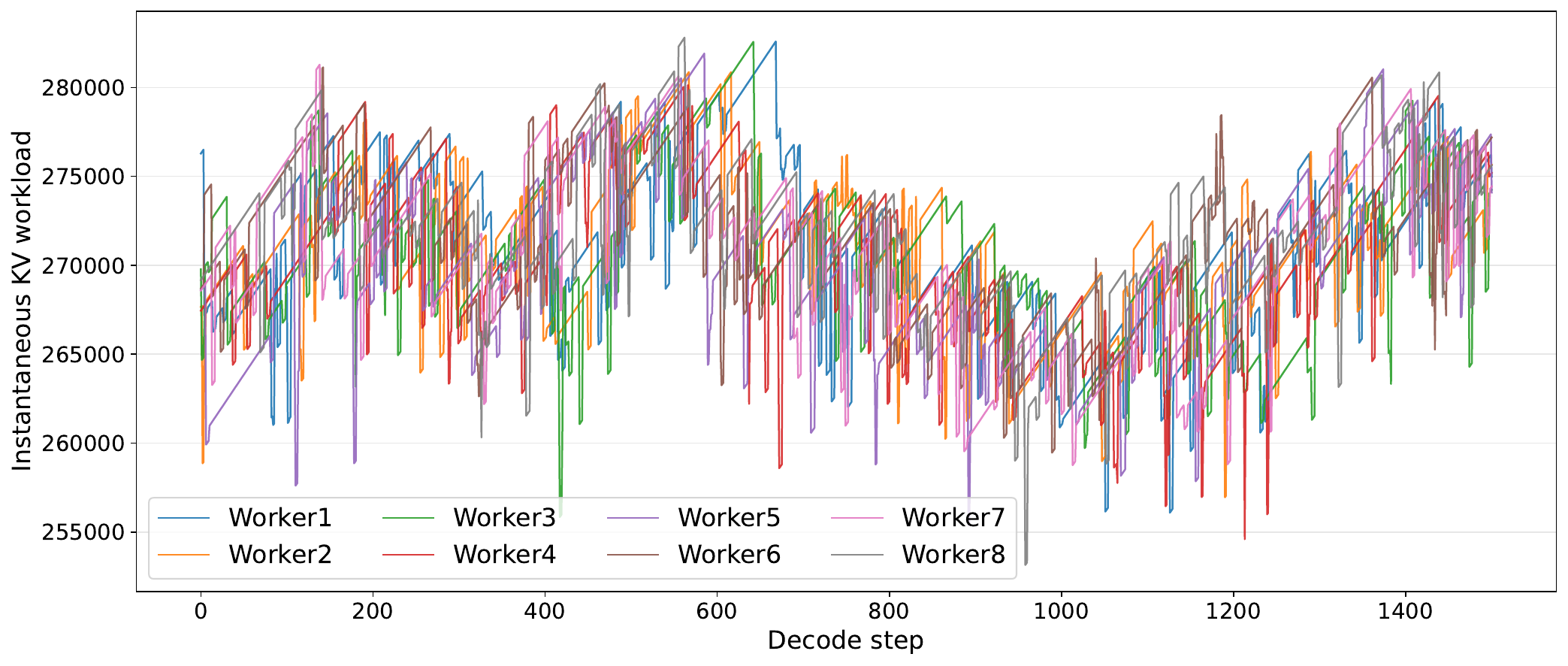}
    \caption*{\centering\footnotesize $\beta{=}1,\,\gamma{=}0.9$ \\ \textit{mean $I$\ 32.3k}}
\end{subfigure}\hfill
\begin{subfigure}[t]{0.48\linewidth}
    \centering
    \includegraphics[width=\linewidth]{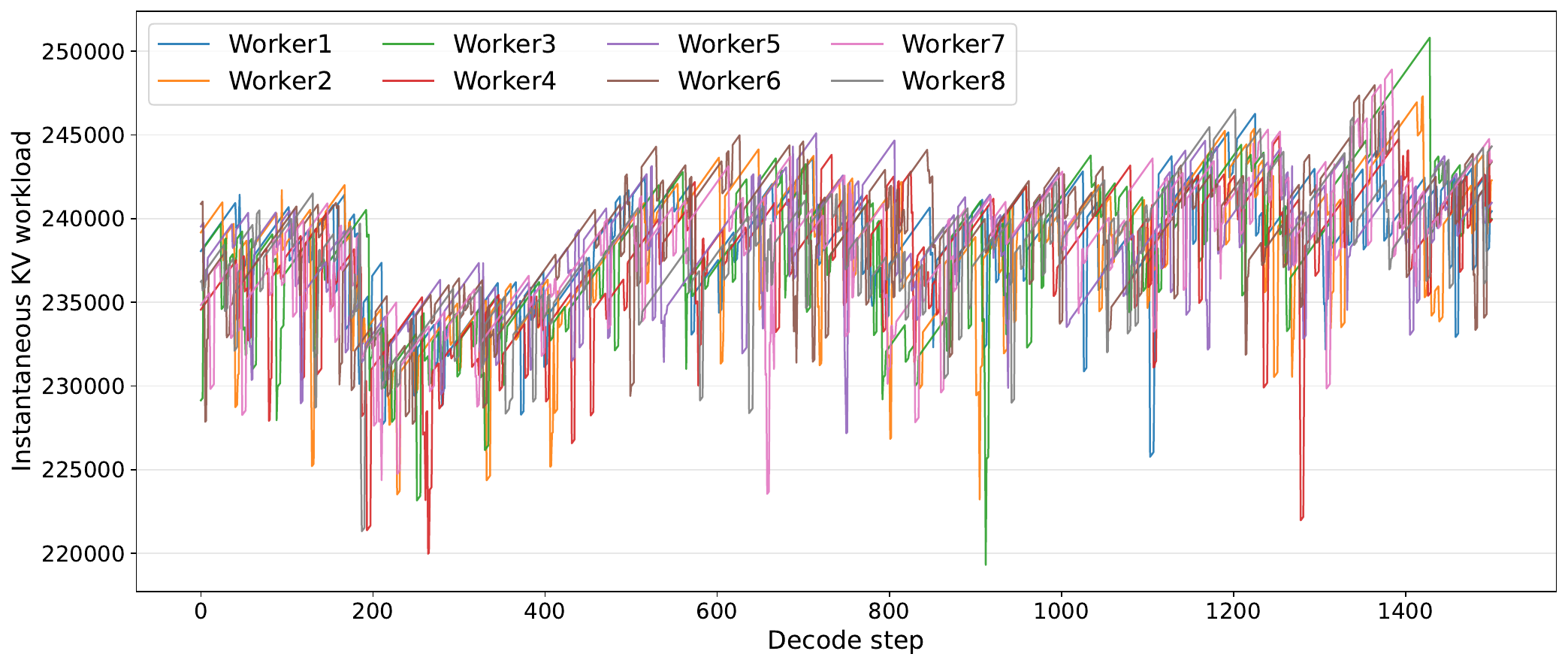}
    \caption*{\centering\footnotesize $\beta{=}24,\,\gamma{=}0.9$ \\ \textit{mean $I$\ 24.7k}}
\end{subfigure}\hfill
\begin{subfigure}[t]{0.48\linewidth}
    \centering
    \includegraphics[width=\linewidth]{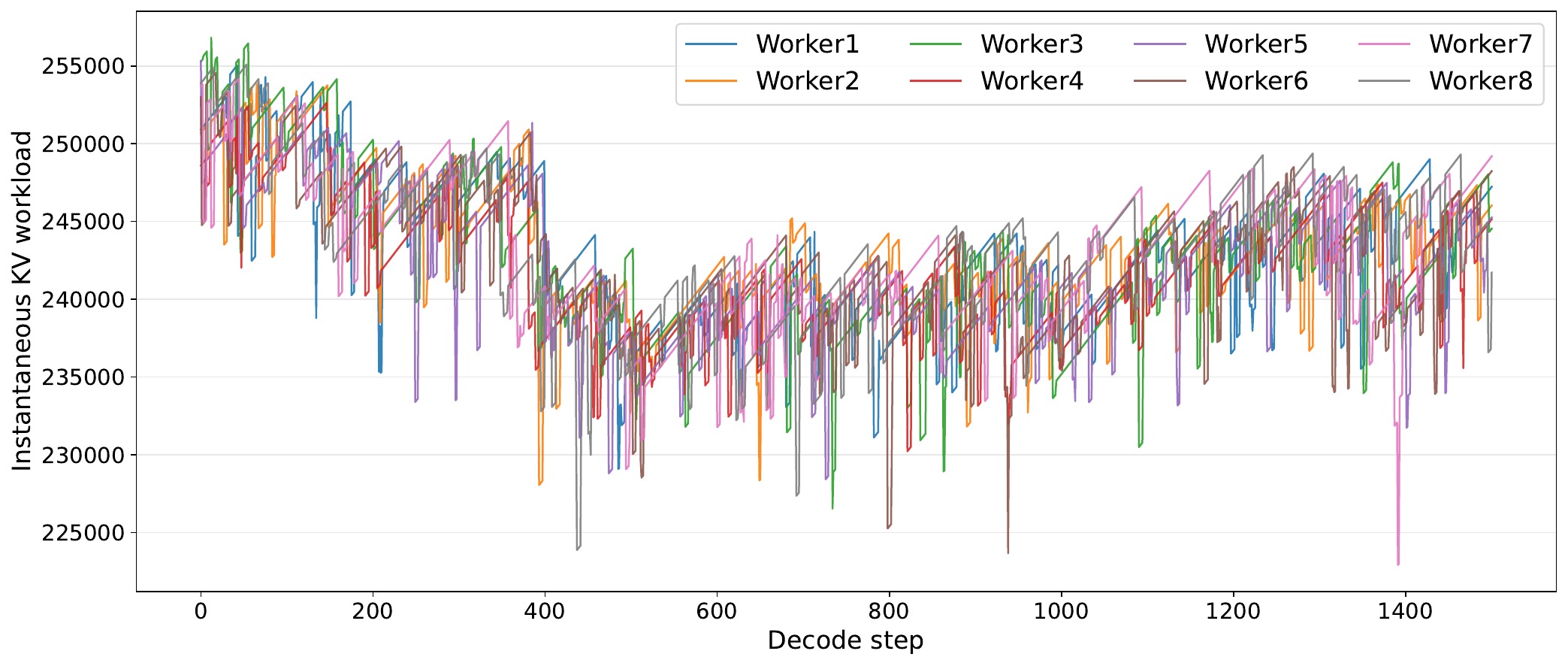}
    \caption*{\centering\footnotesize $\beta{=}48,\,\gamma{=}0.9$ \\ \textit{mean $I$\ 25.6k}}
\end{subfigure}\hfill
\begin{subfigure}[t]{0.48\linewidth}
    \centering
    \includegraphics[width=\linewidth]{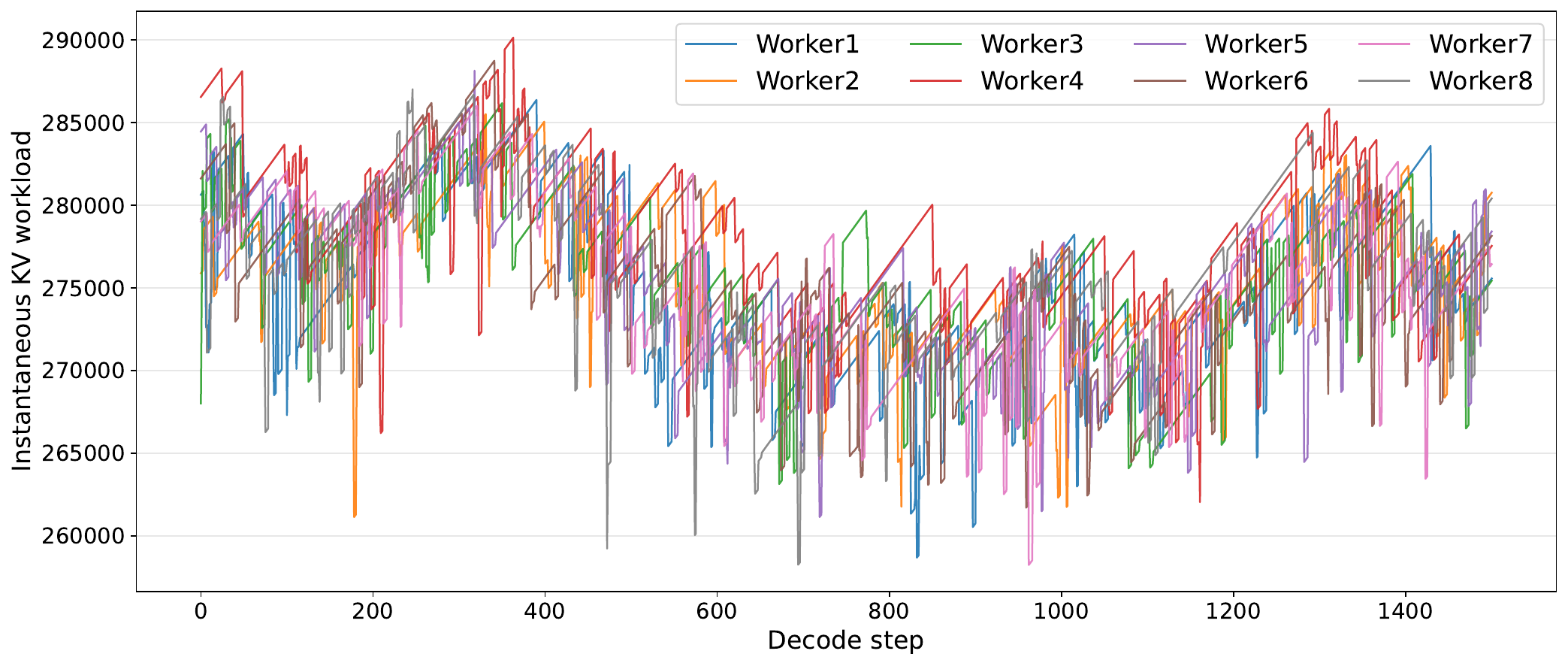}
    \caption*{\centering\footnotesize $\beta{=}96,\,\gamma{=}0.9$ \\ \textit{mean $I$\ 31.9k}}
\end{subfigure}
 
\vspace{0.7em}

\begin{subfigure}[t]{0.48\linewidth}
    \centering
    \includegraphics[width=\linewidth]{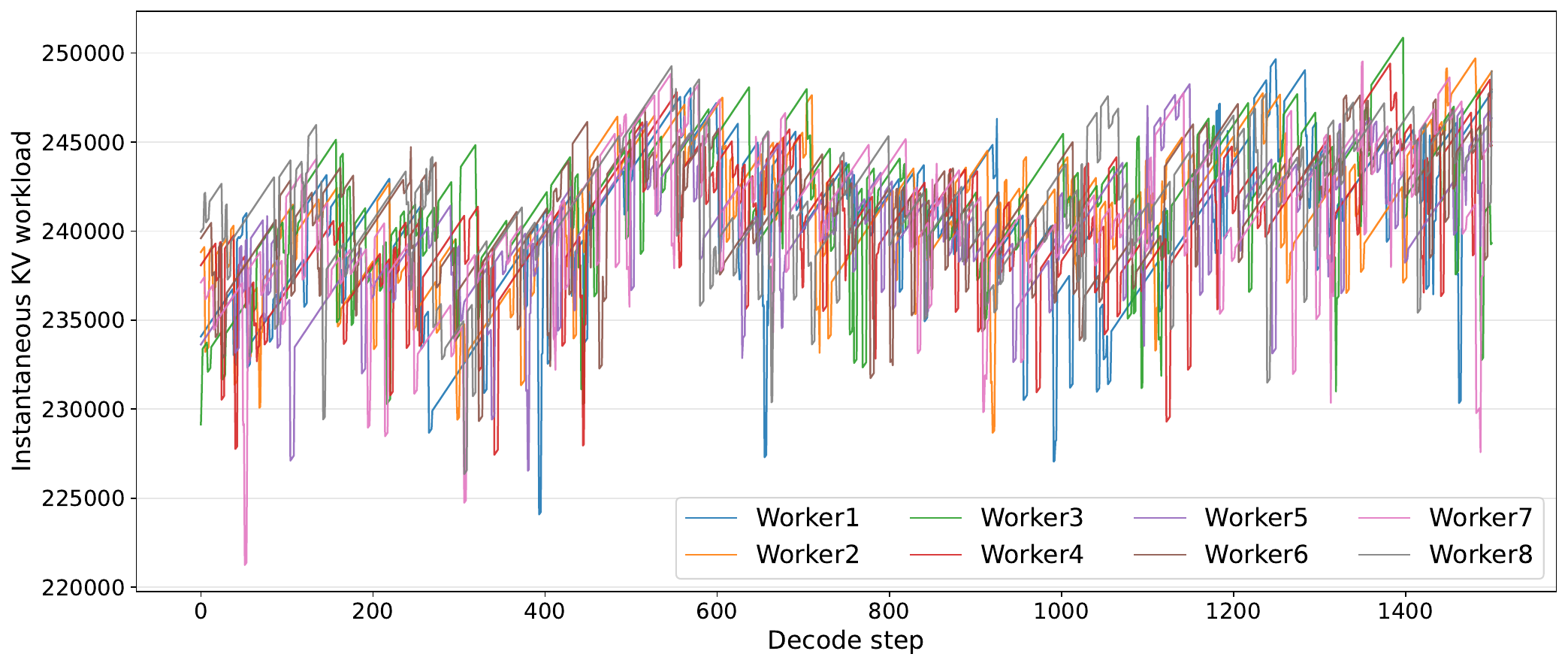}
    \caption*{\centering\footnotesize $\beta{=}48,\,\gamma{=}0.5$ \\ \textit{mean $I$\ 28.3k}}
\end{subfigure}\hfill
\begin{subfigure}[t]{0.48\linewidth}
    \centering
    \includegraphics[width=\linewidth]{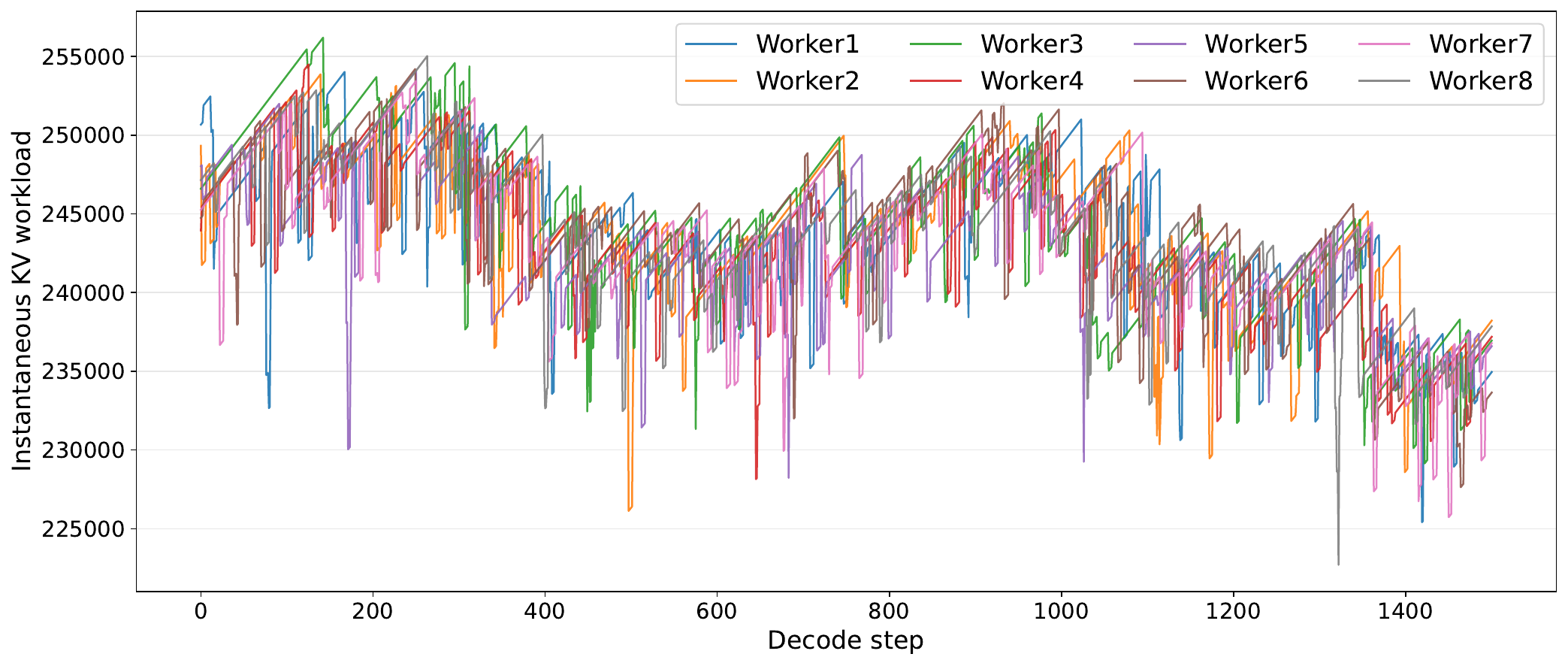}
    \caption*{\centering\footnotesize $\beta{=}48,\,\gamma{=}0.7$ \\ \textit{mean $I$\ 26.6k}}
\end{subfigure}\hfill
\begin{subfigure}[t]{0.48\linewidth}
    \centering
    \includegraphics[width=\linewidth]{plots/Appendix/sensitivity/kv_4p1d_b48_gamma0.9.pdf}
    \caption*{\centering\footnotesize $\beta{=}48,\,\gamma{=}0.9$ \\ \textit{mean $I$\ 25.6k}}
\end{subfigure}\hfill
\begin{subfigure}[t]{0.48\linewidth}
    \centering
    \includegraphics[width=\linewidth]{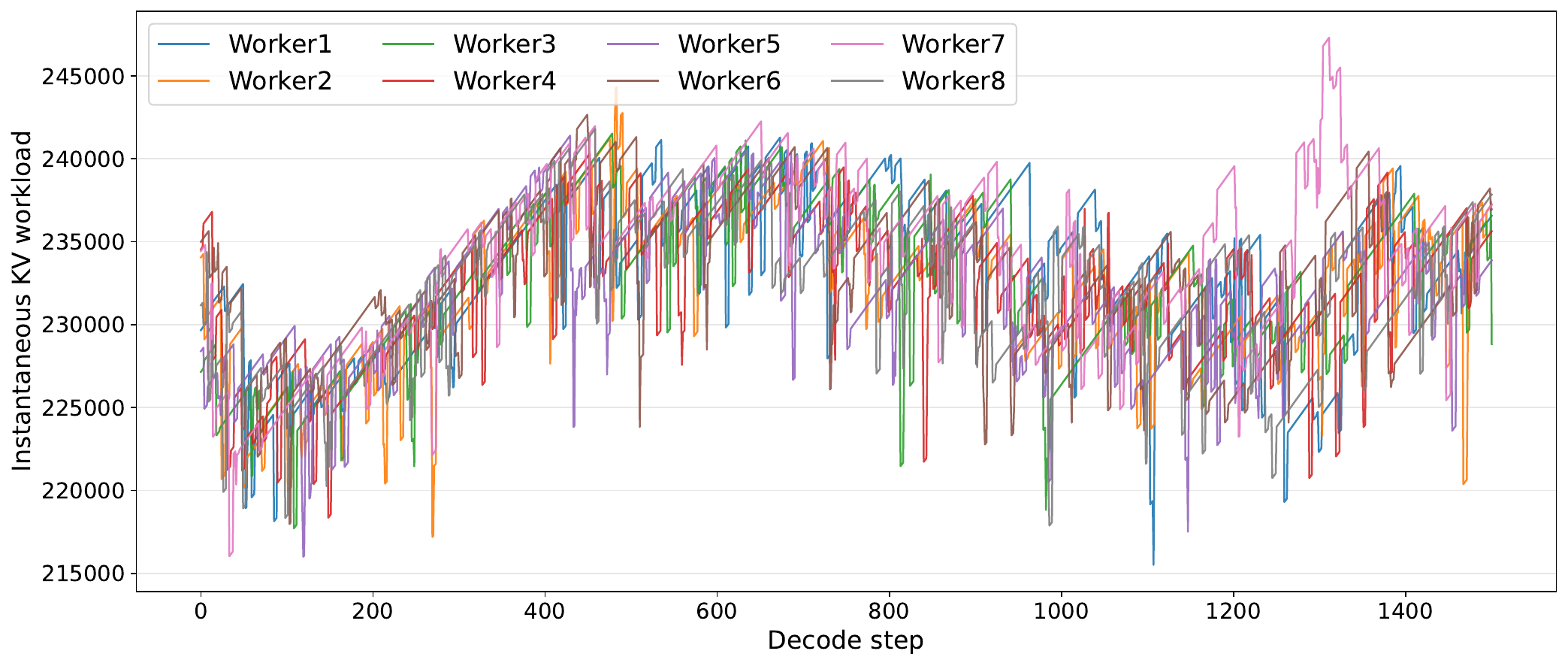}
    \caption*{\centering\footnotesize $\beta{=}48,\,\gamma{=}1.0$ \\ \textit{mean $I$\ 29.3k}}
\end{subfigure}
 
\caption{\textbf{Per-worker KV-cache traces under each swept $(\beta, \gamma)$ on 30002 at $G{=}8$, BR-H oracle, $H{=}80$.} Each panel is annotated with its $(\beta, \gamma)$ and trace-mean imbalance $I$. The configuration $(\beta{=}48, \gamma{=}0.9)$ is shared between the $\beta$- and $\gamma$-sweep blocks and therefore appears in both. Per-panel $y$-axes are independently autoscaled.}
\label{fig:sensitivity-traces-g8}
\end{figure}

Figure~\ref{fig:betagamma} displays the same seven-point sweep. Circles identify the seven swept settings and stars identify two additional operating points.

\begin{figure}[t]
\centering
\includegraphics[width=0.92\linewidth]{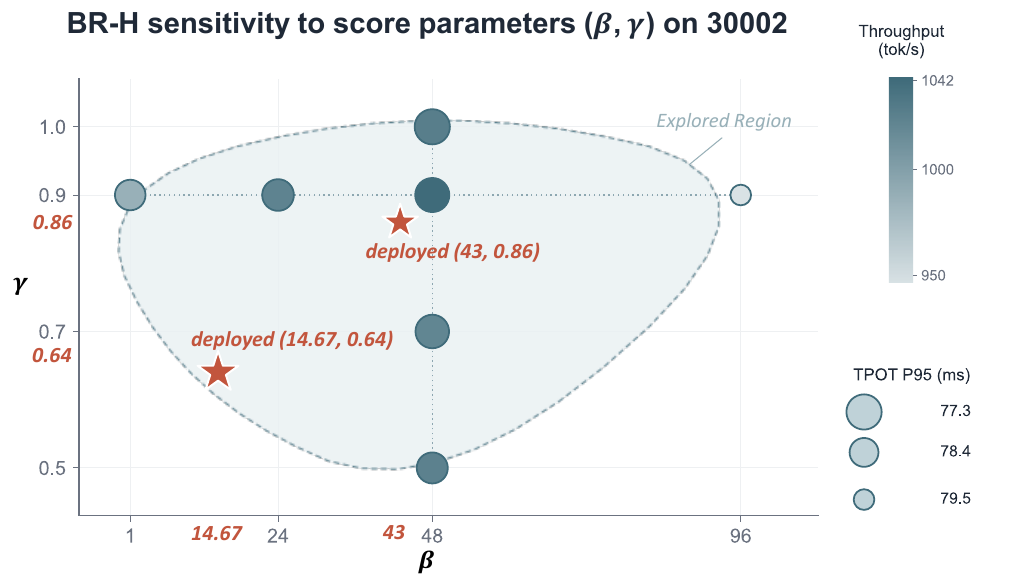}
\caption{\textbf{BR-H sensitivity to $(\beta, \gamma)$ on 30002, oracle, $H{=}80$.} Marker color: throughput (darker $=$ higher). Marker size: $1/\textsc{TPOT P95}$ (larger $=$ faster). Circles: 7 swept configurations. Stars: 2 deployed configurations.}
\label{fig:betagamma}
\end{figure}

At $G=16$, mean imbalance is approximately $104$ to $127$ thousand tokens, throughput approximately $1{,}243$ to $1{,}257$ tok/s, and TPOT P95 approximately $71.1$ to $73.2$ ms. The relative spreads are approximately factors of $1.2$ for imbalance and $1.01$ for throughput. Baseline imbalances are approximately $616$ to $688$ thousand tokens and throughputs approximately $924$ to $927$ tok/s. Across the reported settings and baseline comparisons, the imbalance reduction factors are approximately $4.8$ to $6.6$, and throughput gains are approximately $34\%$ to $36\%$. Table~\ref{tab:sensitivity_sweep_g16_app} gives the exact values.

\begin{table}[htbp]
\centering
\caption{\textbf{$(\beta, \gamma)$ sensitivity sweep at $G{=}16$ (5P1D, 144 NPUs) on 30002, BR-H oracle, $H{=}80$.} Same cross-shaped sweep as Table~\ref{tab:sensitivity_sweep_g8_app}, evaluated at the largest cluster size.}
\label{tab:sensitivity_sweep_g16_app}
\begin{tabular}{lccccc}
\toprule
Sweep & $\beta$ & $\gamma$ & Mean $I$ $\downarrow$ & TPOT P95 (ms) $\downarrow$ & Throughput (tok/s) $\uparrow$ \\
\midrule
\emph{$\beta$ at $\gamma{=}0.9$} &  1 & 0.9 & 104{,}229 & 71.583 & 1{,}256.30 \\
                                 & 24 & 0.9 & 117{,}879 & 72.198 & 1{,}248.51 \\
                                 & 48 & 0.9 & 118{,}396 & 72.610 & 1{,}255.77 \\
                                 & 96 & 0.9 & 112{,}771 & 72.260 & 1{,}256.17 \\
\midrule
\emph{$\gamma$ at $\beta{=}48$}  & 48 & 0.5 & 116{,}074 & 72.227 & 1{,}247.32 \\
                                 & 48 & 0.7 & 109{,}763 & 71.114 & 1{,}256.53 \\
                                 & 48 & 0.9 & 118{,}396 & 72.610 & 1{,}255.77 \\
                                 & 48 & 1.0 & 127{,}207 & 73.167 & 1{,}243.79 \\
\bottomrule
\end{tabular}
\end{table}

\begin{figure}[htbp]
\centering

\begin{subfigure}[t]{0.48\linewidth}
    \centering
    \includegraphics[width=\linewidth]{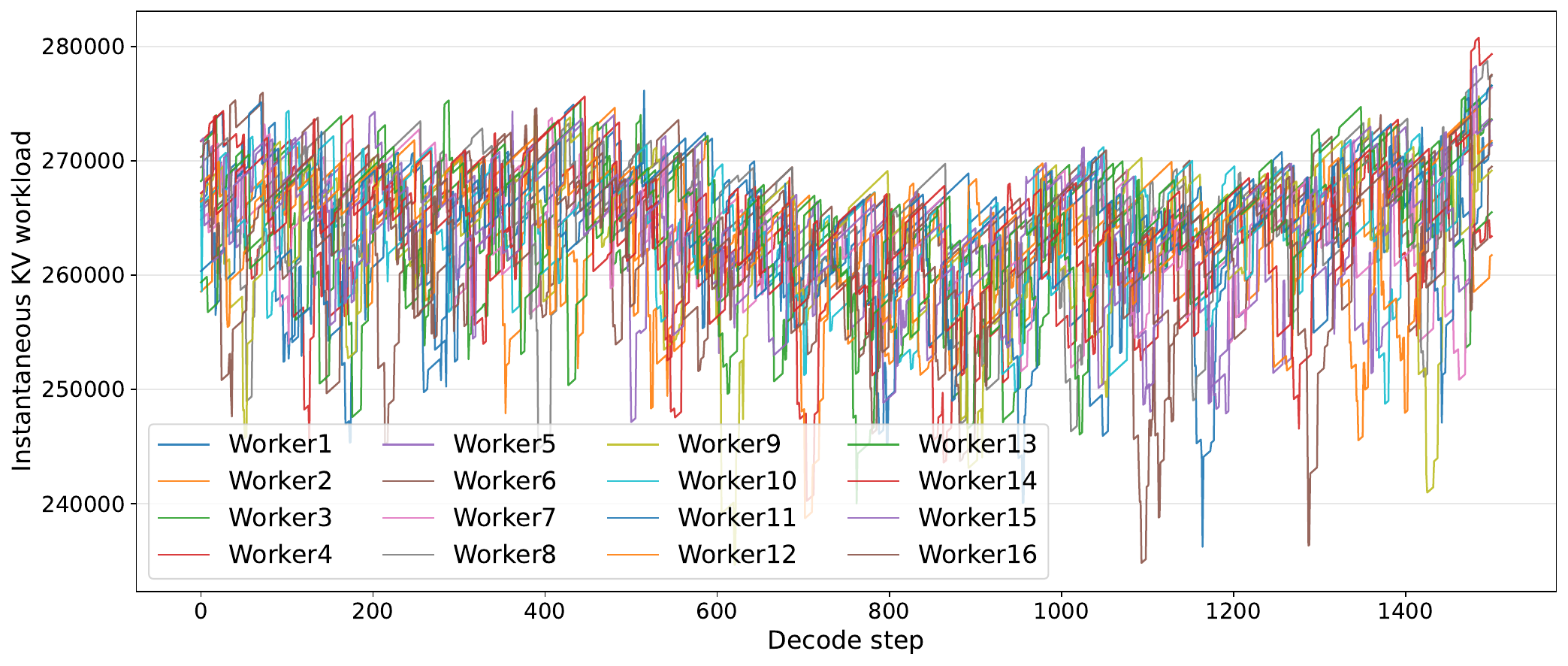}
    \caption*{\centering\footnotesize $\beta{=}1,\,\gamma{=}0.9$ \\ \textit{mean $I$\ 104k}}
\end{subfigure}\hfill
\begin{subfigure}[t]{0.48\linewidth}
    \centering
    \includegraphics[width=\linewidth]{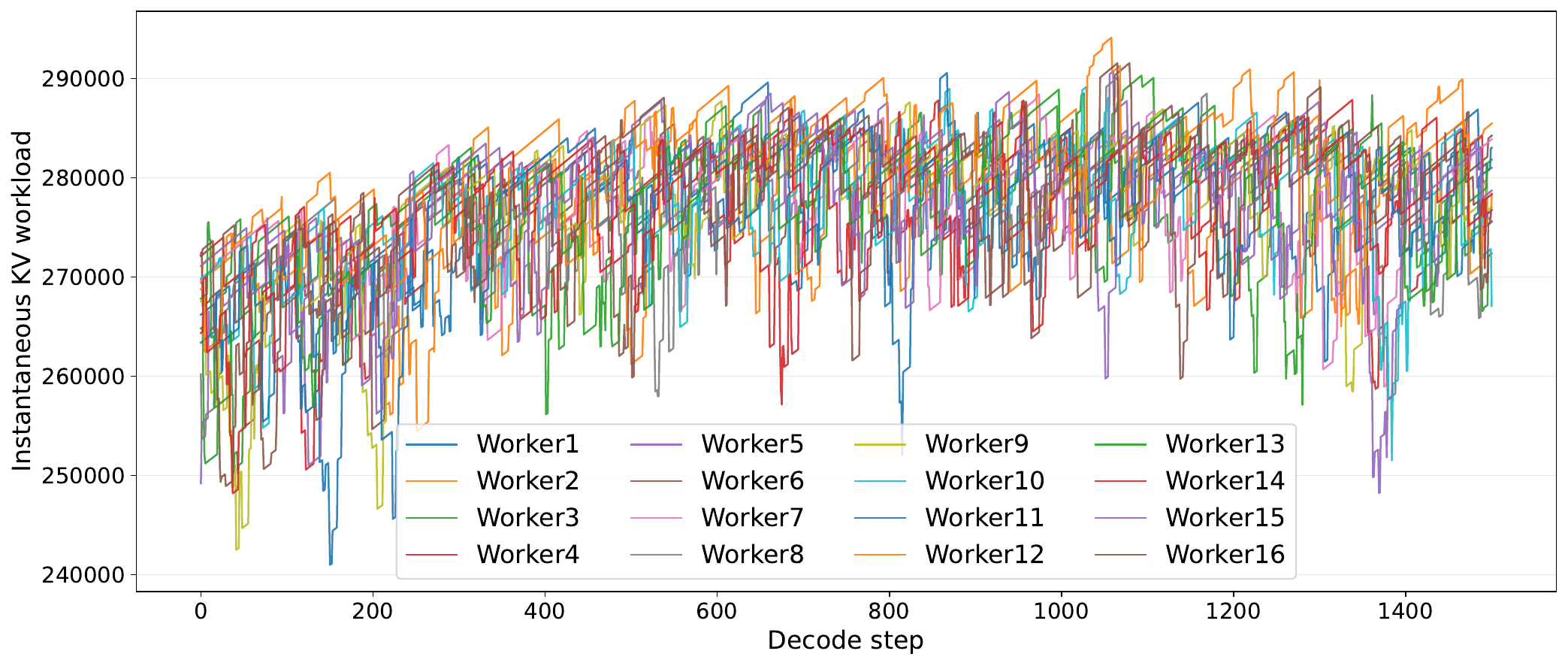}
    \caption*{\centering\footnotesize $\beta{=}24,\,\gamma{=}0.9$ \\ \textit{mean $I$\ 118k}}
\end{subfigure}\hfill
\begin{subfigure}[t]{0.48\linewidth}
    \centering
    \includegraphics[width=\linewidth]{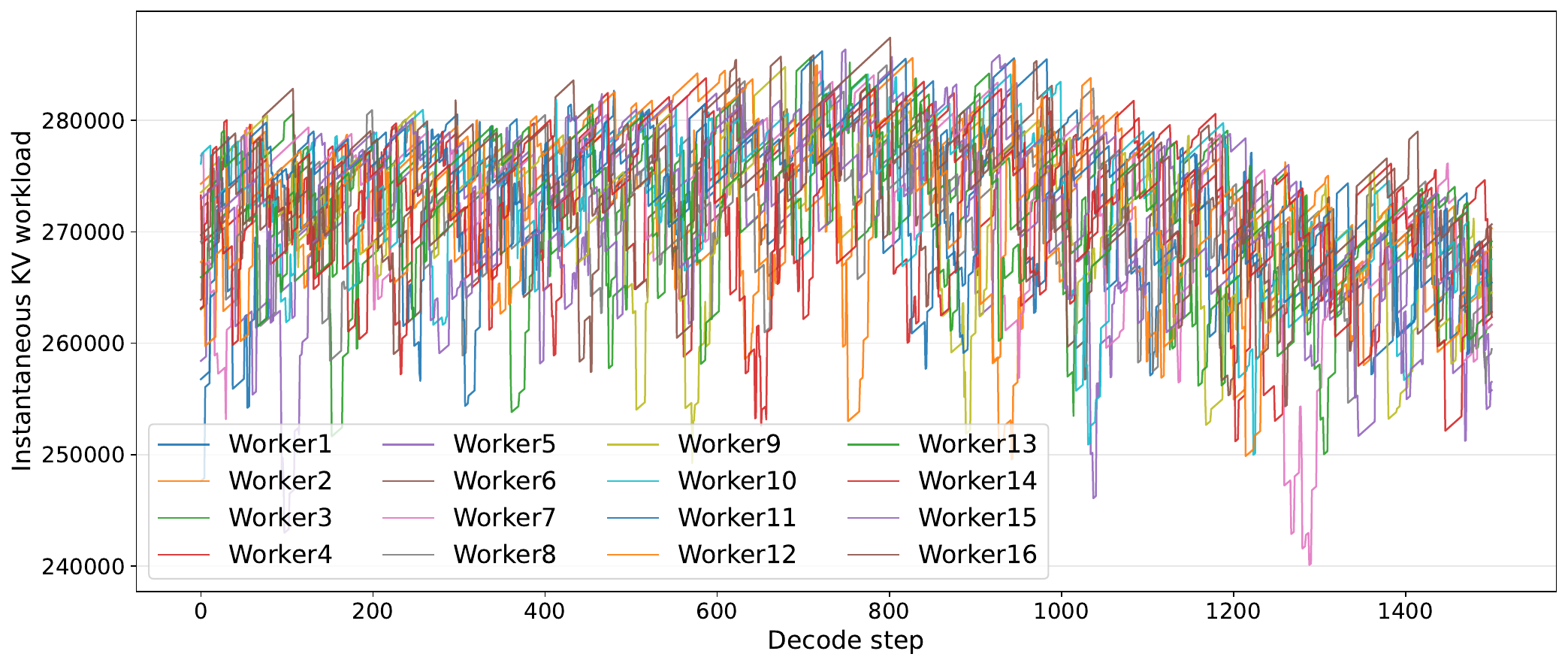}
    \caption*{\centering\footnotesize $\beta{=}48,\,\gamma{=}0.9$ \\ \textit{mean $I$\ 118k}}
\end{subfigure}\hfill
\begin{subfigure}[t]{0.48\linewidth}
    \centering
    \includegraphics[width=\linewidth]{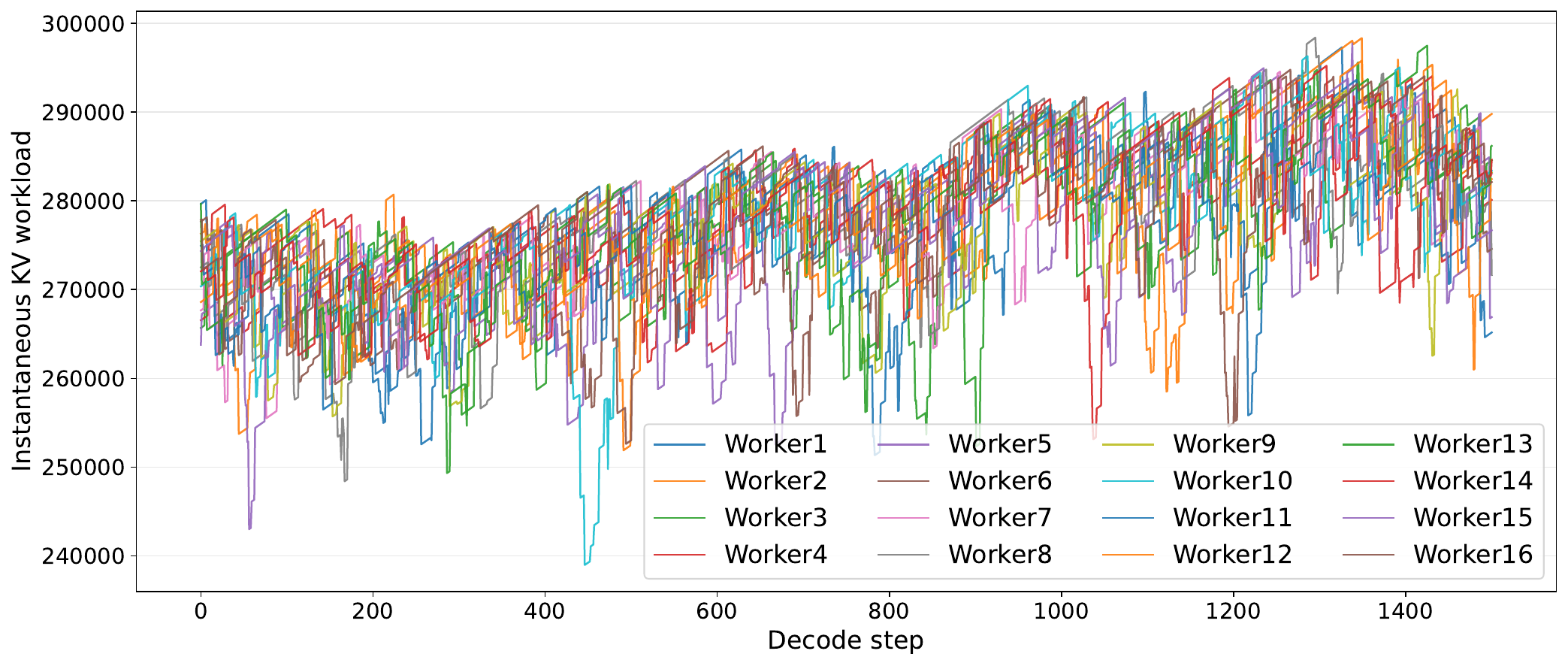}
    \caption*{\centering\footnotesize $\beta{=}96,\,\gamma{=}0.9$ \\ \textit{mean $I$\ 113k}}
\end{subfigure}
 
\vspace{0.7em}

\begin{subfigure}[t]{0.48\linewidth}
    \centering
    \includegraphics[width=\linewidth]{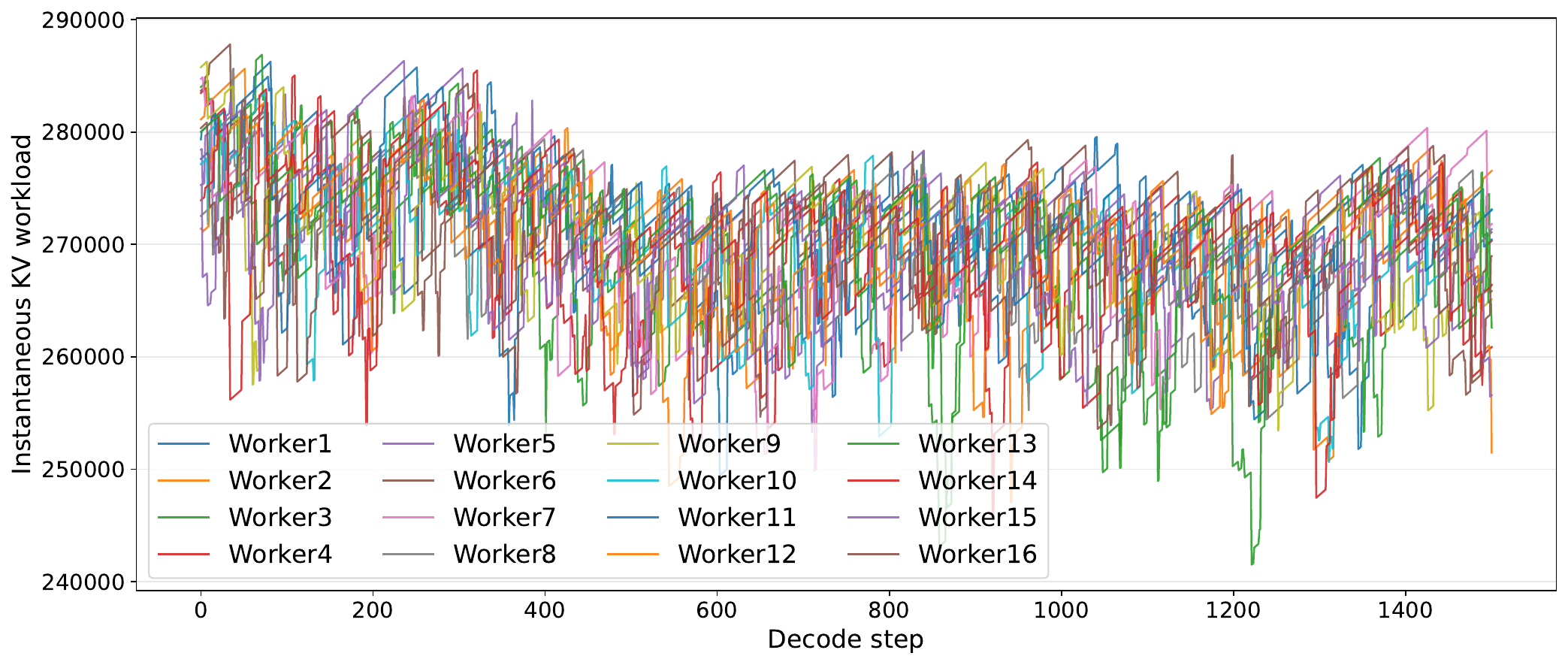}
    \caption*{\centering\footnotesize $\beta{=}48,\,\gamma{=}0.5$ \\ \textit{mean $I$\ 116k}}
\end{subfigure}\hfill
\begin{subfigure}[t]{0.48\linewidth}
    \centering
    \includegraphics[width=\linewidth]{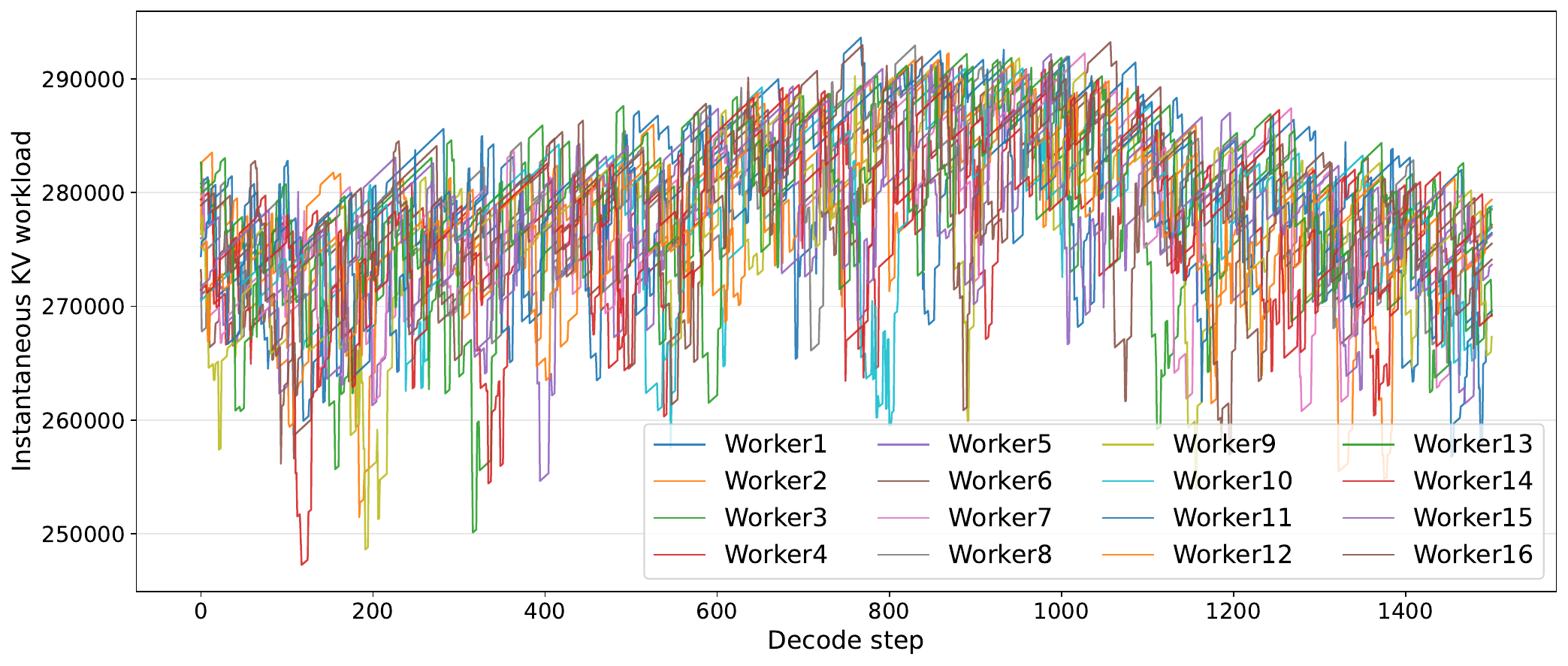}
    \caption*{\centering\footnotesize $\beta{=}48,\,\gamma{=}0.7$ \\ \textit{mean $I$\ 110k}}
\end{subfigure}\hfill
\begin{subfigure}[t]{0.48\linewidth}
    \centering
    \includegraphics[width=\linewidth]{plots/Appendix/sensitivity/kv_16g_b48_gamma0.9.pdf}
    \caption*{\centering\footnotesize $\beta{=}48,\,\gamma{=}0.9$ \\ \textit{mean $I$\ 118k}}
\end{subfigure}\hfill
\begin{subfigure}[t]{0.48\linewidth}
    \centering
\includegraphics[width=\linewidth]{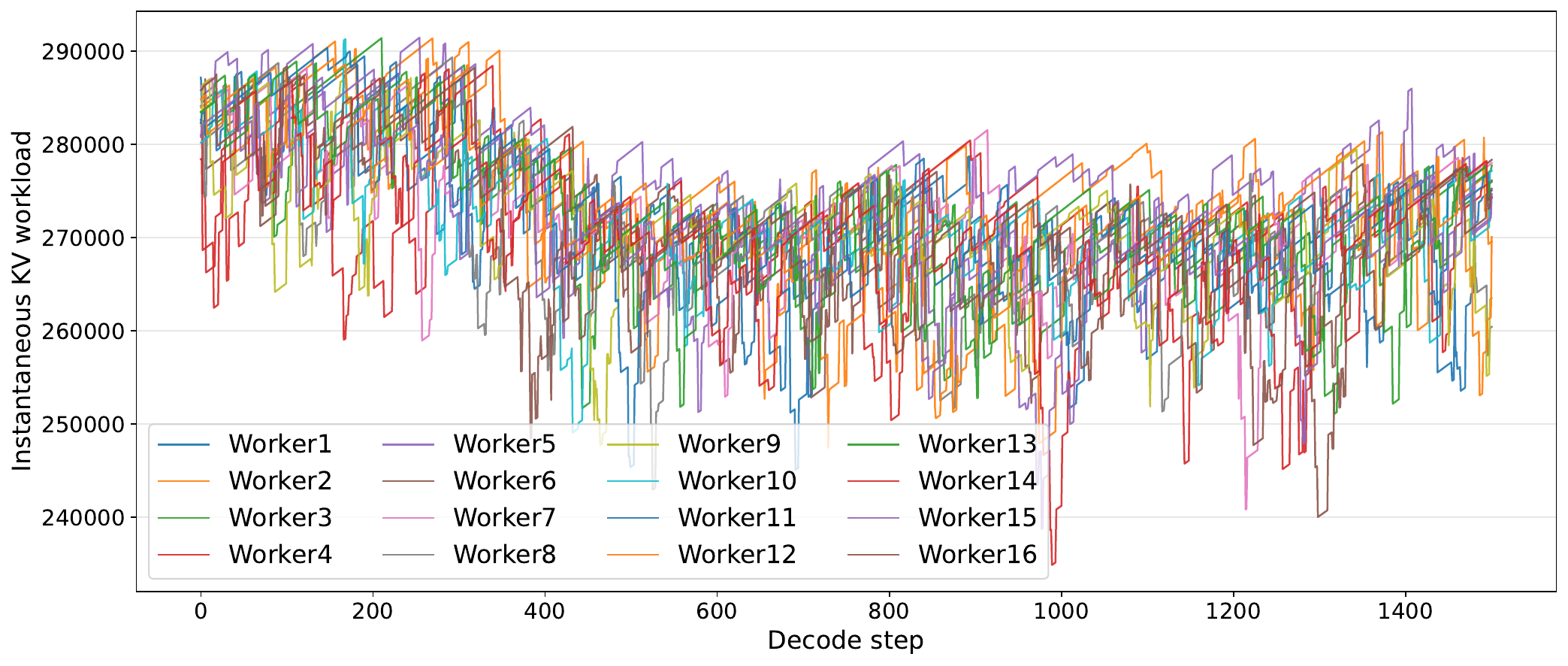}
    \caption*{\centering\footnotesize $\beta{=}48,\,\gamma{=}1.0$ \\ \textit{mean $I$\ 127k}}
\end{subfigure}
 
\caption{\textbf{Per-worker KV-cache traces under each swept $(\beta, \gamma)$ on 30002 at $G{=}16$, BR-H oracle, $H{=}80$.} Same convention as Figure~\ref{fig:sensitivity-traces-g8}; each panel is labelled by $(\beta, \gamma)$ and trace-mean imbalance $I$.}
\label{fig:sensitivity-traces-g16}
\end{figure}

The worker trajectories remain close within each tested parameter setting. Although the absolute bands are wider at $G=16$ than at $G=8$, variation within the parameter sweep is smaller than the differences from the baselines at the corresponding configuration. This supports sensitivity conclusions within the tested settings; it does not establish performance outside this parameter region.

\newpage 

\section{Geometric Pattern of Decode Length in Open-source Dataset}

Figure \ref{fig:pd_dist} visualizes the decode length distribution across several open-source datasets \citep{wang2023openchat,wang2024burstgpt,zheng2023lmsys,zhao2024wildchat}, which supports our geometric decode length assumption in the theory.

\begin{figure}[t]
  \centering
  \includegraphics[width=\columnwidth]{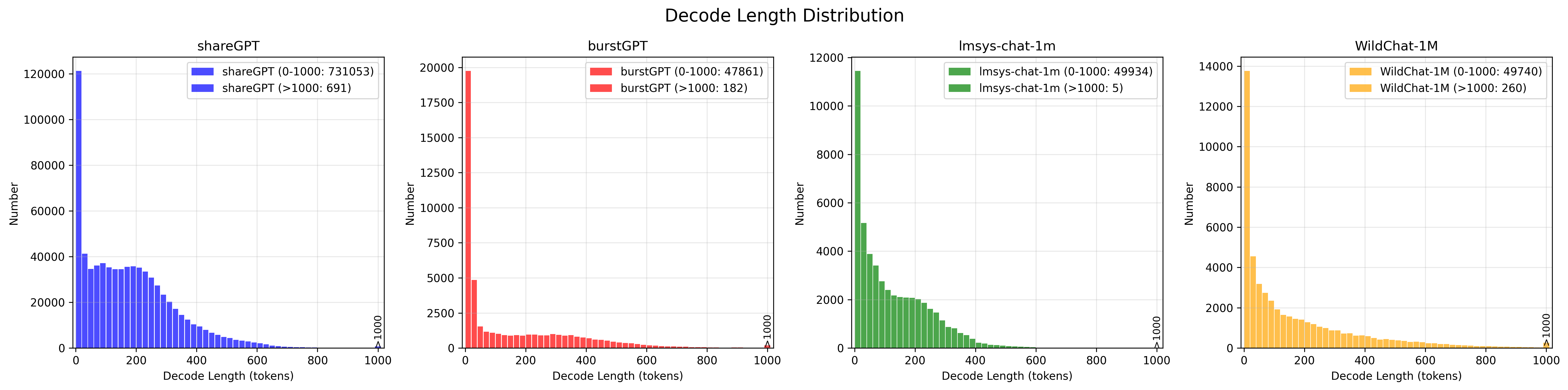}
  \caption{Empirical distributions of decode lengths from production LLM traces \citep{wang2023openchat,wang2024burstgpt,zheng2023lmsys,zhao2024wildchat}. Decode lengths exhibit a geometric (discrete-exponential) pattern.}
  \label{fig:pd_dist}
\end{figure}
\end{document}